\documentclass{article}
\PassOptionsToPackage{
	authoryear,
		sort&compress,
	comma
}{natbib}
\usepackage[
  final,main
]{neurips_2026}

\usepackage{custom}
\usepackage{subfiles}
\usepackage{minitoc}

\renewcommand{\paragraph}[1]{\textbf{#1.}}

\providecommand{\main}{.} 
\graphicspath{{\main/fig/}{fig/}}

\title{Intervention by Polling Strategic Agents}
\author{%
	Chenyu Zhang \\
  \textnormal{MIT}\\
	\texttt{zcysxy@mit.edu}
	\And Rohit Parasnis\thanks{{Work done while at MIT}} \\
  \textnormal{IIT Bombay}\\
	{\texttt{rohit@ee.iitb.ac.in}}
	\And Saurabh Amin \\
  \textnormal{MIT}\\
	\texttt{amins@mit.edu}
}
\date{}

\begin{document}
\renewcommand{\style}{\textstyle}
\maketitle

\begin{abstract}
	A planner in a network of strategic agents faces three entangled challenges: the optimum depends on agents' private information, queried agents may misreport to steer the outcome, and exact computation does not scale.
	We study these challenges in multi-activity network games with heterogeneous private technologies, in which the planner sets non-discriminatory prices.
	We show that the optimal prices admit a centrality-based decomposition of the welfare kernel: each agent's contribution scales with its squared centrality in a network reweighted by agents' preferences across activities.
	This decomposition motivates \poll, a polling algorithm in which the planner samples one agent per round, walks briefly through the agent's neighborhood, and updates the price from a local report.
	From the same decomposition flow three forms of efficiency: computationally, \poll uses significantly fewer operations than exact computation and other distributed methods, requiring up to three orders of magnitude less communication on a real-world network with over 300,000 agents; statistically, its query complexity scales with topology and preference heterogeneity rather than explicitly with population size; and economically, it converges to welfare-maximizing prices while admitting behavior-specific implementations that induce truthful reports and detect adversarial deviations.
\end{abstract}

\section{Introduction}\label{sec:intro}

Consider a government agency that wants to set subsidies for firms in an R\&{D} sector \citep{schot2018Threeframes}, a policymaker that must design electricity tariffs across a network of prosumers \citep{simshauser2016Distributionnetwork}, or a platform that rewards users for different types of engagement \citep{haghtalab2025platforms}. 
In each case a central planner chooses a policy---prices, subsidies, or rewards---to maximize a system-wide objective such as total welfare, which represents the collective well-being of the agents affected by the policy. 
The agents---firms, prosumers, or users---are \textbf{networked}, e.g., a firm's return on R\&{D} depends on the innovation of its collaborators \citep{cohen1989Innovationlearning}.
The agents are also \textbf{strategic}: each responds to the policy in its own interest, and each holds private information, such as its cost of investment and its technological capabilities, that the planner cannot directly observe.
Networked and strategic agents appear across distributed systems, including social learning \citep{acemoglu2011BayesianLearning}, multi-agent control \citep{zhou2024graphon}, and mechanism design \citep{nisan1999Algorithmicmechanism}.
In these settings, three challenges become tightly entangled:

\begin{enumerate}
	\item \textbf{Learning.} The optimal policy depends on every agent's private information and network position, which means that the planner can only evaluate it by querying agents.
	\item \textbf{Strategy.} A queried agent benefits from a policy that favors it and may therefore distort its response to steer the solution toward its own preference.
	\item \textbf{Computation.} As the number of agents grows, bottlenecks arise in time, space, communication, and adaptivity.
\end{enumerate}

The network coupling locks these challenges together: an agent's manipulation strategy depends on both its private information and its network position, and the planner cannot disentangle the two without simultaneously learning from reports and disciplining them. 
To tackle such challenges across multiple applications, we ask a fundamental question:
\begin{myquote}
\textit{\textls[-18]{How can the planner efficiently learn to maximize social welfare by polling strategic agents in a network?}}
\end{myquote}

We study this question using a \textit{multi-activity network game} model, in which strategic agents allocate \textit{effort} across several \textit{activities} and the planner sets \textit{prices} for those activities.
In the R\&D example, an activity corresponds to a research direction, an firm's effort represents its investment level, and the prices are government subsidies.
We extend this model to allow agents to have heterogeneous utility structures, interpreted in the R\&{D} example as private technologies, that determine both the cost of each activity and the complementarities among activities. Thus, two firms in the same network may face entirely different technological tradeoffs. This heterogeneity makes the planner’s problem genuinely network-dependent and distinguishes our model from prior multi-activity games with a common utility structure \citep{chen2018Multipleactivities,kor2023Multiactivityinfluence,parasnis2025incentive}. 

Under this model, we show that the answer to our research question lies in agents' \emph{centralities} in a network reweighted by their \emph{preferences} across activities. This characterization motivates a \emph{polling} process: the planner iteratively samples one agent, walks briefly through its neighborhood, aggregates a local preference for the proposed price change, and revises the price accordingly. Cast in optimization terms, the planner runs a stochastic iterative scheme whose gradient oracle is a strategic agent whose report may depend on its private type, its network position, and its incentive to influence the outcome. The design problem is therefore what to share, what to ask, and how to revise---given that each response may be strategically chosen.

\textbf{Main results.} Our characterization of the planner's optimal intervention (\cref{sec:voting}) gives rise to a distributed learning algorithm, \poll (\cref{sec:poll}), with three types of efficiency:
\begin{itemize}
	\item \textbf{Computational efficiency} (\cref{sec:eff-comp}). \poll uses significantly fewer \textbf{operations} than exact computation or other distributed methods.
	      Lifting any memory burden from the agents, \poll also has a much lower \textbf{space} and \textbf{communication} complexity, since agents need not store local variables or communicate with all neighbors at each iteration to update them.
	      Additionally, \poll has an efficient \textbf{architecture}: it admits fully decentralized, asynchronous, and online implementation, which promises scalability, fault tolerance, and environmental adaptability.
	      On the complete DBLP coauthorship network with $317{,}080$ agents, \poll reaches a near-optimal welfare with $316$--$1572\times$ fewer pairwise exchanges than other distributed methods.
	\item \textbf{Learning efficiency} (\cref{sec:eff-learn}). We show that \emph{homogeneity} in network topology and in activity preferences sharply reduces \poll's \textbf{query} complexity, as each response is more representative of the overall system.
	      Our characterization further yields various natural \textbf{variance reduction} properties that improve learning efficiency in realistic networks.
	      Without modification, \poll is also robust to endogenous and exogenous \textbf{noise}, demonstrating its statistical efficiency.
	\item \textbf{Economic efficiency} (\cref{sec:eff-econ}). \poll converges to the intervention that \textbf{maximizes social welfare} under the planner's constraints.
		Moreover, \poll admits \textbf{strategy-proof} and \textbf{Byzantine-robust} implementations: it induces truthful reports and detects adversarial deviations (see \cref{tbl:sp}).
\end{itemize}

\begin{table}[th]
	\vspace{-1em}
	\caption{Comparison of design principles. Please refer to \cref{sec:eff-comp} for a detailed description of the baseline methods.
		Time is the total operation complexity, equal to the communication complexity up to a $\operatorname{poly}(m)$ factor.
		Space aggregates storage across all agents and the planner.
		All rates are asymptotic as the number of agents $n$, number of activities $m$, and inverse accuracy $\epsilon^{-1}$ grow large.
		A method is adaptive if it accommodates new agents or dynamic game structures without modifying the algorithm or restarting learning, and it is 
		truthful if no agent has any incentive to misreport.
		Robustness is against endogenous noise, exogenous shocks, and adversarial deviations.
	}
	\label{tbl:comp}
	\begin{threeparttable}
		\centering
		\resizebox{\textwidth}{!}{
			\begin{NiceTabular}{llllllll}
				\toprule
				Method
				 & Time
				 & Space
				 & Memory-free
				 & Async.
				 & Adaptive
				 & Truthful
				 & Robust                                                                   \\
				\midrule
				\rowcolor{mygreen!30}\poll\!\!\tnote{\textasteriskcentered}
				 & $m^{2}\ubar{\lambda}^{-2}(1-\delta \sigma_{\max}(DG))^{-2}\epsilon^{-2}$
				 & $m$
				 & \CheckmarkBold
				 & \CheckmarkBold
				 & \CheckmarkBold
				 & \CheckmarkBold
				 & \CheckmarkBold                                                           \\
				Exact
				 & $n^{3}+n^{2}m^{2}+nm^{3}$
				 & NA
				 & NA
				 & NA
				 & NA
				 & NA
				 & \XSolidBrush                                                             \\
				Consensus\tnote{\textdagger}
				 & $nm^{2} \ln \epsilon^{-1}$
				 & $nm^{2}$
				 & \XSolidBrush
				 & \XSolidBrush
				 & \XSolidBrush
				 & \checkwrong
				 & \checkwrong                                                              \\
				Federated\tnote{\textdagger}
				 & $nm^{2} \ln \epsilon^{-1}$
				 & $nm^{2}$
				 & \XSolidBrush
				 & \XSolidBrush
				 & \XSolidBrush
				 & \checkwrong
				 & \checkwrong                                                              \\
				Gossip\tnote{\textdagger}
				 & $nm^{2}(1-\lambda_{2}(D^{-1} G))^{-\frac{1}{2}}\ln \epsilon^{-1}$
				 & $nm^{2}$
				 & \XSolidBrush
				 & \CheckmarkBold
				 & \CheckmarkBold
				 & \checkwrong
				 & \checkwrong                                                              \\
				Random-walk\tnote{\textdagger}
				 & $nm^{2} (1-\lambda_{2}(D^{-1} G))^{-2}  \ln \epsilon^{-1}$
				 & $nm^{2}$
				 & \XSolidBrush
				 & \CheckmarkBold
				 & \CheckmarkBold
				 & \checkwrong
				 & \checkwrong                                                              \\
				\bottomrule
			\end{NiceTabular}
		}
		\begin{tablenotes} \footnotesize
			\item[\textasteriskcentered] See \cref{sec:eff-learn} for parameter definitions.
			\item[\textdagger] Standard formulations of these frameworks assume honest agents and therefore do not retain truthfulness or robustness in the strategic setting we study, while robust variants of each have been developed separately and lie outside this comparison.
		\end{tablenotes}
	\end{threeparttable}
	\vspace{-1em}
\end{table}

These three efficiencies flow from a single source: the \textit{centrality-based} characterization, which inspires the polling process, aligns the price with the polled agent's interest, and exposes an ambiguity agents cannot resolve---information the planner withholds to encourage truth.
The underlying principle extends beyond network games to other settings where the above challenges are entangled.
For example, in \textit{semi-decentralized federated learning},
a central server (planner) must learn a reward to improve the clients' (agents) local models, and the optimal value of this reward depends on hidden, private data geometries; each client is strategic and may misrepresent its data profile to attract a larger reward; and the server's computational constraints can impede constant synchronization with the clients (see \cref{app:fl-motivation} for a detailed discussion).

\subsection{Related work}

\paragraph{Intervention in network games}
\citet{ballester2006WhosWho} connected Katz--Bonacich centrality to Nash equilibria in network games.
\citet{parise2021AnalysisInterventions} surveyed network interventions via variational inequalities.
\citet{galeotti2020TargetingInterventions} characterized optimal price intervention under a quadratic budget, with multi-activity extensions in \citet{chen2018Multipleactivities,kor2023Multiactivityinfluence,parasnis2025incentive}.
These works consider a homogeneous cost structure across agents, while we study how heterogeneous technologies compound network interactions.
\citet{dasaratha2024IncentiveDesign} studied optimal contracts with network spillovers and agent-specific instruments.
All the above assume that the planner knows either the network topology or the agents' types or both, while we consider a learning setup where the planner is unaware of the network topology and agents hold private information.

\paragraph{Strategic distributed learning}
Distributed learning with strategic agents has been studied in collaborative and federated learning \citep{dorner2023Incentivizinghonesty,el-mhamdi2021CollaborativeLearning,procaccia2026IncentivesFederated,donahue2021Modelsharinggames,chakarov2025IncentivizingTruthful,karimireddy2022Mechanismsthat,huang2023Evaluatingincentivizing,haghtalab2025platforms,blum2021Oneone},
algorithmic mechanism design \citep{parkes2003MDPBasedApproach,fallah2024Fairallocation,feigenbaum2007DistributedAlgorithmic}, control and optimization \citep{tanaka2017FaithfulImplementations,zhong2026truthfulmechanism,zhou2017PrivateTruthful,zhang2021MechanismDesign}, and performative prediction \citep{perdomo2020Performativeprediction,mendler-dunner2020Stochasticoptimization,narang2023Multiplayerperformative,piliouras2023MultiagentPerformative},
under a range of behavioral assumptions and strategic variables. Our design draws on ambiguity aversion \citep{gilboa1989Maxminexpected}, incentive design \citep{basar1984AffineIncentive,ho1981Informationstructure}, and Byzantine-robust distributed learning \citep{blanchard2017Machinelearning,karimireddy2023ByzantineRobustLearning,yin2018ByzantineRobustDistributed}.

\section{Problem setup}\label{sec:setup}

Consider a network of $n$ interacting agents ($\X\coloneqq \{1,\dots,n\}$) that engage in $m$ activities ($\A\coloneqq \{1,\dots,m\}$).
Given a price vector $p\in\R_{\ge 0}^m$ and an effort profile $s\in\R^{nm}$, the utility of agent $i$ is
\[\label{eq:util}
	u_{i}(s,p) = \left< s_{i},p \right> - \frac{1}{2} \left< s_{i},W_{i} s_{i}\right> + \sum_{j\ne i} \delta G_{ij} \left< s_{i}, W_{i}s_{j} \right>,
\]
where the first term is the \textit{revenue} from a \textit{non-discriminatory} (public, same for all agents) price $p$ and agent $i$'s effort levels $s_i\in\R^{m}$ across the activities;
the second term is the \textit{cost}, with different agents having distinct cost structures captured by a positive definite matrix $W_{i}\in\R^{m\times m}$, referred to as their private \textit{technologies}, which governs how the agents' activities complement each other and reflects proprietary characteristics that the agent will not disclose (e.g., R\&D capabilities \citep{cohen1990Absorptivecapacity});
the third term is the \textit{network spillover}, whose strength is quantified by $\delta>0$; and $G\in \R_{\ge 0}^{n\times n}$ is the network's adjacency matrix with $G_{ij}>0$ if and only if agent $j$ influences agent $i$,
so that an agent's ability to absorb the spillovers depends on both its technology $W_{i}$ and its network position.
This setup generalizes multi-activity network games \citep{chen2018Multipleactivities,kor2023Multiactivityinfluence} by allowing heterogeneous technologies, which enriches the model and yields new insights for optimal intervention design.

\begin{remark}
	In this work, we focus on \textit{strategic complementarity}:
	one agent’s effort increases the return to its neighbors’ effort.
	This spillover effect is first studied in general settings by \citet{bulow1985multimarket,milgrom1990rationalizability} and later in network games by seminal works such as \citet{ballester2006WhosWho,galeotti2010network,jackson2015games,bramoulle2014StrategicInteraction} for single-activity games and \citet{chen2018Multipleactivities} for multi-activity games.
	Flipping the sign of $\delta$ yields \textit{strategic substitutability}. Mixed-sign spillovers give a more general model.
\end{remark}

\begin{assumption}[Small network effect]\label{assump:net}
	The effect of network interactions is sufficiently small for $\delta\sigma_{\max}(DG) < 1$ to hold, where $D$ is the diagonal degree matrix of the network. \end{assumption}

Assumption \ref{assump:net} is analogous to standard assumptions in the network games literature \citep{ballester2006WhosWho,bramoulle2014StrategicInteraction,parise2019variationalinequality,galeotti2020TargetingInterventions} that ensure a unique Nash equilibrium $s^*(p)$  while guaranteeing that best-response dynamics are oscillation-free. Substituting this equilibrium into the social welfare $\frac{1}{n}\sum_i u_i$ yields the planner's intervention problem
\[\label{eq:planner}
	\max_{p\in\R_{\ge 0}^m} \; \frac{1}{2}\langle p, \Phi p \rangle - \frac{\mu}{2}\langle (p - p_0), \Psi(p - p_0) \rangle,
\]
where $\Phi\in\R^{m\times m}$ is the \emph{social welfare kernel}, defined as
\[\label{eq:sw-kernel}
	\Phi \coloneqq \frac{1}{n}(\mathbf{1}_n^T\otimes I_m) W^{-T}((I_n - \delta G^T)^{-1}\otimes I_m)\, W\, ((I_n - \delta G)^{-1}\otimes I_m) W^{-1}(\mathbf{1}_n\otimes I_m),
\]
where $W \coloneqq\operatorname{blkdiag}(W_{i})\in\R^{mn \times mn}$ is the block diagonal matrix with $W_{i}$ as the $i$-th diagonal block, and $\otimes$ is the Kronecker product.
The second term is the planner's cost of intervention, with its strength quantified by $\mu>0$, $p_0$ being the initial price, and $\Psi\in\R^{m\times m}$ encoding budget constraints, sustainability considerations, or prior beliefs (see \cref{sec:intervention} for detailed modeling motivation).

\begin{assumption}[Sufficient regularization]\label{assump:reg}
	We have $\mu\Psi \succ \Phi$ so that the objective function \cref{eq:planner} is concave and the optimal intervention both exists and is unique.
\end{assumption}

Under \cref{assump:reg}, the first-order condition gives the optimal intervention
\[\label{eq:opt-p}
	p_* = (I - \mu^{-1}\Psi^{-1}\Phi)^{-1} p_0.
\]

\section{A centrality-based interpretation}\label{sec:voting}

The optimal intervention is shaped by the social welfare kernel $\Phi$, which aggregates agents' heterogeneous technologies through the network (see \cref{eq:sw-kernel}) and thus is unknown to the planner.
This section characterizes the structure of $\Phi$, which motivates the design of our learning algorithm and underpins its efficiency guarantees. Related calculations are deferred to \cref{sec:calc}.

\begin{proposition}\label{prop:cen-planner}
	Define \(\tilde{W} _{i} = \Psi^{1/2}W_{i}\Psi^{1/2}\) and \(\tilde{\Phi} \coloneqq \Psi ^{-1 /2}\Phi \Psi ^{-1 /2}\).
	We have
	\[\label{eq:cen-planner}
		\tilde{\Phi} = \frac{1}{n}\sum_{i \in \X} \tilde{\cen}_{i}^T \tilde{W}_{i}^{-1} \tilde{\cen}_{i}
	\]
	where \(\widetilde{\cen}_{i}\in\R^{m\times m}\) is the \(i\)-th block\footnote{$\tilde{\cen}_{i}\coloneqq (e_{i}\otimes I_{m})^{T}\tilde{\cen}$ where $e_{i}\in\R^{n}$ is the $i$-th standard basis vector.\label{fn:block}} of  \emph{Katz--Bonacich centralities} of the modified network \(\tilde{G}\in\mathbb R^{nm\times nm}\):
	\[
		\tilde{\cen}= (I - \delta\tilde{G})^{-1} (\mathbf{1}_{n}\otimes I_{m})\in \R^{(n\times 1)(m \times m)}, \quad \text{where}\quad
		\tilde{G}_{ij} = G_{ij} \tilde{W} _{i}\tilde{W} _{j}^{-1}\in\R^{m\times m}
	\]
	is the $ij$-th block of \(\tilde{G}\in\R^{(n\times n)(m\times m)}\), i.e., the modified network has matrix-valued edge weights.
\end{proposition}

To elucidate \cref{prop:cen-planner}, we make the following simplification, only in this section: $\Psi=I_{m}$ (so that $\tilde{\Phi}=\Phi$ and $\tilde{W}_{i}=W_{i}$) and agents' cost structures share the same eigenbasis, i.e., there exist an orthogonal matrix $U\in \R^{m\times m}$ and diagonal matrices $\Lambda_i\in \R^{m\times m}$ such that $W_{i} = U \Lambda_{i} U^{T}$ holds for all $i\in\X$.
We show in \cref{apx:cen} that $\Phi$ and $\tilde{\cen}_{i}$ also share the eigenbasis $U$. Then, \cref{eq:cen-planner} reads
\(
\phi_{a} = \frac{1}{n}\sum_{i\in\X} \tilde{\cen}_{i,a}^{2} \cdot \lambda_{i,a}^{-1},   \text{ with }
\tilde{\cen}_{a} = (I - \delta\tilde{G}_{a})^{-1} \bm{1}_{n},  \text{ and }
\tilde{G}_{a,ij} = G_{ij} \lambda_{i,a}/\lambda_{j,a},  \forall a\in [m]
,\)
where $\phi_{a},\tilde{\cen}_{i,a},\lambda_{i,a}$ are the $a$-th eigenvalues of $\Phi,\tilde{\cen}_{i},W_{i}$, respectively; $\tilde{\cen}_{i,a}$ is also the $i$-th entry of $\tilde{\cen}_{a}\in\R^{n}$.
Each $a\in[m]\coloneqq \{ 1,\dots,m \}$ corresponds to a column of $U$, which specifies a linear combination of activities, referred to as an \emph{activity bundle}.
We also observe two comparative statics: the optimal intervention for activity bundle $a$ increases in $\phi_{a}$ and an agent's equilibrium effort in $a$ increases in $\lambda_{i,a}^{-1}$. We therefore refer to $\lambda_{i,a}^{-1}$ as agent $i$'s \emph{preference} for activity bundle $a$.

\cref{prop:cen-planner} now becomes clear: each activity bundle $a$ induces a network $\tilde{G}_{a}$ that reweighs the original network by the ratios of the interacting agents' preferences for the bundle. The
agents interact within this \textit{bundle-induced} network to form an \textit{aggregate preference}, weighted by their squared centralities in the bundle-induced network, which determines the optimal intervention.

Bundle-induced networks assign centralities in a seemingly counter-intuitive way. First, it empowers the \emph{influenced, not the influential}: an agent's preference carries more weight when the agent receives more influence from others (high susceptibility in $G$).
Instead of targeting specific agents, the welfare-maximizing planner targets activity bundles via a non-discriminatory price, and therefore listens to the most susceptible agents---those that \emph{absorb} the most influence. The bundle-induced network thus acts as an \emph{echo chamber} that amplifies the intervention's utility gains.
Second, it \emph{favors the unfavored}: an agent unwilling to engage in the bundle (low preference $\lambda_{i,a}^{-1}$) has greater say in the bundle-induced network. This is because agents absorb network spillovers through their own technologies (see \cref{eq:util}), so a high cost for a bundle implies proportionally higher compensation from neighbors' efforts---clear in the R\&D context, where a higher investment yields a higher return from an innovative community.
This centrality structure is contextualized in \cref{sec:exp}.

\subsection{Centrality elasticity}\label{sec:ce}

It is now clear that agents' preferences are only half the story, as the optimal intervention also depends on their positions in the network.
Consequently, without considering the network structure, one cannot ascertain whether a low-preference agent benefits more or less from intervention than a high-preference one.
The following proposition captures how the sensitivity of an agent's equilibrium utility to its preference is affected by network structure via \textit{centrality elasticity}.

\begin{proposition}\label{prop:ce}
	The equilibrium utility of an agent $i\in [n]$ at price $p$ is $u_{i}(p)\coloneqq u_{i}(s^*(p),p) = \frac{1}{2}\left< p, \Phi_{i}p \right>$, where the eigenvalues of $\Phi_i$ are given by $\phi_{i,a} = \tilde{\cen}_{i,a}^{2}\lambda_{i,a}^{-1}$. Thus,
	$u_{i}$ is monotonically increasing w.r.t.\ $\phi_{i,a}$ for all $a\in\A$.
	Furthermore, we have
	\begingroup\small\vspace{-0.5em}
	\[
		\ds\frac{\partial \ln \phi_{i,a}}{\partial \ln \lambda_{i,a}}
		= 2\underbrace{\frac{\partial \ln \tilde{\cen}_{i,a}}{\partial \ln \lambda_{i,a}}\vphantom{\frac{\tilde{K}_{a,ii}}{\sum_{i\in \X} \tilde{K}_{a,ij}}}}_{\text{centrality elasticity}} - 1
		= \underbrace{\frac{\sum_{j\in \X\setminus \{ i \}}\tilde{K}_{a,ij}}{\sum_{j\in \X} \tilde{K}_{a,ij}}}_{\text{others' contribution}} - \underbrace{\frac{\tilde{K}_{a,ii}}{\sum_{j\in \X} \tilde{K}_{a,ij}}}_{\text{self-contribution}}
		,\]
	\vspace{-0.5em}\endgroup
	where $\tilde{K}_{a}\coloneqq (I - \delta \tilde{G}_{a})^{-1}$ is the   centrality matrix of the interaction network induced by $a$.
\end{proposition}

Centrality elasticity (CE) is determined by the shares of the agent's own and others' contributions to its centrality.
Specifically, agent $i$'s centrality counts the walks originating at $i$, and the self-contribution (resp.\ others' contribution) is the fraction of these walks ending at $i$ (resp.\ at other nodes).

Consistent with the previous subsection, an agent that absorbs more influence from others has a high CE and thus benefits more from spillovers as it further reduces its preference for the bundle.
In the R\&D example, where $\lambda_{i,a}$ is the investment cost, CE quantifies the marginal return on investment.
A firm in a high-tech cluster may want to invest more in R\&D to harness the innovative environment. Without such an environment, however, the investment may not justify its cost.

Importantly, the decision to invest, i.e., the sign of $\partial \ln \phi_{i,a} / \partial \ln \lambda_{i,a}$, flips precisely when the self-contribution and others' contribution are equal.
An agent ignorant of the network structure cannot discern these shares and thus remains uncertain whether increasing or decreasing its preference will improve its utility.
This fundamental \textit{ambiguity} is distinctive to the network setting and serves as a crucial lever for the planner to mitigate strategic manipulation in \cref{sec:eff-econ}.

\section{A polling algorithm} \label{sec:poll}

We introduce \poll, a distributed algorithm for computing the optimal intervention \cref{eq:opt-p}, motivated by the characterization in \cref{eq:cen-planner}.
Our design is driven by a need for computational efficiency: computing \cref{eq:opt-p}  exactly requires $O(n^{3} + n^{2}m^{2} + nm^{3})$ operations, which is prohibitive when the agent population or activity space is large.
Throughout, we let $\nu\coloneqq\mu^{-1}$ and $q \coloneqq \Psi^{1/2}p$.

\paragraph{Step 1. Iterative communication}
Equation \cref{eq:opt-p} directly gives the following iterative scheme:
\[\label{eq:iter}\let\displaystyle\textstyle
	\underbrace{\sum_{\tau=0}^{t+1} \nu ^{\tau} \tilde{\Phi}^{\tau} q_{0}}_{q_{t+1}} = q_0 + \nu \tilde{\Phi} \underbrace{\sum_{\tau=0}^{t} \nu ^{\tau} \tilde{\Phi}^{\tau} q_{0}}_{q_{t}}
	\overset{t\to \infty}{\longrightarrow} \sum_{\tau=0}^{\infty} \nu ^{\tau} \tilde{\Phi}^{\tau} q_{0} = q_*
	.\]
Under \cref{assump:reg}, the iteration is contractive and converges to the optimal solution $q_*$, with a computational complexity of $O(n^{3}+ tm^{2})$ after $t$ iterations.

To address the $O(n^{3})$ computational burden of calculating $\tilde{\Phi}$, suppose at each time step $t$ we have an easy-to-compute estimate $\hat{\Phi}_t$ of $\tilde{\Phi}$. Then, \cref{eq:iter} is associated with the stochastic approximation scheme
\[\label{eq:poll-1}
	q_{t+1} = q_{t} - \alpha _{t}((I - \nu \hat{\Phi}_{t})q_{t} - q_0),\]
were $\alpha_t$ is a step size that satisfies the condition in \citet{robbins1951stochasticapproximation}. The convergence rate of \cref{eq:poll-1} depends on the bias and variance of the stochastic estimate $\hat{\Phi}_t$.

Since the intervention problem is deterministic, we introduce a random variable $Z_t$ and a mapping $\hat{\Phi}_t=\hat{\Phi}(Z_t)$ such that $\mathbb{E}_{Z_t}\hat{\Phi}(Z_t) = \tilde{\Phi}$.
The semantic characterization of $\tilde{\Phi}$ in \cref{eq:cen-planner} provides an intuitive and practical way to construct and sample $Z_t$ in a distributed and asynchronous manner.

\paragraph{Step 2. Sample an agent}
Equation \cref{eq:cen-planner} gives
\[
	\tilde{\Phi} = \frac{1}{n}\sum_{i\in\X} \widetilde{\cen}_i^T \tilde{W}^{-1} _{i} \widetilde{\cen}_i = \mathbb{E}_{i \sim \operatorname{Unif}(\X)} [ \widetilde{\cen}_i^T \tilde{W}^{-1} _{i} \widetilde{\cen}_i ].
\]

\paragraph{Step 3. Sample a walk length $k$}
Expanding the centrality  of agent $i$ as a Neumann series yields
\[
	\widetilde{\cen}_{i} = \sum_{k=0}^{\infty} \delta ^{k} [\tilde{G}^{k} \mathbf{1}]_{i} = (1-\delta)^{-1}\mathbb{E}_{k\sim \operatorname{Geom}(\delta)} [ \tilde{G}^{k}\mathbf{1} ]_{i},
\]
where the distribution is the geometric distribution with success probability $1-\delta$ supported on $\NN = \left\{ 0,1,\dots \right\}$, $\mathbf{1} = \mathbf{1_{n}}\otimes I_{m}$, and $[\cdot ]_{i}$ denotes the $i$-th block\cref{fn:block} of a matrix in $\R^{(n \times 1)(m\times m)}$.

\paragraph{Step 4. Sample a walk}
We have
\[
	[\tilde{G}^{k}\mathbf{1}]_{i}
	= \underset{j_{1},\dots,j_{k}\in \X_{i}}{\sum} \tilde{G}_{ij_{1}}\tilde{G}_{j_1j_2} \cdots \tilde{G}_{j_{k-1}j_{k}}
	= \underset{(i=j_{0}, j_{1}, \ldots, j_{k})}{\mathbb{E}} \left[  \prod_{l=0}^{k-1} \deg_{j_{l}} \cdot \prod_{l=0}^{k-1} \tilde{G}_{j_{l}j_{l+1}} \right],
\]
where $\deg_{j}$ is the in-degree of node $j$ in $G$, and
$(i=j_{0},\dots,j_{k})$ is a random walk of length $k$ starting at $i$ in which each step samples the next agent uniformly from agent $i$'s
incoming neighbors.

Steps 2--4 interpret the composite structures in the exact characterization of $\tilde{\Phi}$ as expectations of random variables. Replacing these expectations with  appropriately generated stochastic realizations of the corresponding random variables then yields an unbiased estimate of $\tilde{\Phi}$.

Recall our construction $\tilde{G}_{ij} = G_{ij}\tilde{W}_{i}\tilde{W}_{j}^{-1}$ of the bundle-induced network $\tilde G$. Suppose the original graph is unweighted ($G\in\{0,1\}^{n\times n}$).  {Then the expectation  in Step 4 involves a telescoping product:}
\[
	\prod_{l=0}^{k-1} \tilde{G}_{j_{l}j_{l+1}} \!=\! \prod_{l=0}^{k-1} G_{j_{l}j_{l+1}} \prod_{l=0}^{k-1} \tilde{W}_{j_{l}} \tilde{W}_{j_{l+1}}^{-1}
	\!=\! 1 \cdot \tilde{W}_{j_0} \tilde{W}_{j_1}^{-1} \tilde{W}_{j_1} \tilde{W}_{j_2}^{-1} \cdots \tilde{W}_{j_{k-1}} \tilde{W}_{j_k}^{-1} \!=\! \tilde{W}_{i} \tilde{W}_{j_k}^{-1}
	.\]
This observation motivates us to define  $Z_t$ as the random tuple
\[
	Z_t = \Bigl( & i_{t}\sim \operatorname{Unif}(\X), k^{(1)}_{t}, k^{(2)}_{t} \overset{ \text{i.i.d.} }{ \sim } \operatorname{Geom}(\delta), \bigl(j^{(1)}_{0}=i_{t}, j^{(1)}_{1}, \ldots, j^{(1)}_{k}\coloneqq j^{(1)}_{k^{(1)}_{t}}\bigr), \\
	&\bigl(j^{(2)}_{0}=i_{t}, j^{(2)}_{1}, \ldots, j^{(2)}_{k}\coloneqq j^{(2)}_{k^{(2)}_{t}}\bigr) \text{ are two independent random walks}\Bigr),
\]
so that we have
\(
\hat{\Phi}_t = \hat{\Phi}(Z_t) \coloneqq (1-\delta)^{-2}  \prod_{l=0}^{k^{(1)}_{t}-1} \deg_{j^{(1)}_{l}} \prod_{l=0}^{k^{(2)}_{t}-1} \deg_{j^{(2)}_{l}} \cdot \tilde{W}_{j^{(1)}_{k}}^{-1} \tilde{W}_{i_{t}} \tilde{W}_{j^{(2)}_{k}}^{-1}
,\)
along with $\mathbb{E}_{Z_t}\hat{\Phi}(Z_t) = \tilde{\Phi}$.
Note that two independently generated random walks are required to obtain an unbiased estimate of the quadratic term in \cref{eq:cen-planner}.

\begin{wrapfigure}{r}{0.36\textwidth}
	\centering
	\includegraphics[width=0.36\textwidth]{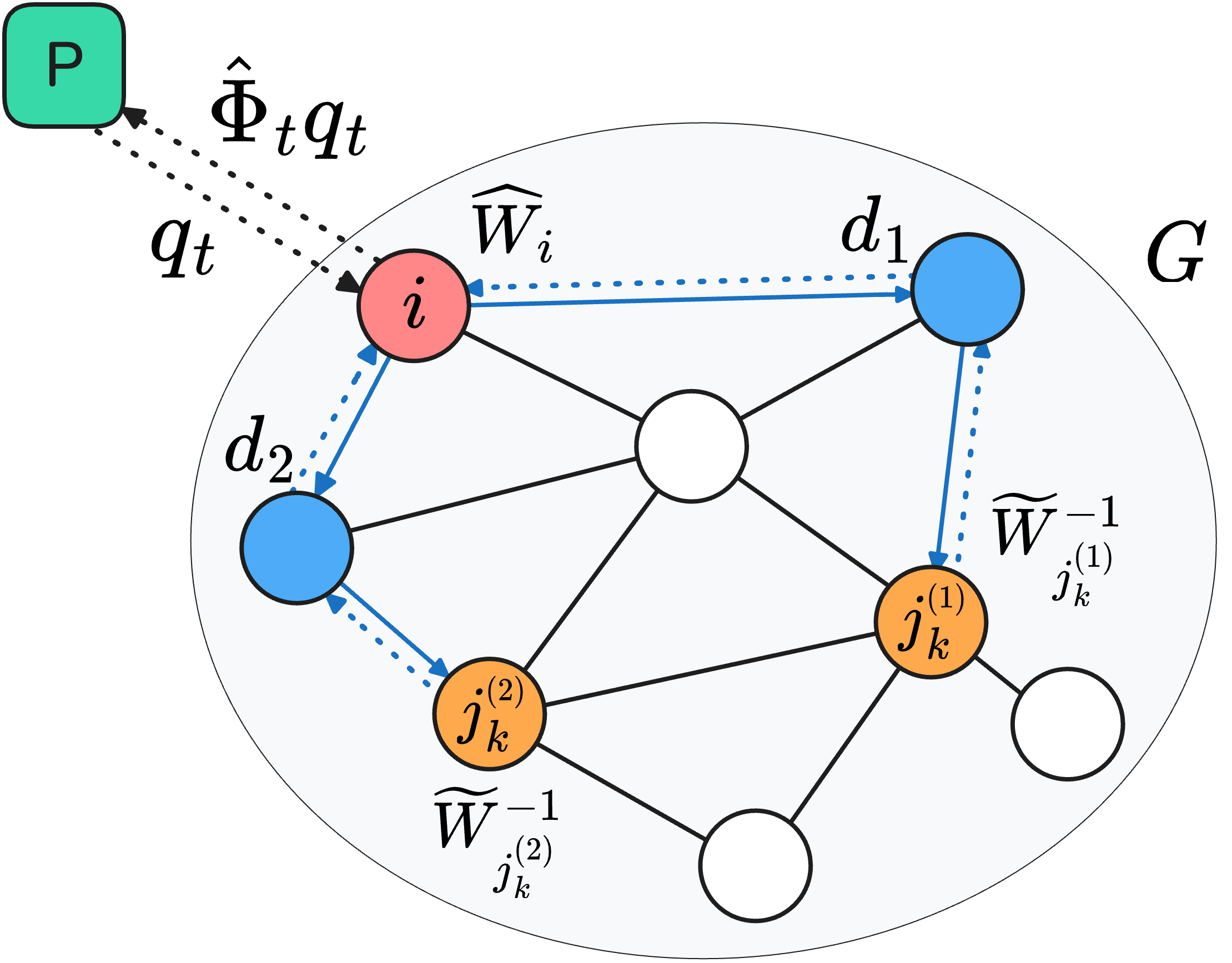}
	\caption{Illustration of \poll. Blue lines represent the random walks and dashed lines represent the information flow. $d_1,d_2$ are the degree information along two walks.}
	\label{fig:alg}
	\vspace{-1em}
\end{wrapfigure}

We name our method \crefrange{eq:poll-1}{eq:poll-2} \poll, as our sampling process implements a distributed ``polling'' process to estimate the aggregated preference.
Specifically, at each time step, the planner
\begin{enumerate*}[label=(\roman*)]
	\item selects an agent uniformly at random,
	\item initiates two independent random walks starting at this agent,
	\item collects degrees along each walk,
	\item asks the endpoint agents to compose the sample vote $\hat{\Phi}_tq_t$ according to their preferences, and then
	\item compares the sample vote with the current price $p_0$, which simulates the communication process between  bundle-induced networks, to update the price intervention as in \cref{eq:poll-1}.
\end{enumerate*}
Throughout the process, neither the planner nor any agent needs to know the network structure beyond its local neighborhood.
Moreover, agents  only report preference-price products and not their raw preferences, preserving privacy as well as reducing communication and computational cost.
\cref{fig:alg} illustrates the algorithm.
A detailed description of the algorithm and its implementation  are provided in \cref{apx:poll}.

\section{Efficiency guarantees}\label{sec:poll-prop}

\subsection{Computational efficiency}\label{sec:eff-comp}

Throughout, we suppose that sampling a neighbor takes $O(1)$ time and $\delta<\sqrt{2}/2$; the latter is implied by \cref{assump:net} whenever some agent has an in- or out-degree at least $2$.
The expected complexity of computing $\hat{\Phi}_{t}$ is $O((1-\delta)^{-1} + m^{2})=O(m^{2})$, independent of the number of agents $n$.
By contrast, computing $\tilde{\Phi}$ exactly requires $O(n^{3}+m^{2})$ time.
Together, \poll requires $O(tm^{2})$ operations with the number of queries $t$ being independent of $n$ (as characterized in \cref{sec:eff-learn}).

Beyond operational complexity, \poll, having been designed for lare-scale systems, imposes minimal resource requirements on both the planner and the agents.
Agents need not maintain any local state across iterations, incurring zero memory overhead.
Agents' data (private technology) is never collected centrally. Hence, the planner only needs a total storage of $O(m)$ to store the price vector.
The communication protocol is lightweight and fully asynchronous:
for each price update, only one queried agent communicates along two random walks and with the planner. The expected total number of communication rounds is $O(4(1-\delta)^{-1} t)=O(t)$.
As a stochastic approximation method implemented in an online manner, \poll automatically adapts to dynamic environments and is robust to both endogenous query noise and exogenous shocks (\cref{prop:noise}).

Alternatively, one can estimate Katz--Bonacich centralities by a distributed power iteration, i.e., \cref{eq:iter} with $\nu\tilde{\Phi}$ replaced by $\delta\tilde{G}$ \citep{foster2001FasterKatz}, at a cost of $O(nm^{2}\ln \epsilon^{-1})$ operations for accuracy $\epsilon$, and then solve \cref{eq:opt-planner} by consensus \citep{nedic2009DistributedSubgradient} or federated \citep{mcmahan2017Communicationefficientlearning} optimization, which repeatedly aggregate local variables across all agents and thus scale in both time and communication with $n$.
Sampling-based centrality estimators \citep{lin2021SAKEEstimating} can similarly be combined with gossip \citep{boyd2006Randomizedgossip,scaman2019Optimalconvergence}, which randomly activates pairs of neighboring agents, or random-walk methods \citep{mao2019WalkmanCommunicationEfficient,nedic2001IncrementalSubgradient,lopes2007Incrementaladaptive}, which pass a single token along the network.
Moreover, all these methods require agents to store and update a local centrality estimate, a memory burden that grows with the number of activities; \poll avoids these costs through sampling and memory-free queries (see \cref{tbl:comp}).

\subsection{Learning efficiency}\label{sec:eff-learn}

To discuss learning efficiency, we consider a system \textit{generated} from an underlying distribution, consisting of both a random graph model and a random preference model.

\begin{assumption}[Independence]\label{asmp:indep}
	Agents' preferences are  identically distributed and the probability distribution that generates these preferences is independent of the network topology.
\end{assumption}

This assumption \emph{separates} the agent space and activity space: an agent's network position carries no information about its activity preference.
Nonetheless, it is sufficiently general to allow for degree correlation, preference correlation, and coupling between the two (see \cref{tbl:vr}) since it imposes independence only on the marginal preference distribution; any two agents may still have correlated preferences that depend on their network positions.
While this generality covers a wide range of real-world scenarios, it also couples statistical dependencies with the network structure, making the analysis more involved.
We first give a universal guarantee for this model and then discuss how different correlation structures affect learning efficiency.

\begin{definition}[Preference heterogeneity]\label{def:het-pref}
	The preference heterogeneity on the activity space is
	\[
		\sigma^{2}_{\A} = \mathbb{E}\|W_{i}^{-1}\|_{\mathrm{F}}^{2}-\|\mathbb{E}W_{i}^{-1} \|_{\mathrm{F}}^{2}
		,\]
	where the expectation is taken w.r.t. the preference generation. \end{definition}

Defined as a variance, $\sigma_{\A}^2$ measures the \emph{diversity} of agents' preferences.
Intuitively, when agents are diverse, the planner must query more to form a fair price, and when agents are homogeneous, a single query suffices to learn the population preference.

\begin{definition}[Network heterogeneity]\label{def:het-net}
	Define the base network $\widehat{G}=I$, degree-scaled network $\widecheck{G} = D^{1 /2}GD^{1 /2}$, and their centralities $\widehat{\cen} = (I-\delta \widehat{G})^{-1}\mathbf{1}$ and $\widecheck{\cen} = (I-\delta \widecheck{G})^{-1}\mathbf{1}$.
	Recall the original network's centrality $\cen = (I-\delta G)^{-1}\mathbf{1}$.
	The network heterogeneity on the agent space is
	\[
		\sigma^{2}_{\X} = \mathbb{E}|\hat{\cen}_{i}\check{\cen}_{i}|^{2} - |\mathbb{E}\cen_{i}^{2}|^{2}		= \Var(\cen_i^{2})
		+ \mathbb{E} [ \hat{\cen}_i^{2}( \check{\cen}_i^{2}-\cen_i^{2} )
			+ \cen_i^{2} ( \hat{\cen}_i^{2}-\cen_i^{2} ) ]
		,\]
	where the expectation is the empirical average across agents.
\end{definition}

Network heterogeneity, mirroring the preference heterogeneity definition, measures the \emph{centrality dispersion} of the network.
It captures not only the centrality heterogeneity $\Var(\cen_i^{2})$, but also how network asymmetry assigns centralities to the agents, encoded by the discrepancy among $\hat{\cen}_i$, $\check{\cen}_i$, and $\cen_i$.
Consequently, even when all agents have the same centrality (e.g., in a regular graph), a random walk retains non-vanishing variance due to network asymmetry.

We  characterize the query complexity of \poll in terms of the  constants
$\bar{\psi}\coloneqq\lambda_{\max}(\Psi)$, $\ubar{\psi}\coloneqq\lambda_{\min}(\Psi)$, $\overline{\lambda} \coloneqq \max_{i\in\X}\lambda _{\max}(W_{i})$, $\underline{\lambda}\coloneqq \min_{i\in\X}\lambda _{\min}(W_{i})$,
and $\bar{d}\coloneqq \sup_{i\in\X}\deg_{i}$,
and condition numbers $\kappa_{\Psi}\coloneqq\bar{\psi}/\ubar{\psi}$, $\kappa_{W}\coloneqq \overline{\lambda}/\underline{\lambda}$, $\kappa_{\nu}\coloneqq (1-\nu \lambda_{\max}(\tilde{\Phi}))^{-1}$, and $\kappa_{\delta}\coloneqq (1-\delta \sigma_{\max}(G))^{-1}$.

\begin{theorem}[Query complexity]\label{thm:query}
		Under \cref{assump:net}--\ref{asmp:indep}
	and a constant step size $\alpha = \kappa_{\nu}\ln t /t$ with $t /\ln t \ge \max\{ 2, 4\sigma^{2} \}$,
	we have $\mathbb{E}\|p_t-p_{*}\|_{2}^{2} \le \tilde{O}(t^{-1}\|p_{*}\|_{2}^{2}\sigma^{2})$, where
	\[
		\sigma^{2} = O\bigl( \nu^{2}\kappa_{\nu}^{2} \kappa_{\Psi}^{2}\kappa_{W}^{2}(\sigma_{\A}^{2} \sigma_{\X}^{2} + \max \{\kappa_{\delta}^{2}\sigma^{2}_{\A}, \bar{\lambda}^{-2}\sigma^{2}_{\X} \}) \bigr)
		,\]
	where $\sigma_{\A}$ and $\sigma_{\X}$ are as defined in \cref{def:het-net,def:het-pref}, respectively.
\end{theorem}

The rate established in Theorem \ref{thm:query} matches the standard sample complexity for stochastic approximation \citep{zhang2025semisgd}.
The proof is deferred to \cref{apx:eff-learn}.
For a well-conditioned system with $\kappa_{\nu},\kappa_{\Psi},\kappa_{W},\kappa_{\delta},\|p_*\|_{2}=O(1)$, using the conservative bounds $\sigma_{\A}^{2}\overset{\scriptscriptstyle\cref{eq:varA-trivial}}{\le} \kappa_{W}^{2}\ubar{\lambda}^{-2}$ and $\sigma_{\X}^{2}\overset{\scriptscriptstyle\cref{eq:varX-trivial}}{\le} (1-\delta \sigma_{\max}(DG))^{-2}\le (1-\delta \bar{d}^{2})^{-2}$ gives a query complexity of $t = O(\epsilon^{-2} \ubar{\lambda}^{-2}(1-\delta \bar{d}^{2})^{-2})$.

\cref{thm:query} identifies the \textit{heterogeneity} levels $\sigma_{\X}$ and $\sigma_{\A}$ as the bottleneck of learning efficiency, confirming the intuition that each query is more representative when the population is more homogeneous.
Importantly, $\sigma_{\X}$ and $\sigma_{\A}$ do not depend explicitly on the system dimension $n$.

\begin{remark}[Free lunch?]
	The query complexity of \poll is heterogeneity-aware but dimension-agnostic. It thus scales to large systems given controlled heterogeneity.
	For a $d$-regular graph, for instance, we have $\sigma_{\X}^{2} \le (1-\delta d^{2})^{-2}$, independent of $n$.
	However, the heterogeneity of real systems typically grows with the system size.
	In the worst case, a star graph has $\bar{d}=n$ and $\sigma_{\max}(DG)=n^{2}$, degrading the bound to $O(\epsilon^{-2}(1-\delta n^{2})^{-2})$.
		Hence, \emph{no} free lunch exists in general.

	That said, \poll is also aware of another facet of system scale: the strengths of the network effect and regularization.
	The discount factors $\delta$ and $\nu$ control the influence of the network effect and agents' preferences, compared to the idiosyncratic component and the planner's preference, respectively.
	It is therefore intuitive that \poll concentrates its queries on learning these two components, with complexity decreasing as the network effect weakens or regularization strengthens.
	An exact method fails to leverage this structure.
\end{remark}

\poll enjoys variance reduction via structures common in real-world systems and via techniques suggested by our characterization in \cref{prop:cen-planner}, catalogued in \cref{tbl:vr}.

\begin{table}[th]
	\vspace{-1em}
	\caption{Variance reduction.}
	\label{tbl:vr}
	\centering
	\resizebox{\textwidth}{!}{		\begin{tabular}{llll}
			\toprule
			Structure                                         &
			Scenario                                          &
			Technique                                         &
			Reduction                                                             \\\midrule
			Small network effect                              & Economic networks
			\citep{parasnis2025incentive}                     &
			Control variate                                   &
			Prop.~\labelcref{prop:cv}                                             \\
			Preference correlation                            &
			Herding \citep{acemoglu2011BayesianLearning}      &
			Not needed                                        &
			Prop.~\labelcref{prop:pref-corr}                                      \\
			Degree correlation                                &
			Assortativity \citep{newman2018Networks}          &
			Not needed                                        &
			Prop.~\labelcref{prop:deg-corr}                                       \\
			Power law degrees                                 &
			Social networks \citep{jackson2011Socialeconomic} &
			Importance sampling                               &
			Prop.~\labelcref{prop:power}                                          \\
															\bottomrule
		\end{tabular}}
	\vspace{-0.4em}
\end{table}

\subsection{Economic efficiency}\label{sec:eff-econ}

To achieve economic efficiency, i.e., to maximize social welfare under the planner's budget constraint,
\poll relies on agents reporting their preferences truthfully.
A strategic agent, however, may misreport to shift the price in its favor; a malicious agent may report arbitrarily to disrupt learning altogether.
Facing different agent behaviors, \poll readily equips the planner with a dedicated device to suppress manipulation, summarized in \cref{tbl:sp}.
Payment mechanisms \citep{groves1973Incentivesteams}, the most common tool to elicit truthful reporting, require the planner to collect money after each report, which is often infeasible in practice.
In this section, we instead leverage the fundamental ambiguity from \cref{sec:ce} to deter strategic manipulation without payment; the other two devices are detailed in \cref{apx:eff-econ}.

\begin{table}[th]
	\vspace{-1em}
	\caption{Manipulation robustness. Approximate truth presents the asymptotic scale of the maximum deviation from the truthful price update as $\alpha_t \to 0$.
		By \cref{prop:noise}, a deviation of $O(\alpha_t ^{\scriptscriptstyle 3 /2})$ does not affect the learning complexity.}
	\label{tbl:sp}
	\centering
	\begingroup\small
	\begin{tabular}{llll}
		\toprule
		Behavior                                                  &
		Device                                                    &
		\multicolumn{2}{l}{Approximate truth}                                                          \\\midrule
		Ambiguity-averse                                          &
		Knightian uncertainty                                     &
		$\alpha_t^{2}$                                            & (Prop.~\labelcref{prop:knight})    \\
		Payoff-maximizing                                         &
		Groves payment                                            &
		$\alpha_t^{2}$                                            & (Prop.~\labelcref{prop:groves})    \\
		Malicious                                                 &
		Byzantine audit                                           &
		$\alpha_t\delta \bar{d} /\ubar{\lambda} \sqrt{\ln m / B}$ & (Prop.~\labelcref{prop:byzantine}) \\
		\bottomrule
	\end{tabular}\endgroup
\end{table}

Without payment, preventing misreporting is fundamentally difficult: in the linear-quadratic game, an agent's utility strictly increases with both the price and the reported utility kernel $\hat{\Phi}_{t}$ in the Loewner order (see \cref{eq:eu}), incentivizing it to inflate its report.
Hope arises from the observation that in \poll, the queried agent reports only its own preference $W_i^{-1}$.
As \cref{sec:ce} establishes, the network topology creates a fundamental ambiguity: whether \textit{overstating} or \textit{understating} $W_i$ raises the agent's utility depends on the unknown network structure.
An agent uncertain about this direction faces the risk that misreporting in the wrong direction yields a strictly lower utility than reporting truthfully.

To engineer this uncertainty, the planner signals the random walk information to the queried agent before soliciting the latter's report. The report is then combined with preference--price products from the endpoint agents to compute the price update.
The planner's strategy is to send multiple \textit{decoy} signals---fictitious walk summaries not used for the actual update---and announces only that one of the signals is genuine.
Unable to identify the true signal, the agent forms an uncertainty set over the possible walk outcomes. The planner's choice of decoys controls the shape of this set.

\begin{proposition}\label{prop:knight}
	Let $p(\hat{W}_{i})$ be the price updated by \poll (with control variate \cref{eq:cv}) given the agent's reported preference $\hat{W}_{i}^{-1}$, and $u_{i}(\hat{W}_i) \coloneqq u_{i}(s^*(p(\hat{W}_{i})),p(\hat{W}_{i}))$.
	Consider the following setup:
	\begin{itemize}\tightlist
		\item \emph{(One-dimensional manipulation).} The agent only overstates or understates its preference by a factor $\rho$, i.e., $\hat{W}_{i} = \rho W_{i}$, and the report must satisfy the known bounds $\ubar{\lambda} I\preceq \hat{W}_{i} \preceq \bar{\lambda} I$.
		\item \emph{(Planner's signals).} For each query, the planner sends multiple signals, one of which contains the true walk information, one of which contains degree products with mean $d_{-} < (1-\delta)\delta^{-1} \kappa_{W}^{\scriptscriptstyle-1 /2}$, and one of which contains $d_{+}> (1-\delta)\delta^{-1}\kappa_{W}^{\scriptscriptstyle1 /2}$.
		\item \emph{(Agent's uncertainty).} For each signal, the agent constructs a belief $\mathcal{P}$, with the mean of the degree products equal to the signal mean, and the mean of the endpoint agents' preferences equal to its own preference in the absence of other information.
		\item \emph{(Maxmin utility).} Given an uncertainty set $\mathcal{C}$, the agent  minimizes the worst-case expected regret, or equivalently, maximize the worst-case expected utility gain
		      \(
		      \max_{\hat{W}_{i}} \min_{\mathcal{P}\in\mathcal{C}} \mathbb{E}_{\mathcal{P}}[u_{i}(\hat{W}_{i}) - u_{i}(W_{i})]
		      .\)
	\end{itemize}
	If $\mathcal{C}$ contains $d_{-}$ and $d_{+}$, the agent's optimal strategy deviates from the truth $W_{i}$ by at most $O(\alpha_t)$.
\end{proposition}

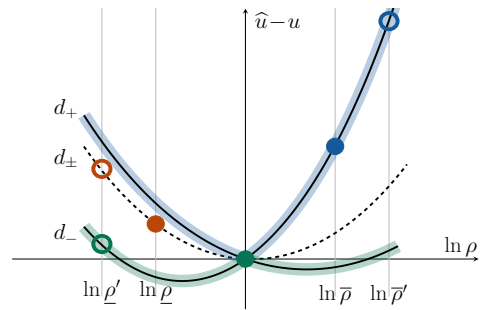
\begin{wrapfigure}{r}{0.52\textwidth}
	\vspace{-1em}
	\centering
	\begin{tikzpicture}[scale=0.497,yscale=0.95]
				\tikzset{every node/.style={font=\LARGE}}
		\begin{axis}[
				axis line style={thick},
				axis lines=middle,
				xmin=-1.3, xmax=1.3,
				ymin=-0.2, ymax=1,
				samples=200,
				width=14cm,
				height=10cm,
				domain=-1:1,
				grid=none,
				xtick=\empty,
				ytick=\empty,
								declare function={
						u(\x,\c) = (1-exp(\x))*(0.99-\c*exp(\x)) / 2/ exp(\x);
					}
			]
			\node at (axis cs:0.18,0.95) {$\hat{u}\!-\!u$};
			\node at (axis cs:1.2,0.05) {$\ln \rho$};

			\addplot[lightgray] coordinates {(-0.8,-0.15) (-0.8,2)};
			\node[above] at (axis cs:-0.8,-0.2) {$\ln\ubar{\rho}'$};
			\addplot[lightgray] coordinates {(0.8,-0.15) (0.8,2)};
			\node[above] at (axis cs:0.8,-0.2) {$\ln\bar{\rho}'$};
			\addplot[lightgray] coordinates {(-0.5,-0.15) (-0.5,2)};
			\node[above] at (axis cs:-0.5,-0.2) {$\ln\ubar{\rho}$};
			\addplot[lightgray] coordinates {(0.5,-0.15) (0.5,2)};
			\node[above] at (axis cs:0.5,-0.2) {$\ln\bar{\rho}$};

			\addplot[black, domain=-0.9:0.9, dashed, black, thick, line width=1.5pt] {u(x,0.92)};
			\node at (axis cs:-1,0.4) {$d_{\pm}$};

			\addplot[domain=-0.9:0, mygreen!80!black, opacity=0.3, line width=10pt] {u(x,2)};
			\addplot[domain=0:0.85, mygreen!80!black, opacity=0.3, line width=10pt] {u(x,0.5)};
			\addplot[domain=0:0.85, myviolet!80!black, opacity=0.3, line width=10pt] {u(x,2)};
			\addplot[domain=-0.9:0, myviolet!80!black, opacity=0.3, line width=10pt] {u(x,0.5)};
			\addplot[black, thick, domain=-0.9:0.85, thick, line width=1.5pt] {u(x,2)};
			\node at (axis cs:-1,0.1) {$d_{-}$};
			\addplot[black, thick, domain=-0.9:.85, thick, line width=1.5pt] {u(x,0.5)};
			\node at (axis cs:-1,0.6) {$d_{+}$};

						\addplot[only marks, mark=*, mark size=6.5pt, myviolet!80!black]
			coordinates {(0.5,0.45)};
			\addplot[only marks, mark=o, mark size=6.5pt, myviolet!80!black, mark options={line width=3pt}]
			coordinates {(0.8,0.95)};

						\addplot[only marks, mark=*, mark size=6.5pt, myorange!80!black]
			coordinates {(-0.5,0.14)};
			\addplot[only marks, mark=o, mark size=6.5pt, myorange!80!black, mark options={line width=3pt}]
			coordinates {(-0.8,0.36)};

						\addplot[only marks, mark=*, mark size=6.5pt, mygreen!80!black]
			coordinates {(0,0)};
			\addplot[only marks, mark=o, mark size=6.5pt, mygreen!80!black, mark options={line width=3pt, draw=mygreen!80!black}]
			coordinates {(-0.8,0.06)};
			
		\end{axis}
	\end{tikzpicture}
	\caption{Expected utility gain $\hat{u}_i - u_i$ from reporting $\hat{W}_{i}=\rho W_{i}$.
		Solid curves are induced by signals $d_{-}$ and $d_{+}$.
		An optimistic agent selects the best case (blue) and overstates its preference;
		a ambiguity-averse agent selects the worst case (green) and reports truthfully $\rho=1$ when restricted to $[\ubar{\rho},\bar{\rho}]$, but misreports when the strategy space expands to $[\ubar{\rho}', \bar{\rho}']$.
		An agent that collapses the signals to a single belief selects the average case (dashed curve) and understates its preference (orange).}
	\label{fig:sb}
	\vspace{-2em}
\end{wrapfigure}

The uncertainty here can be attributed to \citet{knight1921Riskuncertainty}: the agent has no prior over the uncertainty set and cannot reduce it to a single belief.
Agents with such max-min preferences are \textit{ambiguity averse} \citep{gilboa1989Maxminexpected}, a well-documented feature of real-world decision-making.
Intuitively, Knightian uncertainty removes the incentive to misreport, as the agent does not want to regret overstating when it should have understated, or vice versa.
Additionally, the decoy signals must be strong enough to prevent the utility gain from surfacing above zero again at large deviations.
\cref{fig:sb} illustrates how different attitudes towards uncertainty lead to different reporting strategies.

We restrict the focus of \cref{prop:knight} to one-dimensional scaling and max-min utility for ease of exposition. A general manipulation space with a continuum of ambiguity attitudes is analyzed in \cref{apx:knight}.

\section{Numerical experiments}\label{sec:sim}

Please refer to \cref{apx:sim} for the detailed experimental setup and additional results.

\paragraph{Synthetic data}
Simulation results are illustrated in \cref{fig:sim}.
\cref{fig:scale} presents the scaling (wall-clock time vs. agent space size) of various methods on a directed hierarchical network, highlighting the computational efficiency of \poll and validating the  discussion in \cref{sec:eff-comp}.
\cref{fig:het} presents \poll's query complexity across different levels of preference and network heterogeneity on a variation of the random network model in \citet{watts1998Collectivedynamics}, confirming the heterogeneity-aware learning efficiency of \poll, as discussed in \cref{sec:eff-learn}.
\cref{fig:mani} presents \poll's mean relative error (MRE) w.r.t. the optimal price when facing ambiguity-averse strategic agents on an approximately regular graph, given no decoy signals, random decoy signals, or carefully designed decoy signals, confirming the robustness of \poll to manipulation, as discussed in \cref{sec:eff-econ}.

\begin{figure}[th]
	\centering
	\begin{subfigure}{0.325\textwidth}
		\centering
		\includegraphics[width=0.933\linewidth]{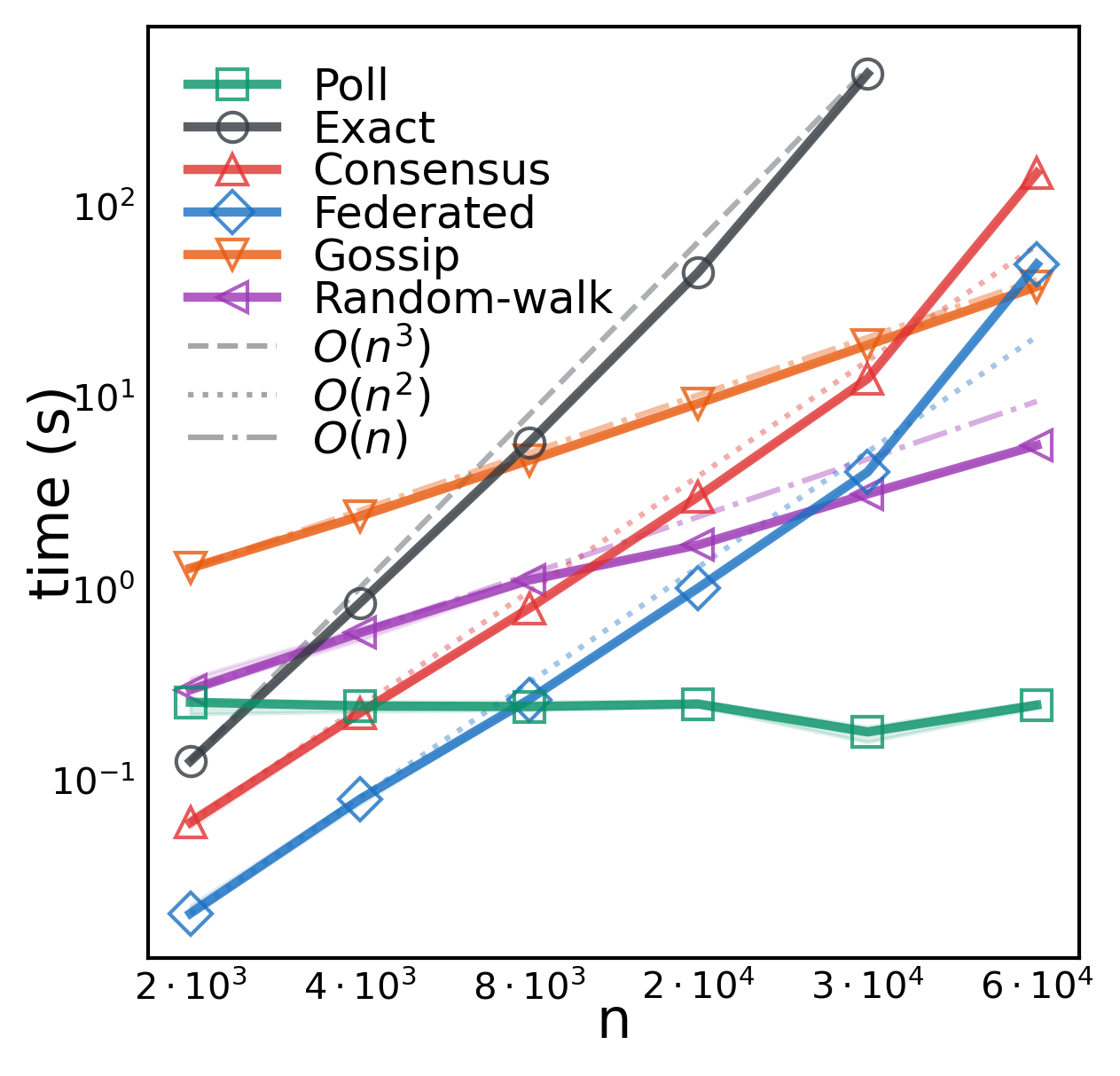}
		\vspace{-4.6px}
		\caption{Scaling.}
		\label{fig:scale}
	\end{subfigure}\hfill
	\begin{subfigure}{0.35\textwidth}
		\includegraphics[width=\linewidth]{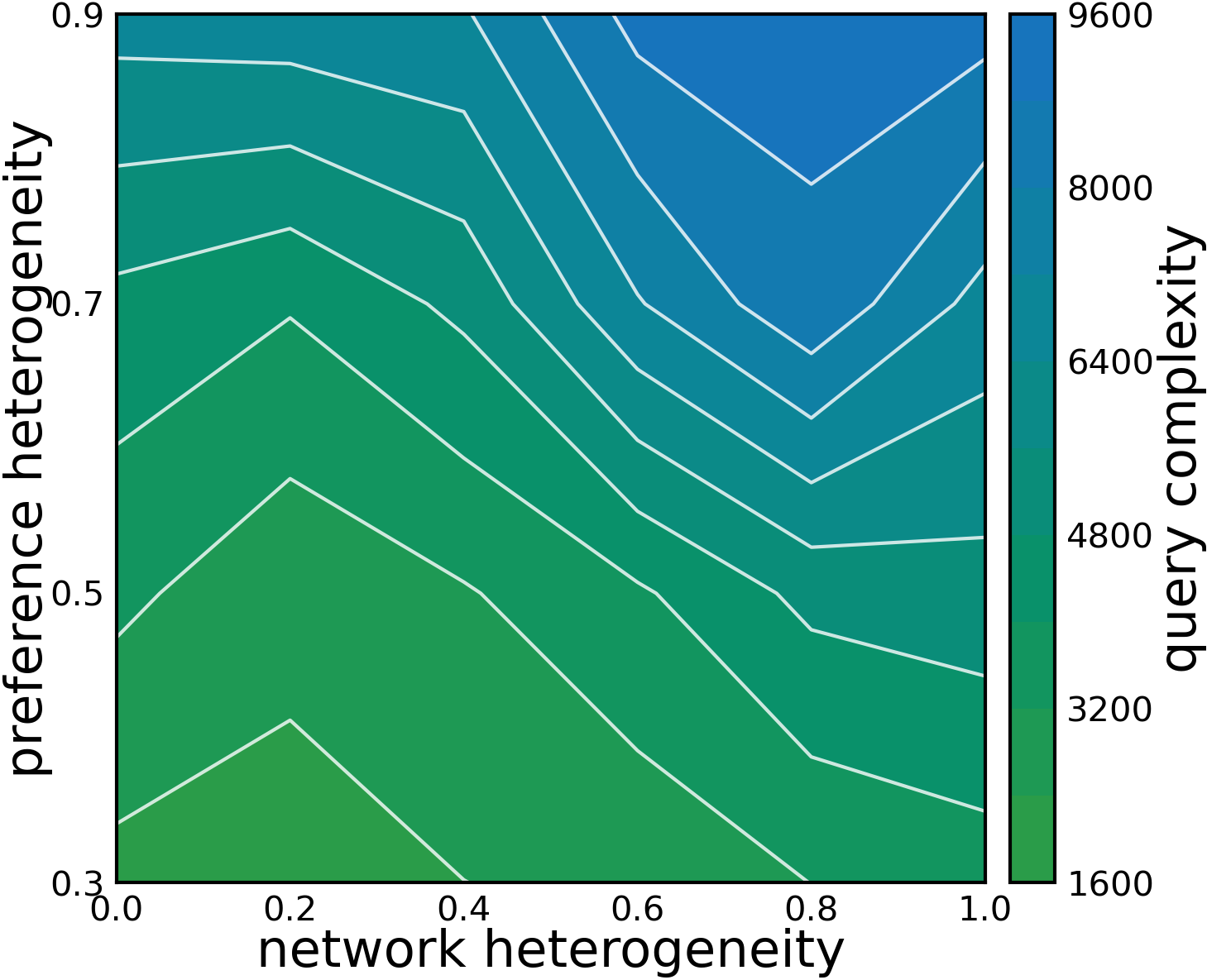}
		\vspace{-14px}
		\caption{Heterogeneity-aware complexity.}
		\label{fig:het}
	\end{subfigure}\hfill
	\begin{subfigure}[b]{0.325\textwidth}
		\centering
		\includegraphics[width=0.947\linewidth]{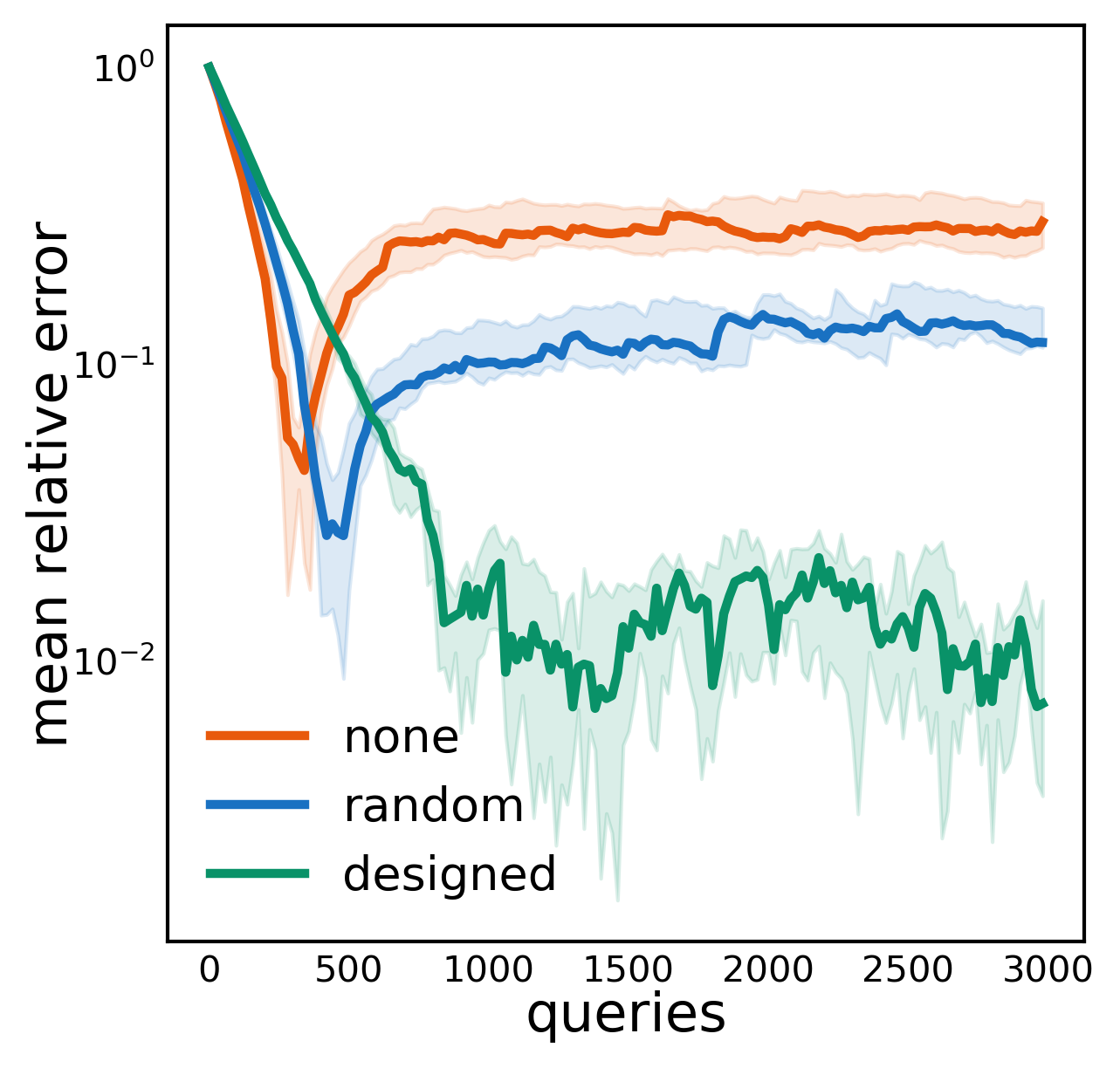}
		\vspace{-5.7px}
		\caption{Truth under uncertainty.}
		\label{fig:mani}
	\end{subfigure}
	\caption{Simulation results.}\label{fig:sim}
\end{figure}

\paragraph{Real-world network}
We further evaluate \poll on the complete DBLP coauthorship network \citep{yang2012Definingevaluating} with $n=317{,}080$ authors and $1{,}049{,}866$ edges.
Authors allocate research effort across the $m=10$ largest research communities (activities), and a funding agency sets field-level funding (prices); an author's private technology encodes its community memberships (see \cref{apx:sim-dblp} for the full setup).
For this experiment, we report the normalized welfare gap $\operatorname{Gap}_t \coloneqq (\operatorname{SW}(p_*)-\operatorname{SW}(p_{t})) / (\operatorname{SW}(p_*)-\operatorname{SW}(p_{0}))$, which is quadratic in the price's MRE, and measure computation by communication complexity, i.e., the number of pairwise exchanges.
As \cref{tbl:dblp} shows, \poll's advantage is amplified on this large real-world network: it reaches $\operatorname{Gap}\le 10^{-5}$ with $316\times$ to $1572\times$ fewer exchanges than the other distributed methods, and by the time Federated and Consensus produce their first updated iterate, \poll has already reduced the gap below $4\times 10^{-8}$.

\vspace{-0.08em}
\begin{table}[th]
	\caption{Communication complexity to reach a welfare gap threshold on the DBLP network. Median pairwise exchanges over three runs is reported; the deviation is under $10\%$ of the median in all cases. Exact computation is omitted due to infeasibility at this scale. Random-Walk does not reach either threshold within $10^9$ exchanges.}
	\label{tbl:dblp}
	\centering
	\begingroup\small
	\begin{tabular}{lllll}
		\toprule
		Gap threshold & \poll & Gossip & Federated & Consensus \\
		\midrule
		$10^{-5}$ & $\mathbf{1.20\times 10^{4}}$ & $3.80\times 10^{6}$ \scriptsize($316\times$) & $1.23\times 10^{7}$ \scriptsize($1023\times$) & $1.89\times 10^{7}$ \scriptsize($1572\times$) \\
		$10^{-7}$ & $\mathbf{5.25\times 10^{5}}$ & $5.01\times 10^{6}$ \scriptsize($10\times$)  & $1.78\times 10^{7}$ \scriptsize($34\times$)   & $2.52\times 10^{7}$ \scriptsize($48\times$)   \\
		\bottomrule
	\end{tabular}\endgroup
\end{table}

On the same network, we also stress-test the incentive guarantee with multi-dimensional reports and a continuum of ambiguity attitudes, from fully ambiguity-averse as in \cref{prop:knight} to fully optimistic.
The designed decoys are still able to elicit truthful reports for a nontrivial range of ambiguity attitudes; see \cref{apx:sim-dblp} for details.

\section{Conclusion}

We derived a polling algorithm whose computational, learning, and economic efficiency all flow from one centrality-based characterization of the planner's optimal intervention.
Some of our results rest on a linear-quadratic utility, and the payment-free incentive guarantee requires sufficient ambiguity aversion. Extending to general utilities and designing decoys for more optimistic agents remain open.

\begin{ack}
NSF Award Number 2039771. Directorate for Engineering, Division of Civil, Mechanical, and Manufacturing Innovation
\end{ack}

\bibliography{\subfix{main}}
\bibliographystyle{abbrvnat}

\newpage
\doparttoc
\faketableofcontents
\appendix
\part{Appendix} \parttoc
\newpage
\renewcommand{\style}{\origdisplaystyle}

\section{Additional numerical experiments}\label{apx:sim}

\subsection{Synthetic data}\label{apx:sim-syn}

\paragraph{Setup}
We focus on the simplified setting of $\Psi=U=I_{m}$.
The preferences are generated as follows: a base preference $\lambda_{a} \sim \operatorname{Unif}[2,4]$ is drawn independently for each activity $a$, and each agent's preference for activity $a$ is generated as $\lambda_{i,a} = \lambda_{a} + \epsilon_{i,a}$, with $\epsilon_{i,a} \sim \operatorname{Unif}[-\sigma_{\A},+\sigma_{\A}]$. Then the preference heterogeneity is of scale $\sigma_{\A}^{2}$.
All experiments are averaged over 10 runs with different random seeds, and the interquartile range is plotted for all curves.
All experiments are implemented in Python and run on a MacBook Pro with Apple M1 Pro chip and 16GB RAM.

\paragraph{Computational efficiency}
This experiment inspects how the operation complexity of different methods scales with the system size $n$.
We focus on an activity space dimension $m=10$ and three random graph models: directed hierarchical network (\cref{fig:scale}), Erd\H{o}s--R\'{e}nyi (ER) model (\cref{fig:scale-er}), and approximately regular graph (\cref{fig:scale-reg}).
The directed hierarchical network is generated as a $k$-ary tree with leaf nodes randomly connected with each other to ensure connectivity.
The ER model is generated with edge probability $k /n$.
The approximately regular graph starts with a $k$-regular graph, and then each edge is rewired with probability 0.05.
For this experiment, we set $k=10$ and $\sigma_{\A}=0.05$.
$\delta$ and $\mu$ are calculated after the graph is generated to ensure the small network effect and sufficient regularization conditions.
The mean relative error (MRE) tolerance is set to $10^{-3}$ for all methods except the exact method.
For other distributed methods, the power iteration \citep{foster2001FasterKatz} is used to compute the Katz--Bonacich centrality, asychronously with the price updates.
We observe that among all scenarios, \poll is the only method whose operation complexity does not scale with $n$; other methods' complexities all have a polynomial dependence on $n$, with the exact method scaling the worst at $O(n^{3})$.

\begin{figure}[th]
	\centering
	\begin{subfigure}{0.34\textwidth}
		\centering
		\includegraphics[width=\linewidth]{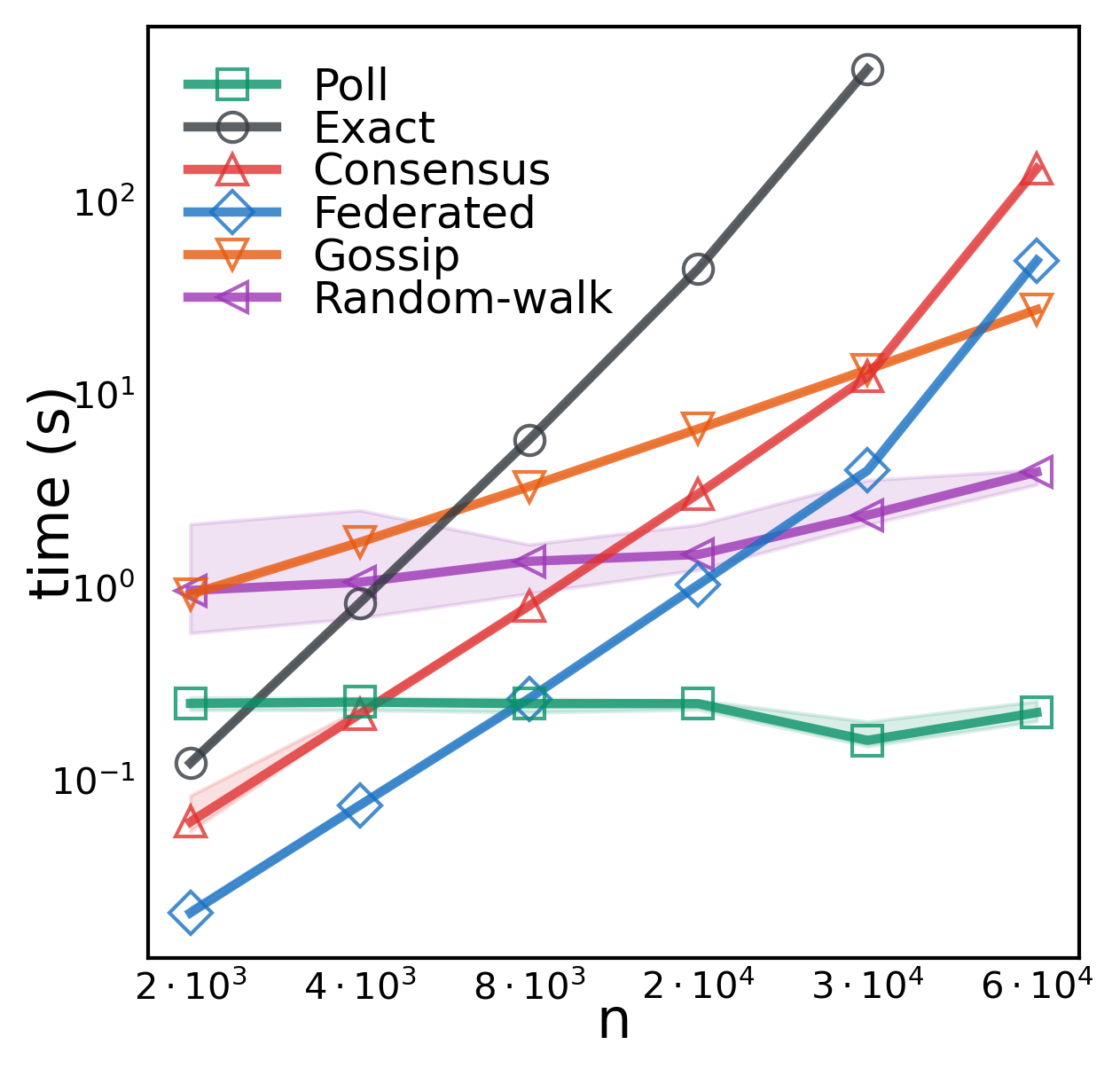}
		\caption{Erdős–Rényi}\label{fig:scale-er}
	\end{subfigure}
	\hfil
	\begin{subfigure}{0.34\textwidth}
		\centering
		\includegraphics[width=\linewidth]{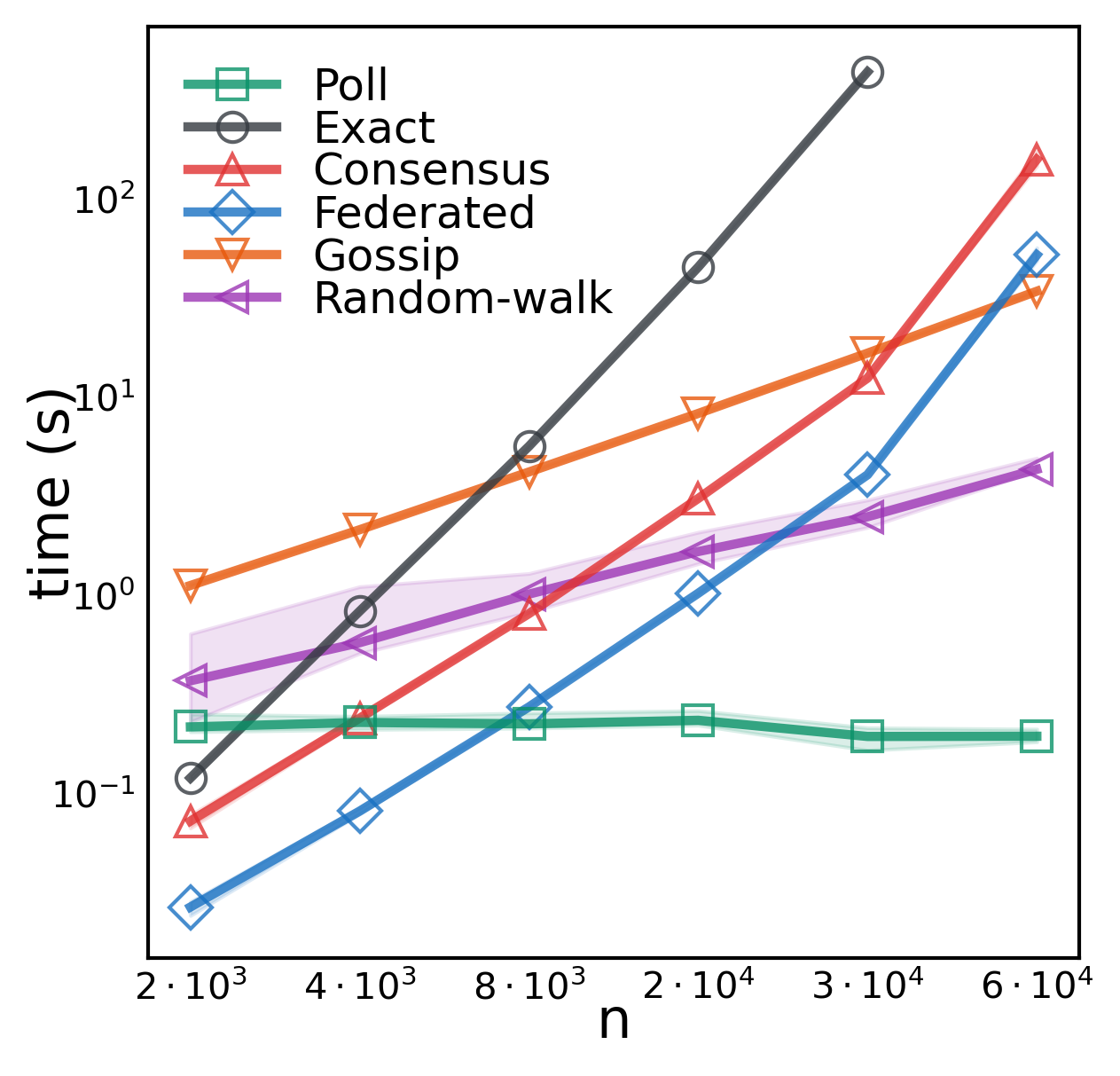}
		\caption{Approximately regular}\label{fig:scale-reg}
	\end{subfigure}
	\caption{Further scaling.}
\end{figure}

\paragraph{Learning efficiency}
This experiment inspects how different levels of preference and network heterogeneity affect \poll's query complexity.
We focus on an activity space dimension $m=10$ and a directed variation of the network in \citet{watts1998Collectivedynamics}:
start with a directed $k$-regular ring, with each node $i$ has out-edges to $(i+1)\bmod n, ..., (i+k/2)\bmod n$ and in-edges from $(i-k/2)\bmod n, ..., (i-1)\bmod n$, and then each edge is rewired with probability (w.p.) $p$; when rewiring, w.p. $1 /2$ the outgoing end is rewired to a $j$ w.p. proportional to $n+1-j$, and w.p. $1 /2$ the incoming end is rewired to a $j$ w.p. proportional to $j+1$.
Parameter $p$ controls the network heterogeneity, with $p=0$ giving a regular ring network and a larger $p$ giving a larger degree heterogeneity, centrality heterogeneity, and network asymmetry.
$k=10$ and the tolerance is set to $10^{-3}$ for this experiment.
The results reported in \cref{fig:het} confirm that the query complexity increases with both preference and network heterogeneity.
We remark that, however, $p$ is just a proxy for the network heterogeneity $\sigma_{\A}$, but not exactly equivalent.

\paragraph{Economic efficiency}
This experiment inspects how different signals affect \poll's performance facing ambiguity-averse strategic agents.
For simplicity, we focus on an activity space dimension $m=1$, and an approximately $10$-regular graph with $n=1000$ nodes (results on an ER graph $G(n,10/n)$ and on a directed Watts--Strogatz graph with nominal mean out-degree $10$ are similar and thus omitted).
\cref{prop:knight} argues that the decoy signals should be designed according to $\delta$ and $\kappa_{W}$ (controlled by $\sigma_{\A}$ in our setup and calculated after the preferences are realized).
Thus, we vary $\delta$ and $\sigma_{\A}$ to inspect the performance under different signals: no signal, random signals, and carefully designed signals.
Agents have a maxmin preference specified in \cref{prop:knight} and always have the incentive to misreport without signals.
A random signal is generated as $d_{\mathrm{r}} = (1-\delta)\delta^{-1}\kappa$ with $\ln \kappa \sim \mathrm{Unif}[-\ln\kappa_{W},\ln\kappa_{W}]$.
Therefore, w.p. $0.5\times 0.5$, two random signals satisfy the conditions of well-designed signals.
The results are reported in \cref{tab:mani}.
We observe that random signals induce ~50\% truthful reporting. This is because agents systematically want to understate their preferences, and thus a sufficiently strong one-directional signal is able to deter misreporting.
This universal tendency is due to a small network effect, where agents believe their self-contribution to their centrality dominates others' contribution (see \cref{sec:ce}).
As the network effect strengthens, the agent's belief becomes more balanced and the misreporting behavior becomes more uniform. However, a larger $\delta$ also increases the variance of the method and thus the query complexity.

\begin{table}[th]
	\caption{\poll performance and misreporting behaviors under different signals.
		The mean relative error (MRE) is averaged over the last 100 iterations of the total 3000 iterations.
		"T", "O", and "U" represent "truthful", "overstated", and "understated" respectively.}
	\label{tab:mani}
	\resizebox{\textwidth}{!}{
		\begin{tabular}{llllllll}\toprule
			\multirow{2}{*}{$\delta$} & \multirow{2}{*}{$\sigma_\A$} & \multicolumn{2}{c}{No signal} & \multicolumn{2}{c}{Random signals} & \multicolumn{2}{c}{Designed signals}                                              \\\cmidrule(lr){3-4}\cmidrule(lr){5-6}\cmidrule(lr){7-8}
			                          &                              & MRE                           & T / O / U (\%)                     & MRE                                  & T / O / U (\%) & MRE      & T / O / U (\%) \\\midrule
			\multirow{5}{*}{0.01}     & 1.0                          & 2.22e-01                      & 0 / 0 / 100                        & 9.92e-02                             & 59 / 0 / 41    & 1.75e-03 & 100 / 0 / 0    \\
			                          & 1.5                          & 3.77e-01                      & 0 / 0 / 100                        & 1.69e-01                             & 59 / 0 / 41    & 2.26e-03 & 100 / 0 / 0    \\
			                          & 2.0                          & 5.93e-01                      & 0 / 0 / 100                        & 2.67e-01                             & 58 / 0 / 42    & 2.94e-03 & 100 / 0 / 0    \\
			                          & 2.5                          & 9.43e-01                      & 0 / 0 / 100                        & 4.23e-01                             & 57 / 0 / 43    & 3.83e-03 & 100 / 0 / 0    \\
			                          & 3.0                          & 2.27e+00                      & 0 / 0 / 100                        & 9.50e-01                             & 56 / 0 / 44    & 5.71e-03 & 100 / 0 / 0    \\\midrule
			\multirow{5}{*}{0.03}     & 1.0                          & 1.16e-01                      & 0 / 3 / 97                         & 5.19e-02                             & 59 / 1 / 40    & 1.22e-02 & 100 / 0 / 0    \\
			                          & 1.5                          & 2.02e-01                      & 0 / 2 / 97                         & 9.08e-02                             & 58 / 1 / 41    & 1.20e-02 & 100 / 0 / 0    \\
			                          & 2.0                          & 3.19e-01                      & 0 / 2 / 98                         & 1.44e-01                             & 57 / 1 / 41    & 1.18e-02 & 100 / 0 / 0    \\
			                          & 2.5                          & 5.04e-01                      & 0 / 2 / 98                         & 2.29e-01                             & 56 / 2 / 42    & 1.14e-02 & 100 / 0 / 0    \\
			                          & 3.0                          & 1.09e+00                      & 0 / 2 / 98                         & 4.85e-01                             & 55 / 2 / 43    & 1.14e-02 & 100 / 0 / 0    \\\midrule
			\multirow{5}{*}{0.05}     & 1.0                          & 8.60e-02                      & 0 / 9 / 91                         & 5.01e-02                             & 59 / 4 / 37    & 8.94e-02 & 100 / 0 / 0    \\
			                          & 1.5                          & 1.30e-01                      & 0 / 10 / 90                        & 4.75e-02                             & 57 / 5 / 38    & 9.40e-02 & 100 / 0 / 0    \\
			                          & 2.0                          & 1.97e-01                      & 0 / 10 / 90                        & 5.26e-02                             & 55 / 6 / 39    & 9.94e-02 & 100 / 0 / 0    \\
			                          & 2.5                          & 2.91e-01                      & 0 / 11 / 89                        & 7.42e-02                             & 53 / 7 / 40    & 1.06e-01 & 100 / 0 / 0    \\
			                          & 3.0                          & 5.17e-01                      & 0 / 11 / 89                        & 1.70e-01                             & 50 / 9 / 41    & 1.16e-01 & 100 / 0 / 0    																																	\\\bottomrule
		\end{tabular}}
\end{table}

\subsection{Real-world network}\label{apx:sim-dblp}

\paragraph{Setup}
We use the complete DBLP coauthorship network \citep{yang2012Definingevaluating}, which contains $n=317{,}080$ authors and $1{,}049{,}866$ undirected coauthorship edges, with an average degree of $6.62$.
Authors are agents, and the $m=10$ largest ground-truth research communities $\mathcal{C}_{1},\dots,\mathcal{C}_{m}$ are activities.
Agent $i$ chooses its research effort $s_{i}\in\R^{m}$ across these communities, while a funding agency (the planner) sets a common vector of field-level funding $p\in\R^{m}_{\ge 0}$.
We model agent $i$'s private technology as
\[
	W_{i} = I_{m} + \operatorname{diag}\bigl(\mathbf{1}\{i\notin\mathcal{C}_{a}\}\bigr),
\]
so that membership in a community represents a lower effort cost, or equivalently a higher preference, for that field; in particular, $\kappa_{W}=2$.
The resulting utility is exactly \cref{eq:util}.
For intervention, we set $\Psi=I_{m}$, $\delta = \bigl(5\max_{i,a}\lambda_{i,a}[G\lambda_{a}^{-1}]_{i}\bigr)^{-1}$ with $\lambda_{a}^{-1}\coloneqq(\lambda_{i,a}^{-1})_{i\in\X}$, $\mu = 25\max_{a}\phi_{a}$, and $p_{0}\sim\operatorname{Unif}[0,1]^{m}$.
All modeling choices were fixed before implementing the methods or running the experiments.

\paragraph{Metrics}
We report the normalized social welfare gap
\[
	\operatorname{Gap}_{t} = \frac{\operatorname{SW}(p_*)-\operatorname{SW}(p_{t})}{\operatorname{SW}(p_*)-\operatorname{SW}(p_{0})},
\]
which is quadratic in the MRE $\|p_{t}-p_*\|_{2}/\|p_*\|_{2}$ used in \cref{apx:sim-syn}, since the social welfare is quadratic in the price.
Computation is measured by communication complexity, i.e., the number of bidirectional pairwise exchanges.
This is a close proxy for wall-clock time, as each pairwise exchange takes comparable time across methods, while more directly revealing the communication burden.
Exact computation is infeasible at this scale and is omitted.

\paragraph{Computational efficiency}
All baselines use sparse adjacency representations, and we sweep method-specific hyperparameters to select the best step size for each method.
We report tail-averaged iterates for \poll and Random-Walk, the last iterate for the deterministic Federated and Consensus methods, and the globally averaged last iterate for Gossip.
\cref{tbl:dblp,fig:dblp-budget} report the communication complexity required to reach a given welfare gap and the welfare gap attained with an increasing number of communication rounds, respectively.
At the threshold $10^{-5}$, \poll reduces communication by factors of $1572\times$, $1023\times$, and $316\times$ relative to Consensus, Federated, and Gossip, respectively; at the stricter threshold $10^{-7}$, the factors remain $48\times$, $34\times$, and $10\times$.
When the synchronous methods first produce an updated iterate (Federated at $1{,}366{,}946$ exchanges and Consensus at $2{,}099{,}732$ exchanges), \poll has already reduced the welfare gap to $3.15\times 10^{-8}$ and $2.51\times 10^{-8}$, respectively.

Federated eventually converges to machine precision through deterministic, synchronous full-population aggregation, whereas \poll delivers a near-optimal solution much sooner through stochastic, asynchronous updates.
The relevant practical comparison is therefore how much communication each method needs to attain a welfare gap already negligible to the planner.

\begin{figure}[th]
	\centering
	\includegraphics[width=0.6\textwidth]{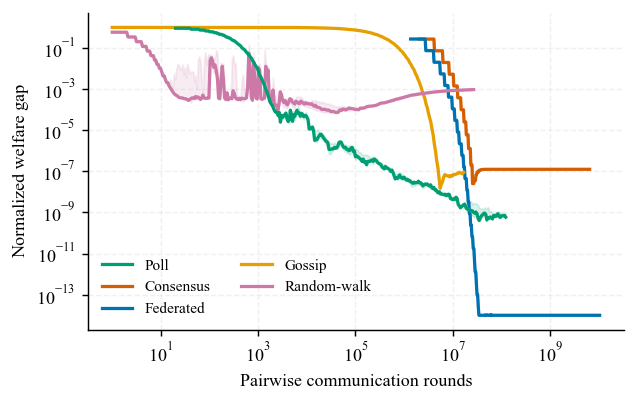}
	\caption{Welfare gap attained versus communication rounds on the DBLP network.}
	\label{fig:dblp-budget}
\end{figure}

\paragraph{Economic efficiency}
We extend the manipulation experiment of \cref{apx:sim-syn} to the DBLP network in two ways.
First, agents may manipulate their reports along all $m$ activities; following \cref{rmk:general-space}, \poll queries one activity at a time, which prevents an agent from coordinating its reports across activities and reduces each query to a one-dimensional manipulation problem.
Second, agents follow the $\theta$-maxmin preference (\cref{def:theta}) with $\theta\in[0,1]$, which interpolates between fully optimistic ($\theta=0$) and fully ambiguity-averse ($\theta=1$) attitudes.
In this experiment, each agent has only one profitable manipulation direction per activity: either to claim a community membership it does not have, or to disclaim one it has.
We run $500{,}000$ \poll iterations for each of $10$ random seeds and report the seed-median welfare gap and the truth rate (fraction of truthful reports).

Besides the designed and random signals of \cref{apx:sim-syn}, we consider \emph{polarized} signals.
Since each agent has a single profitable manipulation direction, the planner can construct two polarized signal sets: $\mathcal{C}_{-}$ deters claiming a membership the agent does not have but may encourage disclaiming one it has, and $\mathcal{C}_{+}$ has the opposite effect.
Without any knowledge of the queried agent, the planner selects each set with probability $1/2$.

\begin{figure}[th]
	\centering
	\includegraphics[width=0.7\textwidth]{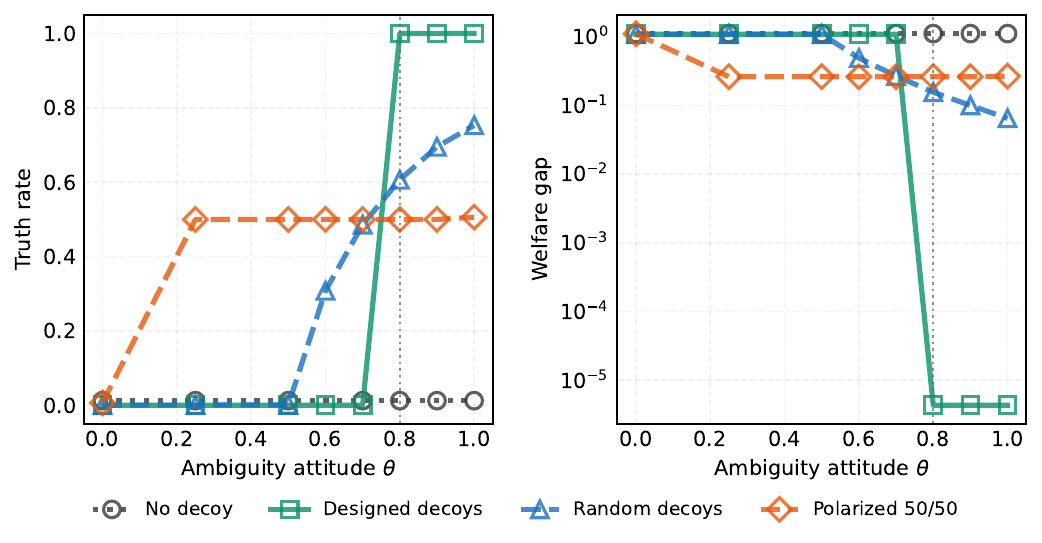}
	\caption{Overall truth rate and tail welfare gap on the DBLP network under $\theta$-maxmin agents, reported as seed medians over ten runs.}
	\label{fig:dblp-theta}
\end{figure}

\cref{fig:dblp-theta} shows that the designed signals eliminate manipulation for all $\theta\ge 0.8$, extending the maxmin guarantee of \cref{prop:knight} to a nontrivial range of ambiguity attitudes.
This matches \cref{prop:theta}: substituting $\kappa_{W}=2$, fixing $d_{-}=1$, the choice $d_{+}=\sqrt{2.5(1-\delta)^{2}/\delta^{2}-1}$ used here balances the two conditions and gives
\[
	\theta_{*} = 0.8 + \frac{6\delta^{2}}{25(1-\delta)^{2}-20\delta^{2}} \approx 0.8 \quad \text{for small } \delta.
\]
For more optimistic agents ($\theta<0.8$), the designed signals fail because they not only signal the worst case in both manipulation directions but also amplify the best case, which an optimistic agent exploits.
Random signals partially avoid this and retain non-vacuous deterrence: since each agent has a single profitable manipulation direction, two random signals may simultaneously encourage truth and discourage the profitable misreport.
Polarized signals make the reporting behavior insensitive to the ambiguity attitude, as their effect depends only on whether the selected signal set matches the agent's profitable manipulation direction.
As a compromise, polarized signals outperform uniform randomization in the optimistic regime ($\theta\le 0.7$) but underperform it in the pessimistic regime ($\theta\ge 0.8$).
Both randomizations are alternative distributions over the decoy signals of \cref{prop:knight} rather than standalone mechanisms, illustrating how the same design principle accommodates different ambiguity attitudes.

\section{General problem setup}\label{apx:setup}

\begin{table}[th]
	\caption{Notation.}
	\label{tbl:nota}
	\centering
	\resizebox{1\textwidth}{!}{%
		\begin{tabular}{cp{0.38\linewidth}cp{0.38\linewidth}}
			\bf Symbol        & \bf Description                          &
			\bf Symbol        & \bf Description                            \\
			\midrule
						$a$               & activity index                           &
			$\alpha$          & step size                                  \\
			$B$               & random walk batch size                   &
						$\beta$           & importance sampling exponent               \\
						$\mathcal{C}$     & uncertainty set                          &
			$\chi$            & Byzantine audit report                     \\
			$D$               & degree matrix                            &
			$\delta$          & network effect strength                    \\
			$d$               & degree product                           &
						$\epsilon$        & error                                      \\
			$\mathcal{F}$     & filtration                               &
			$\gamma$          & power law exponent                         \\
			$f,g,h$           & function                                 &
			$\iota$           & small probability                          \\
			$G$               & network adjacency matrix                 &
						$\kappa$          & condition number                           \\
			$i,j$             & agent index                              &
			$\lambda,\Lambda$ & eigenvalue(s) of activity kernel           \\
			$K$               & Katz--Bonacich centrality matrix         &
			$\mu$             & regularization strength                    \\
			$k$               & walk length                              &
			$\nu$             & regularization strength inverse            \\
			$l$               & walk step index                          &
			$\Phi$            & utility/social welfare kernel              \\
			$\A$              & activity space                           &
			$\phi$            & utility/social welfare kernel eigenvalue   \\
			$M$               & average preference                       &
			$\Psi$            & planner's activity kernel                  \\
			$m$               & number of activities                     &
			$\rho$            & misreport factor                           \\
			$\X$              & agent space                              &

			$\sigma^{2}$      & variance                                   \\
			$n$               & number of agents                         &
			$\tau,t$          & time step                                  \\
						$\mathcal{P}$     & belief                                   &
			$\xi$             & random walk                                \\
			$p,q$             & price                                      \\
						$R$               & regret                                     \\
			$r$               & correlation function                       \\
			$\mathcal{S}$     & strategy space                             \\
			$s$               & strategy/effort                            \\
						$U$               & unitary eigenbasis                         \\
			$u$               & utility                                    \\
			$v$               & payment                                    \\
			$W$               & activity kernel/technology                 \\
						$Z$               & random element                             \\
								\end{tabular}%
	}
\end{table}

This section extends \cref{sec:setup} to consider general multi-activity network games (MANGs) and derives their Nash equilibrium and optimal intervention.
We also summarize the notation used in the main text and appendix in \cref{tbl:nota}.

We consider a collection of agents (agent space, not necessarily finite) $\X$ and a collection of activities (activity space, not necessarily finite) $\A$.
In an MANG, each agent $i \in \X$ specifies its non-negative effort level on each activity $a \in \A$, giving a strategy profile $s$ in the strategy space $\mathcal{S} \subset \R^{\X \times \A}$.
The general form of agent $i$'s utility given strategy $s \in \S$ is given by
\[
	u_i(p,s) = \left< p(i,\cdot ), s(i,\cdot ) \right>_{\A} + \left< s(i,\cdot ), (Qs)(i,\cdot) \right>_{\A}
	,\]
where $p \in \R_{\ge 0}^ { \X \times \A }$ is the marginal price,
$Q \in \mathcal{L}(\X \times \A)$ is a linear operator that specifies the interaction between agents and activities, and $\left< \cdot, \cdot \right>_{\A}$ is the inner product over the activity space.
Thus, a quadratic multi-activity network game is specified by tuple $(\X, \A, \S, p, Q)$.
For a function $f$ on $\X\times \A$, we sometimes use the notation $f_{i,a} = f(i,a)$ for brevity.

\begin{remark}[General agent space and coaltion]\label{rmk:coalition}
	Although having a similar form, MANG is not equivalent to a network game on the general agent space $\X \times \A$, where each agent specifics a scalar effort $s_{i,a}$ and gets a payoff $u_{i,a}(p,s) = p_{i,a} s_{i,a} + s_{i,a} (Ms)_{i,a}$. Because in the latter case, calculating the best response involves competition between different activities for each $i\in \X$, while in a MANG, competition only happens between different agents.
	In other words, if we treat each agent-activity pair as a player in MANG, then fully cooperative \emph{coalitions} will be formed between $\{ (i,a) \}_{a\in\A}$ for each $i\in \X$.
	This viewpoint also helps us tackle strategic manipulation in \cref{apx:knight} where we need to break the potential coalition.
	In summary, the distinction between the agent space and the activity space in MANGs is crucial.
\end{remark}

\subsection{Nash equilibrium}

To calculate the best response, we decompose $Q$ as
\[
	Q((i,a),(i',a')) = \underbrace{Q((i,a),(i',a')) \mathbbm{1}\left\{ i=i' \right\}}_{{Q}^{\mathrm{I,n}} } + \underbrace{Q((i,a),(i',a')) \mathbbm{1}\left\{ i\neq i' \right\}}_{Q^{\mathrm{C}}}.
\]
That is, ${Q}^{\mathrm{I,n}}$, which is not necessarily symmetric, consists of the diagonal ``blocks'' capturing the idiosyncratic effects, while $Q^{\mathrm{C}}$ consists of the off-diagonal ``blocks'' capturing the network effects.
Then, we have $\frac{\partial (Q^{\mathrm{C}}s)_{i}}{\partial s_{i}} = 0$ and $\frac{\partial ({Q}^{\mathrm{I,n}}s)_{i}}{\partial s_{i}} = Q^{\mathrm{I,n}}$.
Define $Q^{\mathrm{I}} = -(Q^{\mathrm{I,n}} + (Q^{\mathrm{I,n}})^{H})$, where $H$ is the Hermitian transpose; $Q^{\mathrm{I}}$ is self-adjoint.
We generalize the small network effect assumption in \cref{assump:net}:

\begin{assumption}[Concave utility and small network effect] \label{assump:net-gen}
	$Q^{\mathrm{I}}$ is positive definite so the utility is concave in $s_{i}$.
	The network effect is sufficiently small such that $Q^{\mathrm{I}}\succ \frac{1}{2}(Q^{\mathrm{C}}+(Q^{\mathrm{C}})^{H})$ and thus $Q^{\mathrm{I}}-Q^{\mathrm{C}}$ is invertible.
\end{assumption}

Therefore, for an unbounded strategy space $\S = \R^{\X \times \A}$, the first-order condition (FOC) gives a unique Nash equilibrium under \cref{assump:net-gen}:
\[\label{eq:ne-general}
	\frac{\partial u_{i}}{\partial s_i} = p_{i} + (({Q}^{\mathrm{I,n}}+({Q}^{\mathrm{I,n}})^{H})s)_{i} + (Q^{\mathrm{C}} s)_{i} = 0
	\implies s^{*}(p) = (Q^{\mathrm{I}}-Q^{\mathrm{C}})^{-1} p
	.\]

\subsection{Specific structures}

When the agent and action spaces are finite with $|\X|=n$ and $|\A|=m$ and the game kernel $Q$ has a tensor (Kronecker) structure: $Q = (-\frac{1}{2}I_{n}+\delta G) \otimes W$, where $G\in\R^{n \times n}$ is the adjacency matrix of an agent network without self-loop, $\delta>0$ is a small constant capturing the network effect strength, and $W\in\R^{m \times m}$ is the symmetric activity interaction matrix.

This tensor structure indicates that all agents share the same activity interaction kernel $W$, and all activities share the same agent interaction network $G$.
Then, the utility function becomes
\[\label{eq:utility-tensor}
	u_{i}(p,s) = \sum_{a\in\A} p_{i,a} s_{i,a} - \frac{1}{2}\sum_{a,a
		'\in\A} W_{a,a'} s_{i,a}s_{i,a'}
	+ \delta\sum_{a,a'\in\A, i'\in\X} W_{a,a'} G_{i,i'} s_{i,a} s_{i',a'}
	.\]
This reduces to the MANGs considered in \citet{chen2018Multipleactivities,kor2023Multiactivityinfluence,parasnis2024OptimalInterventions}. However, we will see later that if we consider a non-discriminatory intervention, the tensor structure will obscure the network effect in intervention.

As motivated in \cref{sec:intro,apx:micro}, it is more realistic to model agents having heterogeneous \textit{technology} $\{ W_{i} \}_{i\in\X}$.
Modifying \cref{eq:utility-tensor}, we have
\[
	u_{i}(p,s) =& \sum_{a\in\A} p_{i,a} s_{i,a} - \frac{1}{2}\sum_{a,a'\in\A} W_{i,a,a'} s_{i,a}s_{i,a'}
	+ \delta\sum_{a,a'\in\A, i'\in\X} W_{i,a,a'} G_{i,i'} s_{i,a} s_{i',a'} \\
	=& \sum_{a\in\A} p_{i,a} s_{i,a} - \sum_{a,a'\in\A}W_{i,a,a'}\left(\frac{1}{2}s_{i,a}s_{i,a'} - \delta\sum_{i'\in\X} G_{i,i'} s_{i,a} s_{i',a'}\right)
	.\]
Concisely, we have
\[\label{eq:game-kernel}
	Q = \operatorname{blkdiag}(W_{i}) ((-\tfrac{1}{2}I_{n}+ \delta G) \otimes I_{m}) 	.\]
Notably, $Q$ no longer has a tensor structure on $\X \times \A$.
We remark that in this case, even if $W_{i}$ and $G$ are symmetric, $Q$ may not be symmetric.
Plugging \cref{eq:game-kernel} into \cref{eq:ne-general} gives
\[\label{eq:ne}
	s^*(p) = (\operatorname{blkdiag}(W_{i})((I_{n}-\delta G))\otimes I_{m}))^{-1} (\mathbf{1}_{n}\otimes p)
	,\]
which is the one used in the main text and subsequent sections.

\subsection{Optimal intervention}\label{sec:intervention}

Suppose the agent space has a finite measure $|\X|=n$.
Given the closed-form expression \cref{eq:ne-general} of the Nash equilibrium, the social welfare function is also quadratic in $p$:
\[
	\int_{i\in\X} u_{i}(s^*,p)\d i =& \left< p,s^* \right>_{\X \times \A} - \frac{1}{2} \left< s^*, Q^{\mathrm{I}} s^* \right>_{\X \times \A} + \left< s^*, Q^{\mathrm{C}} s^* \right>_{\X \times \A}  \\
	=& \frac{1}{2} \left< s^*, Q^{\mathrm{I}} s^* \right>_{\X \times \A} \\
	=& \frac{1}{2} \left\langle p, (Q^{\mathrm{I}}-Q^{\mathrm{C}})^{-H} Q^{\mathrm{I}} (Q^{\mathrm{I}}-Q^{\mathrm{C}} )^{-1} p \right\rangle_{\X \times \A}
	,\]
where $\left< \cdot ,\cdot  \right>_{\X\times \A}$ is the inner product over the whole space $\X \times \A$.
Suppose the price is \textit{non-discriminatory}, i.e., all agents face the same price for each activity. With a slight abuse of notation, let $p\in\R_{\ge 0}^{\A}$ be this public price. Then,
\[
	\frac{1}{n}\int_{i\in\X} u_{i}(s^*,p)\d i
	=& \frac{1}{2} \left\langle \mathbf{1}_{\X} \otimes p, \frac{1}{n}(Q^{\mathrm{I}}-Q^{\mathrm{C}})^{-H} Q^{\mathrm{I}} (Q^{\mathrm{I}}-Q^{\mathrm{C}} )^{-1} (\mathbf{1}_{\X} \otimes p) \right\rangle_{\X \times \A} \\
	=& \frac{1}{2} \left\langle (\mathbf{1}_{\X} \otimes I_{\A}) p, \frac{1}{n}(Q^{\mathrm{I}}-Q^{\mathrm{C}})^{-H} Q^{\mathrm{I}} (Q^{\mathrm{I}}-Q^{\mathrm{C}} )^{-1}(\mathbf{1}_{\X} \otimes I_{\A}) p \right\rangle_{\X \times \A}\\
	=& \frac{1}{2} \biggl\langle p, \underbrace{\frac{1}{n}(\mathbf{1}_{\X} \otimes I_{\A})^{H} (Q^{\mathrm{I}}-Q^{\mathrm{C}})^{-H} Q^{\mathrm{I}} (Q^{\mathrm{I}}-Q^{\mathrm{C}} )^{-1}(\mathbf{1}_{\X} \otimes I_{\A})}_{\Phi} p \biggr\rangle_{\A} \label{eq:util-kernel}
	,\]
where $\mathbf{1}_{\X}$ is the identity element $\X$, $I_{\A}$ is the identity operator on $\A$, and $\mathbf{1}_{\X}\otimes I_{\A}: \A \to \X \times \A, p \mapsto \mathbf{1}_{\X} \otimes p$.
Plugging \cref{eq:game-kernel} into the expression of $\Phi$ gives \cref{eq:sw-kernel}.

The social planner proposes a non-discriminatory price $p$ to maximize social welfare with a regularization that encodes the planner's preferences $\Psi\in \mathcal{L}(\X\times \A)$ and ensures the objective is concave:
\[
	u_0 (p) = \frac{1}{2}\left< p, \Phi p \right>_{\A} - \frac{\mu}{2}\left< p-p_0, \Psi (p-p_0) \right>_{\A}
	,\]
where $p_0$ is the initial price and the decision variable $p$ is the modified price.
We assume $\Psi$ is self-adjoint and satisfies \cref{assump:reg}.
Then, the FOC gives the optimal intervention:
\[
	\frac{\mathrm{d} u_{0}}{\mathrm{d} p} = \Phi p - \mu \Psi (p-p_0) = 0 \implies
	p_* = \argmax_{p\in\R_{\ge 0}^{\A}} u_{0}(p) = (\mu \Psi - \Phi)^{-1} \mu \Psi p_0
	,\]
proving \cref{eq:opt-p}.

We see that the network effect comes into the optimal intervention through $\Phi$. If we have a tensor structure as in \cref{eq:utility-tensor}, then
\[
	\Phi =& \frac{1}{n}(\mathbf{1}_{n}^{T} \otimes I_{m}) ((I_{n}-\delta G) \otimes W)^{-T} (I_{n}\otimes W) ((I_{n}-\delta G) \otimes W)^{-1}(\mathbf{1}_{n} \otimes I_{m}) \\
	=& \frac{1}{n}(\mathbf{1}_{n}^{T}(I_{n}-\delta G)^{-T} (I_{n}-\delta G)^{-1}\mathbf{1}) \cdot (I_{m} W^{-T} W W^{-1} I_{m})\\
	=& \frac{1}{n}\underbrace{\|(I_{n}-\delta G)^{-1}\mathbf{1}\|^{2}}_{c_{G}} \cdot W^{-1}
	.\]
Note that $c_{G}$ is a constant w.r.t. the activity index, meaning that the underlying network does not affect the planner's perception of the relative importance of different activities.
Specifically, the planner's objective becomes
\[
	\argmax_{p}u_{0}(p) = \argmax_{p} \frac{1}{2n}\left< p,W^{-1}p  \right> - \frac{\mu}{2 c_{G}}\left< p-p_0, \Psi (p-p_0) \right>
	.\]
That is, different networks only differ in the regularization strength, essentially decoupling intervention from the network structure.

On the other hand, with heterogeneous activity kernels, we show in \cref{sec:voting} that the optimal intervention for each activity depends on an \textit{activity-specific network}, capturing the intricate interplay between the optimal intervention and the network structure widely observed in real-world systems.

We provide some modeling motivations for the regularization term.

\begin{itemize}\tightlist
	\item The regularization encodes the planner's budget constraint, making the objective function correspond to a hard $L_2$ budget constraint, with $\mu$ being the Lagrange multiplier.
	\item It can incorporate sustainability considerations. As $\Psi$ determines the cost of intervening different activities, it can have a large Rayleigh quotient for unsustainable activities to suppress raising prices for them, thus promoting sustainable activities and discouraging unsustainable ones.
	\item An infinitely large Rayleigh quotient in some directions also implements a \emph{partial} intervention, where the planner only has the power to intervene on some activities but not others.
	\item The intervention problem is also a principal-agent Stackelberg game, and the planner's strategy space, the price space $\R^{m}_{+}$, is the dual space of an agent's strategy space.
	      Thus, having a similar interpretation as $W_{i}$ in the agent's utility function, $\Psi$ specifies a quadratic cost in exerting intervention.
	\item The planner can also have their prior belief about agents' private technologies encoded in $\Psi$. Suppose the initial price $p_0$ is set by the planner based on this prior belief, then the objective penalizes deviating from the initial price, with $\mu$ reflecting the ratio between the confidence of the planner and the evidence revealed by the agents' actions.
\end{itemize}

\section{Motivations}\label{apx:micro}

We provide more micro-founded motivations for our setup, with a focus on heterogeneous activity kernels, network spillover, and non-discriminatory price intervention.

\subsection{Public price and private technology}

In a networked market, extensive empirical evidence suggests heterogeneity, diversification, and specialization among agents. For example, different firms have different innate abilities \citep{zhang2024fedsarsa}, R\&D \citep{knott2008RDReturns}, or institutional environments \citep{bhaumik2025Multinationalstrategy}. This motivates us to consider a market with a public price and private technologies:
$$
	u_{i}(s,p) = \lang p, s_{i} \rang - c_{i}(s_{i},s_{-i}),
$$
where the cost function $c_{i}$ captures both the agent-specific technology and the network spillover from other agents.
In our setup with multiple activities, the technology is characterized by the activity interaction matrix $W_{i}$.

\subsection{Absorptive capacity}

In a networked market, the network spillover does not only depend on external connections but also the agent's internal capabilities.
For example, in the context of R\&D, \citet{cohen1989Innovationlearning,cohen1990Absorptivecapacity} suggest that a firm’s R\&D efforts contribute to its ability to recognize, assimilate and apply relevant information from external sources to generate new knowledge.
Such ability is termed "absorptive capacity" and is crucial for firms to benefit from knowledge spillovers in the network. On the other hand, R\&D requires investment and involves costs. Additionally, R\&D costs and gains should be aligned to capture the absorptive capacity.

Our quadratic cost function models the absorptive capacity, capturing both inter‐sectoral (across activities) and geographical (across agents) R\&D spillovers \citep{belderbos2020Intersectoralinternational}.
Specifically, the network spillover term is positive, capturing beneficial absorption;
for agent $i$, we use the same activity interaction kernel $W_{i}$ in both the quadratic cost term and network compensation term, capturing the alignment between R\&D costs and gains;
we assume agents' activity interaction kernels share the same eigenbasis, capturing the standardized interface for R\&D and knowledge exchange.

Another motivating example is the phenomenon of industry clusters. For example, a labor-intensive manufacturing firm may benefit more from being located in a cluster of similar firms, with easier access to specialized labor.
Our model captures this complementarity in technology.

\subsection{Non-discriminatory subsidy/tax}

Although discriminatory price intervention is generally required for optimal intervention, we impose a non-discriminatory price intervention for the following considerations:

\begin{itemize}
	\item Fairness, reasonable, and non-discriminatory principle in pricing \citep{duckworthFRANDfoe}.
	\item Practicality of intervention process and similar efficiency compared to discriminatory intervention \citep{heijmans2023UnravelingCoordination}.
	\item Preventing strategic behaviors such as arbitrage.
\end{itemize}

In our setup, a non-discriminatory intervention \textit{aggregates} the opinions of all agents and thus demonstrates various network effects, which may play less of a role in a personalized intervention.

\subsection{Semi-decentralized federated learning}
\label{app:fl-motivation}

We describe a stylized federated-learning interpretation of the model in
Section~2. The purpose of this example is not to model a full federated-learning
protocol, but to show how the mathematical objects in our network game arise naturally from
local data geometry, public (broadcast) server-side incentives, and peer-to-peer model
exchange.

Consider \(n\) clients. Client \(i\) holds a local dataset
\[
    X_i \in \mathbb{R}^{N_i \times m}, \qquad y_i \in \mathbb{R}^{N_i},
\]
where the rows of \(X_i\) are feature vectors and \(y_i\) contains labels. The
client fits a local linear model \(s_i \in \mathbb{R}^m\). Its empirical square
loss is
\[
    \frac{1}{2N_i}\|X_i s_i-y_i\|_2^2
    =
    \frac12 s_i^\top \Sigma_i s_i - b_i^\top s_i + \kappa_i,
    \qquad
    \Sigma_i := \frac{1}{N_i}X_i^\top X_i,\quad
    b_i := \frac{1}{N_i}X_i^\top y_i ,
\]
where $\kappa_i$ and $b_i$ are constants. Thus the local feature geometry enters through the Gram matrix
\(\Sigma_i\). If the features
are centered, then \(\Sigma_i\) is the empirical covariance matrix of client
\(i\)'s features. Without centering, $\Sigma_i$ is the uncentered second-moment matrix,
which means its eigenvalues measure feature energy along the corresponding modes.

The server chooses a public intervention \(p\in\mathbb{R}^m\), common to all
clients, which can be interpreted as a reward, subsidy, or global first-order
signal assigned to feature directions. To isolate the effect of this public
intervention, we suppress the client-specific label term \(b_i\), or equivalently
absorb it into the baseline price \(p_0\) and treat fluctuations in \(b_i\) as
exogenous noise. Retaining \(b_i\) would add affine terms to the equilibrium but
would not change the role of \(W_i\) as the private local geometry that shapes
the client's response to the public signal.

Clients also exchange their local ML models over a communication graph. Let \(G_{ij}\ge 0\)
measure the influence of client \(j\)'s model on client \(i\). A semi-decentralized
federated-learning objective for client \(i\) can then be written as the utility
\[
    u_i(s,p)
    =
    p^\top s_i
    -
    \frac12 s_i^\top W_i s_i
    +
    \delta \sum_{j\neq i} G_{ij}\, s_i^\top W_i s_j .
\]
Equivalently, client \(i\) minimizes a local quadratic loss regularized by a
peer-alignment reward. Indeed,
\[
    s_i^\top W_i s_j
    =
    \frac{1}{N_i}\langle X_i s_i, X_i s_j\rangle,
\]
which means that the spillover term rewards client \(i\) when its predictions align with the
predictions of neighboring clients, evaluated on client \(i\)'s own local feature
distribution $W_i$. The matrix \(W_i\) therefore plays two roles: it penalizes local
model complexity and determines how strongly client \(i\) can absorb useful
information from neighboring models.

The first-order condition for client \(i\)'s best response is
\[
    p - W_i s_i + \delta \sum_{j\neq i} G_{ij} W_i s_j = 0,
\]
or
\[
    s_i
    =
    W_i^{-1}p
    +
    \delta \sum_{j\neq i}G_{ij}s_j .
\]
Thus the public signal \(p\) is filtered through the inverse local geometry
\(W_i^{-1}\). If \(u\) is an eigenvector of \(W_i\) with eigenvalue
\(\lambda_{i,u}\), then an isolated client responds to the component of \(p\)
along \(u\) with strength proportional to \(\lambda_{i,u}^{-1}\). In this sense,
\(\lambda_{i,u}^{-1}\) is the client's preference for that feature mode. At the
same time, the network term makes neighboring models complementary: effort by
neighboring clients in the same mode, where ``effort'' is to be interpreted as local ML models,  increases client \(i\)'s marginal utility
through the factor \(W_i\). Hence peer exchange can partially correct the
client's local geometric bias by encouraging effort in directions that would
otherwise be dampened by the local Gram matrix.

This interpretation explains the informational structure studied in the paper.
The matrices \(W_i\) encode private statistical properties of the clients' data
and are therefore not observable by the server. The graph \(G\) also may not be known to the clients, because each client communicates with a small set of neighbors.
Nevertheless, the welfare-maximizing public intervention depends jointly on
these private geometries and on the network positions of the clients. A client
may also have an incentive to misrepresent its geometry in order to steer the
public reward toward directions---and thereby, towards ML models---that benefit it. The polling algorithm developed
in the main text addresses precisely this setting: the server learns the
welfare-maximizing public signal by querying clients and walking locally through
the communication graph, without requiring centralized access to the full network
or to the raw data geometries.

\textit{An alternative interpretation of the utility $u_i(s,p)$:} For a fixed magnitude $\|s_i\|$ of the model vector $s_i$, the agent knows
that $W_i$ penalizes $s_i$ more if it places more weight on feature directions
with higher variance. This is easily verified as follows: the second loss term
decomposes as
\begin{equation}
  s_i^\top W_i s_i
  = \|s_i\|^2 \sum_{\ell=1}^{m} \lambda_\ell \,\langle \hat{s}_i, v_\ell\rangle^2,
  \qquad \hat{s}_i := \frac{s_i}{\|s_i\|},
\end{equation}
where $v_\ell$ is the unit vector along the $\ell$-th feature direction, and
$\lambda_1 \ge \lambda_2 \ge \cdots \ge \lambda_m \ge 0$ are the eigenvalues of
$W_i$ (all non-negative since $W_i$ is PSD). Under centering, $\lambda_\ell$ is
proportional to the empirical variance along $v_\ell$. Thus the second loss
term is minimized at $s_i = v_m$, the feature mode with the smallest
eigenvalue (hence the lowest variance). Therefore, this term constrains $s_i$ to place
lower weight on higher-variance feature directions. Under the linear-model
assumption, this means the agent's bias underlying the second component of the
loss is that the labels are not as sensitive to high-variance features as they
are to low-variance ones. To correct for this bias, for every neighbor $j$ of
agent $i$, the $i$-th agent reduces the loss by
\begin{equation}
  \delta G_{ij}\, s_i^\top W_i s_j = \sum_{\ell} \alpha_\ell \beta_\ell \lambda_\ell,
\end{equation}
where $\alpha_\ell := \langle s_i, v_\ell\rangle$ (respectively
$\beta_\ell := \langle s_j, v_\ell\rangle$) is the alignment between the $i$-th
agent's model (respectively the $j$-th agent's model) and the $\ell$-th feature
direction. Thus, for any fixed model norm $\|s_j\|$, the higher the variance
$\lambda_\ell$ and the greater the value of $\beta_\ell$, which denotes the $j$-th agent's perceived
sensitivity of the labels to the $\ell$-th feature, the more strongly this
network-effect term increases $\alpha_\ell$, thereby correcting the weight that
model $s_i$ places on feature $\ell$ (which would have been smaller otherwise).

\section{Derivations of solution characterizations}\label{sec:calc}

\subsection{Bundle-induced networks}\label{apx:cen}

We formally present the proposition under the simplification in \cref{sec:voting}.

\begin{proposition}\label{prop:cen}
	The eigenvectors of $\Phi$ are the columns of $U$. Moreover, $\Phi$ decomposes as $\Phi = U \Lambda_{\Phi} U^{T}$, where $\Lambda_{\Phi}$ is diagonal and its $a$-th diagonal element is given by
	\[\label{eq:cen}
		\phi_{a} = \frac{1}{n}\sum_{i\in\X} \tilde{\cen}_{a,i}^{2} \cdot \lambda_{i,a}^{-1}
		,\]
	where $\tilde{\cen}_{a,i}$, which is the $i$-th entry of $\tilde{\cen}_{a}\coloneqq (I - \delta\tilde{G}_{a})^{-1} \bm{1}_{n}$, is the Katz--Bonacich centrality of agent $i$ in the interaction network induced by the activity bundle $a\in [m]$.
\end{proposition}

\begin{proof}

	Let $K = (I-\delta G)^{-1}$. Note that $K \bm{1}$ is the Katz--Bonacich centrality vector of the agent network $G$.\footnote{Our notation specifies $G_{ij}=1$ if $y$ influences $i$, and thus $G$ is the transpose of the usual convention in the network literature. As a result, the Katz--Bonacich centrality here is higher for agents that receive more influence from others, which is consistent with the convention.}
	By \cref{eq:sw-kernel}, when agents share the same eigenbasis $U$ of their activity kernels $W_{i} = U\Lambda_{i}U^{T}$, $\Phi$ has the following transformation:
	\[
		n\Phi =& (W_1^{-1},\dots,W_n^{-1}) (K^{T} \otimes I_{m}) W (K \otimes I_{m})W^{-1} (\mathbf{1}_{n}\otimes I_{m}) \\
		=& U(\Lambda_{1}^{-1} U^T, \dots, \Lambda _{n}^{-1}U^{T}) (K^{T} \otimes I_{m}) W (K \otimes I_{m})W^{-1} (\mathbf{1}_{n}\otimes I_{m}) \\
		=& \textstyle U\left(  \sum_{j}K_{1j}\Lambda_{j}^{-1} U^T, \dots,\sum_{j}K_{nj} \Lambda _{j}^{-1}U^{T}  \right)  W (K \otimes I_{m})W^{-1} (\mathbf{1}_{n}\otimes I_{m}) \\
		=& \textstyle U\left(  \sum_{j}K_{1j}\Lambda_{j}^{-1}\Lambda _{1} U^T, \dots,\sum_{j}K_{nj} \Lambda _{j}^{-1}\Lambda _{n}U^{T}  \right)(K \otimes I_{m})W^{-1} (\mathbf{1}_{n}\otimes I_{m}) \\
		=& \textstyle U\left(  \sum_{i}\sum_{j}K_{i1}K_{ij}\Lambda_{j}^{-1}\Lambda _{i} U^T, \dots,\sum_{i}\sum_{j}K_{in}K_{ij} \Lambda _{j}^{-1}\Lambda _{i}U^{T}  \right)W^{-1} (\mathbf{1}_{n}\otimes I_{m}) \\
		=& \textstyle U  (\sum_{i,j,k}K_{ik}K_{ij}\Lambda_{j}^{-1}\Lambda _{i}\Lambda _{k}^{-1} )U^T \label{eq:phi}
		.\]
	Thus, $\Phi$ is also diagonalizable in the eigenbasis $U$; and we write $\Phi = U \Lambda_{\Phi} U^{T}$ with its $a$-th diagonal element given by
	\[
		\phi_{a} = \frac{1}{n}\sum_{i,j,k\in\X} K_{ij}K_{ik}\lambda_{j,a}^{-1}\lambda_{i,a}\lambda_{k,a}^{-1}, \quad \forall a=1,2,\dots,m
		.\]
	Define the vector $\lambda_{a}=(\lambda_{1,a},\dots,\lambda_{n,a})^{T}\in\R^{n}$ and $\lambda_{a}^{-1}$ as its element-wise inverse.
	Then, we can rewrite $\phi_{a}$ as
	\[\label{eq:cen-1}
		\phi_{a} = \frac{1}{n}\left< (K\lambda_{a}^{-1})\odot \lambda_{a}, \operatorname{diag}(\lambda_{a}^{-1} ) ((K\lambda_{a}^{-1} )\odot \lambda_{a}) \right>
		.\]

	Recall our construction of the bundle-specific network:
	\[
		\tilde{G}_{a} = \Lambda_{a} G \Lambda_{a}^{-1}
		.\]
	Its Katz--Bonacich centrality satisfies
	\[
		\tilde{\cen}_{a} =& (I-\delta \tilde{G}_{a})^{-1} \mathbf{1}_{n} \\
		=& (I - \delta \Lambda_{a}G \Lambda_{a}^{-1} )^{-1} \mathbf{1}_{n} \\
		=& \Lambda_{a}(I - \delta G  )^{-1} \Lambda_{a}^{-1} \mathbf{1}_{n} \\
		=& \Lambda_{a}K \lambda_{a}^{-1} \\
		=& K \lambda_{a}^{-1} \odot \lambda_{a}\label{eq:cen-2}
		.\]
	Combining \cref{eq:cen-1,eq:cen-2} gives the desired expression.
\end{proof}

\subsection{Planner-informed second-stage communication}\label{sec:planner}

This section discusses the path of generalizing \cref{prop:cen} to the more general \cref{prop:cen-planner}.
The simplifications in \cref{sec:voting} have two implications:
\begin{enumerate*}[label=(\roman*)]
	\item $\{ W_{i} \}_{i\in \X}$ share the same eigenbasis, so we can define bundle-specific networks, and
	\item $\Psi$ shares the same eigenbasis with $W_{i}$, so we directly get the optimal intervention for each activity bundle from the corresponding eigenvalue of $\Phi$ by $[U^{T}p_*]_{a} = (1 - \mu^{-1} \phi_{a})^{-1} [U^{T}p_0]_{a}$.
\end{enumerate*}
In words, different activity bundles induce different networks to determine the optimal intervention for the bundle independently with no communication between bundle-induced networks.
This is not possible when the agents and the planner do not share the same set of activity bundles.

We now model $\Psi$ as an arbitrary positive definite matrix that may not share the same eigenbasis as $\{ W_{i} \}_{i\in\X}$.
Similarly, we can interpret $\Psi$ as encoding the planner's \emph{preferences} across activities.
A general planner preference introduces a second stage of communication between bundle-induced networks.
Specifically, under \cref{assump:reg}, we can expand the optimal intervention as Neumann series:
\[\label{eq:opt-planner}
	p_* = \sum_{t=0}^{\infty} \mu ^{-t} (\Psi^{-1} \Phi)^{t} p_{0},
\]
which can be interpreted as the result of iterative \textit{planner-mediated communication between bundle-specific networks}, with the higher-level communication network being $\Psi^{-1}\Phi$.
This communication network becomes diagonal, once $\Psi$ shares the same eigenbasis as $\{ W_{i} \}$ and we apply an eigendecomposition to $\Phi$, reducing to the case of no communication between bundle-induced networks.
With a general $\Psi$, a activity bundle of the planner may partially align with multiple activity bundles of the agents, thus prompting communication between bundle-induced networks.
Intuitively, when the planner proposes a price change for one of its activity bundles, multiple agent activity bundles may be affected, and thus requiring different bundle-induced networks to coordinate to achieve the optimal intervention.
Formally, after being transformed by \(\Psi^{-1}\), agents' activity bundles are no longer orthogonal.

Combined with \cref{prop:cen-planner}, the intervention problem corresponds to a \emph{two-stage communication} scheme: agents first communicate within bundle-induced networks to form their
votes, and then the planner mediates the communication between bundle-specific networks to form the final intervention.

We further model the agent's activity kernels $\{ W_{i} \}_{i\in\X}$ as arbitrary positive definite matrices that may not share the same eigenbasis,
which invalidates bundle-induced networks in \cref{prop:cen} as agents no longer have a common set of activity bundles.
If we generalize the definition of a network to allow for matrix-valued edge weights, we can consider a single generalized network that also incorporates the planner's preference, similar to how an bundle-induced network incorporates agents' preferences.
By \cref{eq:opt-p}, we have \(\Psi ^{1 /2}p_* = (I - \mu^{-1} \tilde{\Phi})^{-1} (\Psi ^{1 /2}p_{0})\), where \(\tilde{\Phi} \coloneqq \Psi ^{-1 /2}\Phi \Psi ^{-1 /2}\in\R^{m\times m}\) is characterized in \cref{prop:cen-planner}.

We note that $\tilde{G}$ is \emph{planner-informed}, as it is modified by the planner's preference in addition to those of the agents.
If \(\Psi\) and $W_{i}$ share the same eigenbasis $U$, then \(\{\tilde{W}_{i}\}_{i\in\X}\) are simultaneously diagonalizable, and we can split \(\tilde{G}\) into \(m\) bundle-induced networks with scalar edge weights, reducing to \cref{prop:cen}. 
\subsection{Proof of \texorpdfstring{\cref{prop:cen-planner}}{Proposition 1}}

Recall that $\tilde{\Phi} = \Psi^{-1 /2}\Phi \Psi^{-1 /2}$.
Following \cref{eq:phi}, we have
\[
	\tilde{\Phi} = \textstyle \Psi^{-1 /2}  (\sum_{i,j,k}K_{ik}K_{ij}W_{j}^{-1}W _{i}W _{k}^{-1} )\Psi^{-1 /2}
	.\]
Our centrality measure has a ``dimensionlessness'' structure, which gives
\[
	\tilde{\Phi} = \textstyle \sum_{i,j,z}K_{ik}K_{ij} (\Psi^{-1 /2}W _{j}^{-1} \Psi^{-1 /2})(\Psi^{1 /2}W _{i}\Psi^{1 /2})(\Psi^{-1 /2} W _{k}^{-1} \Psi^{-1 /2}).
\]
Recall \(\tilde{W} _{i} = \Psi^{1/2}W_{i}\Psi^{1/2}\).
Then, following \cref{eq:cen-2}, we get
\[
	\tilde{\Phi} = \sum_{i,j,k}K_{ik}K_{ij} \tilde{W} _{j}^{-1}\tilde{W} _{i}\tilde{W} _{k}^{-1}
	= \sum_{i \in \X} \widetilde{\cen}_{i}^T \tilde{W}^{-1} _{i} \widetilde{\cen}_{i},
\]
where \(\widetilde{\cen}_{i}\) is the \(i\)-th block of the
Katz--Bonacich centrality of the modified network \(\tilde{G}\) with
matrix-valued edge weight $\tilde{G}_{ij} = G_{ij} \tilde{W} _{i}\tilde{W} _{j}^{-1}$.

\subsection{Proof of \texorpdfstring{\cref{prop:ce}}{Proposition 2}}

Following the derivation in \cref{apx:setup}, the utility of agent $i$ at the Nash equilibrium at price $p$ is
\[
	u_{i}(s^*(p)) =& \langle s^*_{i}(p), ( p + Qs^*(p) )_{i} \rangle\\
	=& \langle s^*_{i}(p), ( (Q^{\mathrm{I}}-Q^{\mathrm{C}})s^*(p)  + (-\tfrac{1}{2}Q^{\mathrm{I}}+Q^{\mathrm{C}})s^*(p) )_{i} \rangle\\
	=& \tfrac{1}{2}\langle s^*_{i}(p) , (Q^{\mathrm{I}} s^*(p))_{i} \rangle
	.\]
Note that $Q^{\mathrm{I}} = \operatorname{blkdiag}(W_{i})$ and $Q^{\mathrm{I}}-Q^{\mathrm{C}} = \operatorname{blkdiag}(W_{i})((I_{n}-\delta G)\otimes  I_{m})$. Recall that $K = (I-\delta G)^{-1}$.
Therefore,
\[
	u_{i}(s^*(p)) =& \tfrac{1}{2}\langle s^*_{i}(p) , W_{i} s^*_{i}(p)) \rangle \\
	=& \frac{1}{2} \langle p, \underbrace{\sum_{j,k}K_{ij}K_{ik}W_{j}^{-1}W_{i}W_{k}^{-1} }_{\Phi_{i}}p\rangle
	.\]
Then, when all agents share the same eigenbasis of activity kernel ($W_{i} = U \Lambda_{i} U^{T}$ for all $i\in\X$), we have
\[
	\Phi_{i} =& U\diag (\tilde{\cen}_{a,i}^{2}\lambda_{i,a}^{-1})U^{T}
	.\]
Thus, the $a$-th eigenvalue of $\Phi_{i}$ is $\phi_{i,a} = \tilde{\cen}_{a,i}^{2}\lambda_{i,a}^{-1}$.
We see that agent $i$'s utility kernel $\Phi_{i}$ has a similar form to the social welfare kernel $\Phi$ in \cref{prop:cen}, which is not a surprise since the social welfare is simply the aggregate of all agents' utilities.

Fixing a price, agents' equilibrium utilities are linear and monotonic in their utility kernel w.r.t. the Loewner order.
Therefore, we only need to know how $\phi_{i,a}$ changes w.r.t. the preference.
We have
\[\label{eq:ce-1}\ds
	\frac{\partial \ln \phi_{i,a}}{\partial \ln \lambda_{i,a}}
	= 2\frac{\partial \ln \tilde{\cen}_{a,i}}{\partial \ln \lambda_{i,a}} - 1
		.\]
Further,
\[\label{eq:ce-2}\ds
	\frac{\partial \ln \tilde{\cen}_{a,i}}{\partial \ln \lambda_{i,a}}
	= \frac{\lambda_{i,a}}{\tilde{\cen}_{a,i}} \frac{\partial \tilde{\cen}_{a,i}}{\partial \lambda_{i,a}}
	= \frac{\sum_{j\in\X} K_{ij} \frac{\lambda_{i,a}}{\lambda_{j,a}} - K_{ii} \frac{\lambda_{i,a}}{\lambda_{i,a}}}{\sum_{j\in\X}K_{ij}\frac{\lambda_{i,a}}{\lambda_{j,a}}}
	= 1 - \frac{K_{i i}}{\tilde{\cen}_{a,i}}
	,\]
where $K = (I-\delta G)^{-1}$ is the Katz--Bonacich centrality matrix of the original network.
Note that
\[
	K_{ij} \frac{\lambda_{i,a}}{\lambda_{j,a}} = (\diag(\lambda_{a})) K \diag(\lambda_{a}^{-1}))_{ij} = ((I - \delta \diag(\lambda_{a}) G \diag(\lambda_{a}^{-1}))^{-1})_{ij} = ((I - \delta \tilde{G}_{a})^{-1})_{ij} = \tilde{K}_{a,ij}.
\]
Thus, combining \cref{eq:ce-1,eq:ce-2} gives the desired result.

\section{Centrality simulations}\label{sec:exp}

This section provides a concrete example to illustrate the theoretical results derived in \cref{sec:voting}.
We consider a network of $n=7$ data centers. All centers engage in $m=2$ activities---\emph{compute} ($a_{1}$) and \emph{storage} ($a_{2}$)---and two activity patterns: \emph{archive} ($u_{1}=(1,1)^{T}/\sqrt{2}$) and \emph{retrieve} ($u_{2}=(1,-1)^{T}/\sqrt{2}$).
Specifically, the archival pattern means the center produces one unit of data and stores it; the retrieval pattern means the center retrieves one unit of data from storage and processes it.
The resulting activity interaction kernel basis is $U = \frac{1}{\sqrt{2}}\bigl(\begin{smallmatrix}
			1 & 1 \\
			1 & -1
		\end{smallmatrix}\bigr)$.

We categorize the data centers into three distinct tiers based on their role in the data lifecycle:

\begin{enumerate}
	\item \textbf{Provider-Side centers (Centers 1--2):} Center 1 is a cloud storage hub focusing on the archival pattern ($\lambda_{1,1} = 1, \lambda_{1,2} = 10$). Center 2 is a data processing hub focusing on the retrieval pattern ($\lambda_{2,1} = 10, \lambda_{2,2} = 1$). Both act as sources of influence in the network.
	\item \textbf{Intermediary centers (Centers 3--5):} These centers act as neutral bridges, balancing equally between archival and retrieval tasks. They possess moderate preferences for both patterns: $\lambda_{x,1} = 5, \lambda_{x,2} = 5$ for $x \in \{3,4,5\}$.
	\item \textbf{Client-Side centers (Centers 6--7):} These centers are at the edge of the network and are highly specialized. Center 6 serves user queries (high preference for retrieval: $\lambda_{6,1} = 10, \lambda_{6,2} = 1$). Center 7 serves data backups (high preference for archival: $\lambda_{7,1} = 1, \lambda_{7,2} = 20$).
\end{enumerate}

The network structure is given by the following attenuated adjacency matrix $G$, representing a directed cascade from providers to clients:
\[
	G = 0.1 \cdot \begin{pmatrix}
		0 & 0 & 0 & 0 & 0 & 0 & 0 \\
		0 & 0 & 0 & 0 & 0 & 0 & 0 \\
		1 & 1 & 0 & 0 & 0 & 0 & 0 \\
		1 & 1 & 1 & 0 & 0 & 0 & 0 \\
		1 & 1 & 1 & 1 & 0 & 0 & 0 \\
		0 & 0 & 1 & 1 & 1 & 0 & 0 \\
		0 & 0 & 1 & 1 & 1 & 0 & 0
	\end{pmatrix}
	.\]

We set the initial price $p_{0}=(2,0)^{T}$ and consider the optimal non-discriminatory intervention under an $L_{2}$ constraint with budget $B=1$.

\begin{figure}[th]
	\centering
	\includegraphics[width=0.9\textwidth]{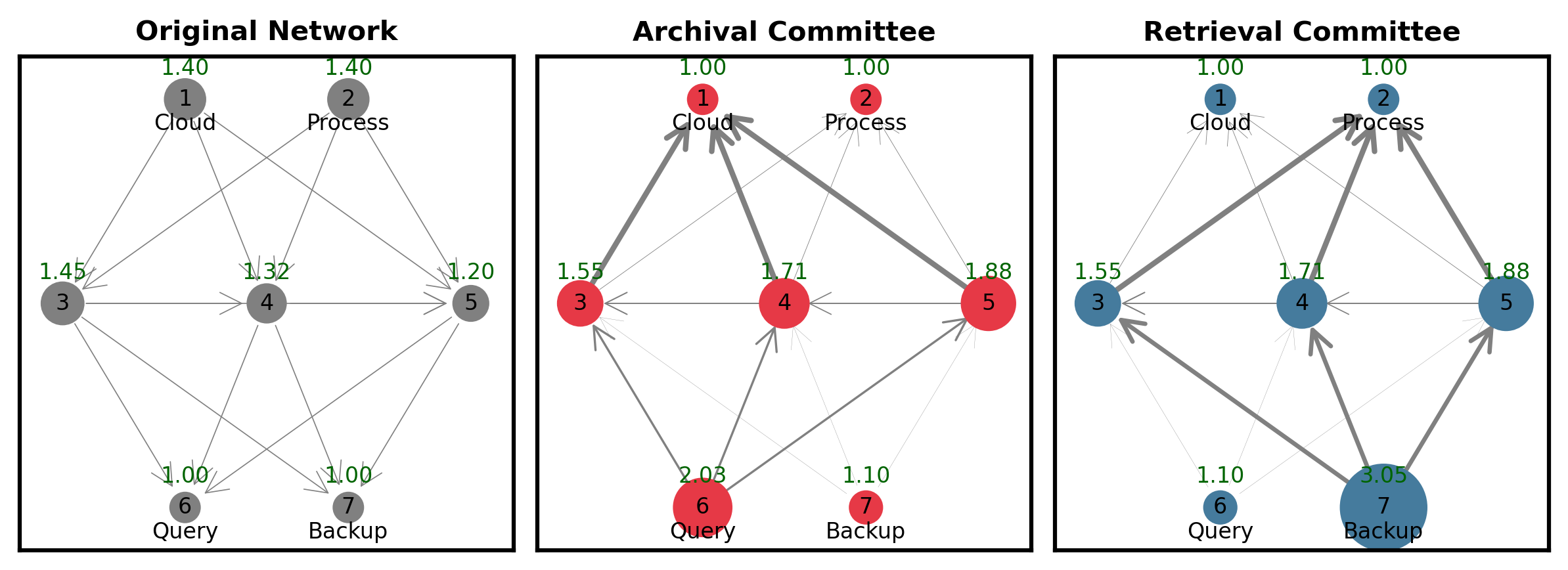}
	\caption{Original network and activity-specific networks for archival and retrieval patterns. Edge width and node size are proportional to edge weight and squared centrality, respectively.}
	\label{fig:network}
\end{figure}

We visualize the original network together with activity-specific networks in \cref{fig:network}. They verify our theory that while in the original network nodes on the provider side have higher influence and are more central, in the activity-specific networks, the most central nodes are on the client side, which absorb most of the influence while casting the largest votes in the committee.
Specifically, in the committee for archival activities, the query center (6), which has a low preference for archival activities, is the most central member; in the committee for retrieval activities, the backup center (7), which has a low preference in retrieval activities, is the most central.

\begin{figure}[th]
	\centering
	\includegraphics[width=0.7\textwidth]{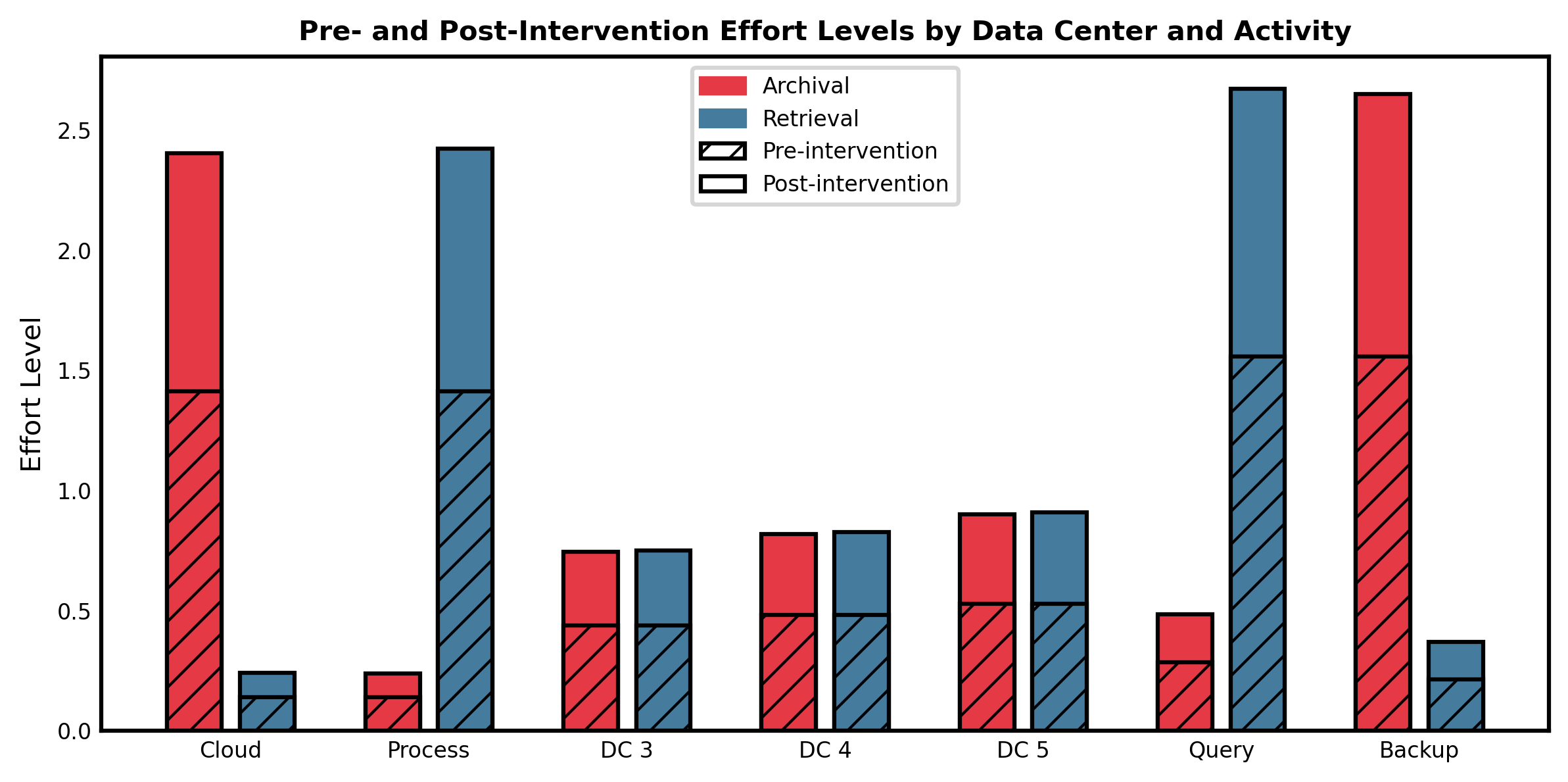}
	\caption{Pre- and post-intervention effort levels across data centers and activities.}
	\label{fig:effort}
\end{figure}

We plot the pre- and post-intervention effort levels in \cref{fig:effort} to further validate our analysis. We see that the cloud (1) and backup (7) data centers indeed prefer, both before and after the intervention, archival over retrieval activities; while the processing (2) and query (6) centers consistently prefer retrieval.
Additionally, the overall network has a slight bias towards archival activities (due to the large $\lambda_{7,2}$), but the intervention spends relatively more on retrieval activities, evident by the effort increases from the neutral centers (3, 4, 5) towards retrieval.
This confirms our counter-intuitive insight that the optimal intervention allocates more resources to the less preferred activity pattern to empower reluctant members who hold significant influence in that activity's committee.

\begin{figure}[th]
	\centering
	\includegraphics[width=0.7\textwidth]{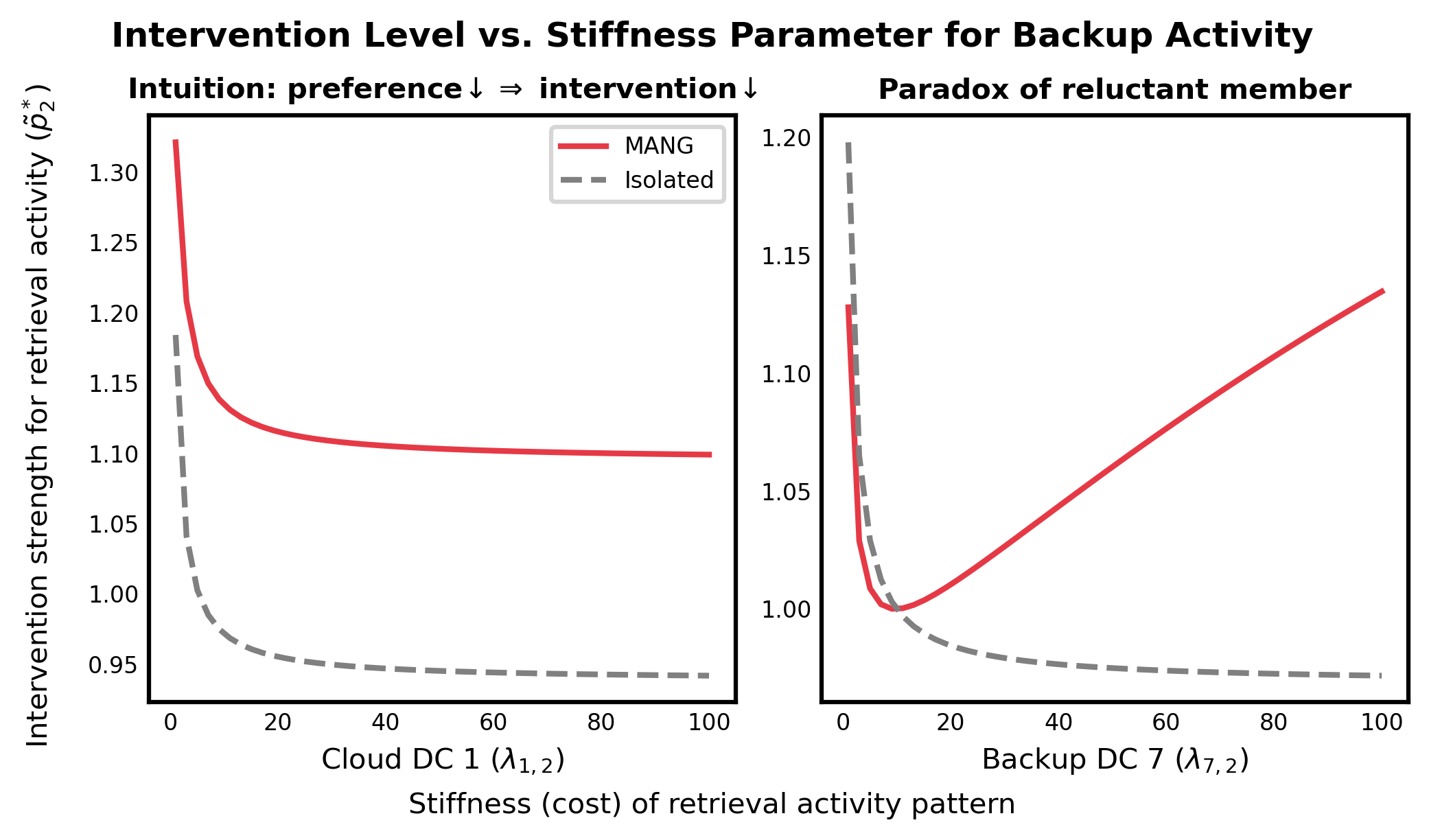}
	\caption{Optimal intervention level vs. stiffness for retrieval activities at cloud center (1) and backup center (7). The dashed lines represent intervention levels when centers are isolated (no network effects).}
	\label{fig:paradox}
\end{figure}

Finally, we further illustrate this ``paradox'' of empowering reluctant members and intervening less preferred activities in \cref{fig:paradox} by inspecting two key data centers: the cloud storage center (1) and the backup center (7). We vary their stiffness parameters for retrieval activities ($\lambda_{1,2}$ and $\lambda_{7,2}$) and plot the resulting optimal intervention level for retrieval activities. We also plot the intervention levels when the centers are isolated (i.e., no network effects) for reference.
We see that, changes in the stiffness of the cloud center (1) behave intuitively: as it becomes more reluctant towards retrieval activities (higher $\lambda_{1,2}$), the intervention level for retrieval activities decreases, as the planner focuses on other activities with higher participation.
However, the backup center (7) exhibits the paradoxical behavior: as it becomes more reluctant towards retrieval activities (higher $\lambda_{7,2}$), the intervention level for retrieval activities increases. This confirms our insight that in a networked setting, optimal interventions may need to counteract individual preferences to effectively leverage the influence dynamics within the network.

\section{Implementation details of \poll}\label{apx:poll}

\paragraph{Random walks}
We provide a more detailed derivation of Step 4 in \cref{sec:poll}:
\[
	[\tilde{G}^{k}\mathbf{1}]_{i} = & \sum_{j_{1}\in \X_{i}} \tilde{G}_{ij_{1}} [\tilde{G}^{k-1}\mathbf{1}]_{j_{1}}                                                                                                                                         \\
	=                                 & \deg_{i} \mathbb{E}_{j_{1} \sim \operatorname{Unif}(\X_i)} [\tilde{G}_{ij_{1}} [\tilde{G}^{k-1}\mathbf{1}]_{j_{1}}]                                                                                                   \\
	=                                 & \deg_i \mathbb{E}_{j_{1} \sim \operatorname{Unif}(\X_i)} [\tilde{G}_{ij_{1}} \deg_{j_1}\mathbb{E}_{j_{2} \sim \operatorname{Unif}(\X_{j_{1}})} [\tilde{G}_{j_{1}j_{2}} [\tilde{G}^{k-2}\mathbf{1}]_{j_{2}}]] \\
	=                                 & \mathbb{E}_{j_{1} \sim \operatorname{Unif}(\X_i), j_{2} \sim \operatorname{Unif}(\X(j_{1}))} [\deg_i\deg_{j_{1}} \cdot \tilde{G}_{ij_{1}} \tilde{G}_{j_{1}j_{2}} [\tilde{G}^{k-2}\mathbf{1}]_{j_{2}}]        \\
	=                                 & \cdots                                                                                                                                                                                                                         \\
	=                                 & \underset{(i=j_{0}, j_{1}, \ldots, j_{k})}{\mathbb{E}} \left[  \prod_{l=0}^{k-1} \deg_{j_{l}} \cdot \prod_{l=0}^{k-1} \tilde{G}_{j_{l}j_{l+1}} \right], \label{eq:walk}
\]
Please refer to \cref{apx:is} for a more general sampling scheme, which also accommodates weighted graphs.
On a random walk, when an intermediate agent has no incoming neighbor, we can simply stop the walk and return a zero estimate, as the degree product will be zero.

\paragraph{Knowledge}
There is a fundamental trade-off between the planner's knowledge and the agents' behavior: the more the planner knows, the less it relies on agents to behave faithfully, and vice versa. For example, if the planner knows the network topology, agents need not report their degrees.
In \poll, we shift this burden to the agents. The planner has no prior knowledge of the network or agents' technologies, making the algorithm fully model-free. Each agent knows only its local neighborhood; agents need their neighbors' degrees only when performing general importance sampling in \cref{apx:is}.
Agents' technologies are private; the initial price $p_0$ and the corresponding Nash equilibrium $s^*(p_0)$ are public.

\begin{algorithm}[th]
	\caption{\texttt{Poll}}
	\label{alg}
	\Input{Initial price $q_0=\Psi^{1 /2} p_0$, game parameters $n=|\X|$, $\delta$, $\nu = \mu^{-1}$}\;
	\Para{Step sizes $\{\alpha_t\}_{t \geq 0}$}\;
	Initialize $q \leftarrow \mathbf{0}$ (or some initial guess)\;
	\For{$t = 0, 1, 2, \dots$}{
		Sample an agent $i$ uniformly from $\X$\;
		Sample two trajectory lengths $k^{(1)}, k^{(2)}$ independently from $\operatorname{Geom}(\delta)$\;
		\For{$x \in \{1, 2\}$}{
			Initialize a random walk at $j^{(x)}_0 = i$\;
			Initialize $d_x \leftarrow 1$\;
			\For{$l = 0$ \KwTo $k^{(x)}-1$}{
				$d_x \leftarrow d_x \cdot \deg_{j^{(x)}_l}$\tcp*{by the planner}
				Sample $j^{(x)}_{l+1}$ uniformly from the incoming neighbors $\X_{j^{(x)}_l}$\tcp*{by the agents}
			}
		}
		Compute stochastic estimate: $\hat{\Phi}_t \leftarrow (1-\delta)^{-2} d_1 d_2 \tilde{W}_{j^{(1)}_{k^{(1)}}}^{-1} \tilde{W}_{i} \tilde{W}_{j^{(2)}_{k^{(2)}}}^{-1}$\tcp*{by the agents}
		Update intervention: $q \leftarrow (1 - \alpha_t) q + \alpha_t (q_0 +\nu \hat{\Phi}_t q )$\;
	}
	\Return $q$ or $p = \Psi^{-1 /2}q$\;
\end{algorithm}

\paragraph{Behavior}
We give a detailed behavioral description and timeline at each time step, expanding \cref{alg}:
\begin{enumerate}
	\item Planner: announces the update rule \poll;
	\item Planner: samples an agent $i$;
	\item Agent $i$: initiates two random walks;
	\item Agent along the walks: report its degree to the planner\footnote{The planner can also ask each walk's endpoint to report the degree product directly, in which case agents along the walk collect and multiply degrees before passing the message onward.} and choose the next step following the algorithm's instructions;
	\item Planner: collects the degree information along the random walks;
	\item Planner: sends the query $q_{t}$ to the endpoint agent $j^{(2)}_{k}$;
	\item Agent $j^{(2)}_{k}$: computes $\tilde{W}_{j^{(2)}_{k}}^{-1} q_t$ and sends it back to agent $i$ through the walk;
	\item Planner: signals the degree information of the walks to agent $i$;
	\item Agent $i$: left-multiplies the received vector with $\tilde{W}_{i}$ and sends it to the other endpoint agent $j^{(1)}_{k}$ through the other walk;
	\item Agent $j^{(1)}_{k}$: left-multiplies the received vector with $\tilde{W}_{j^{(1)}_{k}}^{-1}$ and sends it back to the planner.
\end{enumerate}
The updated price follows deterministically from the announced algorithm and the agents' responses, so no further action is required from the planner after receiving the final report.
This changes if the planner employs a payment mechanism to enforce truthful reporting, where the payment depends on the final report; see \cref{sec:groves}.

Variations of \poll, such as those with control variates (\cref{apx:cv}), use the same communication protocol.
In this work, we analyze strategic misreporting of $W_i$ only, assuming degrees are reported faithfully and walks are executed as specified; the analysis of broader strategic behaviors is left to future work.

\paragraph{Complexity} At each time step, each random walk requires at most two rounds of message passing---one forward traversal for degree collection and one backward pass for the computation result---with each message containing a scalar degree or a vector of length $m$, giving an expected communication complexity of $O(2(m+1)(1-\delta)^{-1} + m) = O(m)$. Three matrix-vector multiplications are performed at the endpoint agents, giving a computational complexity of $O(m^2)$.

\paragraph{Privacy} The planner only receives the final preference-price product $\tilde{W}_{j^{(1)}_{k} }^{-1} \tilde{W}_{i} \tilde{W}_{j^{(2)}_{k} }^{-1}q_{t}$.
Since agent $i$ only receives $\tilde{W}_{j^{(2)}_{k} }^{-1}q_{t}$ with the identity of agent $j^{(2)}_{k}$ masked, the former cannot infer the preference $W_{j^{(2)}_{k}}$. Similarly, agent $j^{(1)}_{k}$ cannot infer the preference $W_{i}$ or $W_{j^{(2)}_{k}}$ from its received vector $\tilde{W}_{i} \tilde{W}_{j^{(2)}_{k} }^{-1}q_{t}$.
Finally, since the agents' cost structures may not share the same eigenbasis, the planner cannot infer agents' preferences from the final report, even if the identities of the endpoint agents are known to the planner.

\section{Proofs of learning efficiency}\label{apx:eff-learn}

Recall the notation:
$\bar{\psi}\coloneqq\lambda_{\max}(\Psi)$, $\ubar{\psi}\coloneqq\lambda_{\min}(\Psi)$, $\overline{\lambda} \coloneqq \max_{i\in\X}\lambda _{\max}(W_{i})$, $\underline{\lambda}\coloneqq \min_{i\in\X}\lambda _{\min}(W_{i})$,
$\bar{d}\coloneqq \max_{i\in\X}\deg_{i}$,
$\kappa_{\Psi}\coloneqq\bar{\psi}/\ubar{\psi}$, $\kappa_{W}\coloneqq \overline{\lambda}/\underline{\lambda}$, $\kappa_{\nu}\coloneqq (1-\nu \lambda_{\max}(\tilde{\Phi}))^{-1}$, $\kappa_{\delta}\coloneqq (1-\delta \sigma_{\max}(G))^{-1}$.
Let $e_{t} = p_{t}-p_{*}$ and $\mathcal{F}_{t}$ be the $\sigma$-field generated by the randomness up to time $t$.

\begin{lemma}[Error dynamics]\label{lem:err-dyn}
	If we have a uniform bound $\nu^{2}\mathbb{E}\|\hat{\Phi}_{t}-\tilde{\Phi}\|^{2} \le \sigma ^{2}$, setting $\alpha = \kappa_{\nu}\ln t /t$ and letting $t /\ln t \ge \max\left\{ 2, 4\kappa_{\nu}^{2}\sigma^{2} \right\}$ gives
	\[
		\mathbb{E}\|e_{t}\|^{2} \le ( 1-\alpha /\kappa_{\nu})^{t}\|e_{0}\|^{2} + 2\alpha \kappa_{\nu} \sigma ^{2}\|p_{*}\|^{2} = \tilde{O}\left( t^{-1} \kappa_{\nu}^{2}\sigma^{2}\|p_{*}\|^{2} \right).
	\]
\end{lemma}

\begin{proof}
	Note that the optimal solution satisfies $(I-\nu \tilde{\Phi})p_{*}=p_{0}$.
	We have the following error dynamics
	\[
		e_{t+1} =& e_{t} - \alpha ((I-\nu \hat{\Phi}_{t})p_{t} - p_{0})\\
		=& e_{t} - \alpha (p_0 - (I-\nu \tilde{\Phi})p_{*}  + (I-\nu \hat{\Phi}_{t})p_{t}  - p_{0})\\
		=& e_{t} - \alpha ((I-\nu \tilde{\Phi})(p_{t}-p_{*}) -(I-\nu \tilde{\Phi})p_{t}  + (I-\nu \hat{\Phi}_{t})p_{t})\\
		=& (I - \alpha (I-\nu \tilde{\Phi})) e_{t} + \alpha \nu (\hat{\Phi}_{t}-\tilde{\Phi})p_t, \label{eq:err-1}
	\]
	which gives the mean squared error (MSE) dynamics
	\[
		\mathbb{E}\|e_{t+1}\|^{2} =& \mathbb{E}\|(I-\alpha (I-\nu \tilde{\Phi})) e_{t}\|^{2} + \alpha^{2} \nu^{2} \mathbb{E}\|(\hat{\Phi}_{t}-\tilde{\Phi})p_t\|^{2}  \label{eq:err-2}\\
		\le& \|I-\alpha (I-\nu \tilde{\Phi})\|^{2}\mathbb{E}\|e_{t}\|^{2} + \alpha^{2}\nu^{2}\mathbb{E}[\mathbb{E}[\|\hat{\Phi}_{t}-\tilde{\Phi}\|^{2}\given \mathcal{F}_{t}]\|p_{t}\|^{2}],
	\]
	where $\mathcal{F}_{t}$ be the filtration generated by the randomness up to time $t$.
	Since $\hat{\Phi}_{t}$ is independent of $\mathcal{F}_t$, we have $\mathbb{E}[\|\hat{\Phi}_{t}-\tilde{\Phi}\|^{2}\given \mathcal{F}_{t}] = \mathbb{E}\|\hat{\Phi}_{t}-\tilde{\Phi}\|^{2}$.
	Then, given the uniform bound,
	\[
		\mathbb{E}\|e_{t+1}\|^{2}\le& \|I-\alpha (I-\nu \tilde{\Phi})\|^{2}\mathbb{E}\|e_{t}\|^{2} + 2\alpha^{2}\sigma^{2}(\mathbb{E}\|p_{t}-p_{*}\|^{2} + \|p_{*}\|^{2}) \\
		\le& (1-2\alpha\kappa_{\nu}^{-1}+\alpha^{2}(\kappa_{\nu}^{-2} + 2\sigma^{2}))\mathbb{E} \|e_{t}\|^{2} + 2\alpha^{2}\sigma ^{2}\|p_{*}\|^{2}. \label{eq:err-3}
	\]
	Setting $\alpha \le \min\{ \kappa_{\nu}/2, 1/(4\kappa_{\nu}\sigma^{2}) \}$ gives
	\[
		\mathbb{E}\|e_{t+1}\|^{2} \le (1-\alpha \kappa_{\nu}^{-1}) \mathbb{E}\|e_t\|^{2} + 2\alpha^{2}\sigma^{2}\|p_{*}\|^{2}
		.\]
	Setting $t$ accordingly gives the desired result.
\end{proof}

As a corollary, we get:

\begin{proposition}\label{prop:noise}
	Consider \poll update rule \cref{eq:poll-1} with additive noise:
	\[
		p_{t+1} = p_t - \alpha_t ((I - \nu \hat{\Phi}_t) p_t - p_0 + \epsilon^{(1)}_t + \epsilon^{(2)}_{t})
		,\]
	where $\mathbb{E}_{Z_t}[\epsilon^{(1)}_{t}] = 0$, $\mathbb{E}\|\epsilon^{(1)}_t| \le \sigma_1^{2}$, and $\|\epsilon^{(2)} _{t}\|\le \sigma_{2}$ almost surely.
	Then, with a constant step size $\alpha \equiv \kappa_{\nu}\ln t / t$, we have
	\[
		\mathbb{E}\|p_{t} - p_{*}\|^{2} = \tilde{O}(t^{-1} \kappa_{\nu}^{2}(\sigma^{2}\|p_{*}\|^{2} + \sigma_1^{2}+\sigma_2^{2}) + \kappa_{\nu}^{2}\sigma_2^{2})
		.\]
\end{proposition}

If the additive noise has zero mean, then the variance of our update direction is increased by the noise variance.
Otherwise, if the additive noise is bounded, then we may have a non-vanishing error floor when $\sigma_2 = \omega(\sqrt{\alpha})$.

\begin{proof}[Proof of \cref{prop:noise}]
	With an additive noise $\epsilon_{t}$, \cref{eq:err-1} becomes
	\[
		e_{t+1}
		= (I - \alpha (I-\nu \tilde{\Phi})) e_{t} + \alpha \nu (\hat{\Phi}_{t}-\tilde{\Phi})p_t + \alpha \epsilon^{(1)} _{t} + \alpha \epsilon^{(2)}_t
		.\]
	Then, \cref{eq:err-2} becomes
	\[
		\mathbb{E}\|e_{t+1}\|^{2} =& \mathbb{E}\|(I-\alpha (I-\nu \tilde{\Phi})) e_{t}\|^{2} + \alpha^{2} \mathbb{E}\|\nu (\hat{\Phi}_{t}-\tilde{\Phi})p_t + \epsilon^{(1)}_{t} + \epsilon^{(2)}_{t} \|^{2}  \\
		&+ 2\alpha \mathbb{E}\left< (I-\alpha (I-\nu \tilde{\Phi}))e_t, \epsilon^{(2)}_{t}  \right>
		,\]
	where the additional cross term is due to the non-zero mean of $\epsilon^{(2)}_{t}$.
	By Young's inequality, \cref{eq:err-3} becomes
	\[
		\mathbb{E}\|e_{t+1}\|^{2} \le& \mathbb{E}\|(I-\alpha (I-\nu \tilde{\Phi})) e_{t}\|^{2} + 3\alpha^{2} (\nu ^{2}\mathbb{E}\|(\hat{\Phi}_{t}-\tilde{\Phi})p_t\|^{2} + \sigma_1^{2} + \sigma_2^{2})  \\
		&+ \alpha ( (2\kappa_{\nu})^{-1}   \mathbb{E}\|(I-\alpha (I-\nu \tilde{\Phi}))e_t\|^{2} + 2\kappa_{\nu} \sigma_2^{2})\\
		\le& ((1 + \alpha (2\kappa_{\nu})^{-1})(1-2\alpha\kappa_{\nu}^{-1}+\alpha^{2}\kappa_{\nu}^{-2})+6\alpha^{2}\sigma^{2}))\mathbb{E} \|e_{t}\|^{2} \\
		& + 3\alpha^{2}(2\sigma ^{2}\|p_{*}\|^{2} + \sigma_1^{2}+\sigma_2^{2}) + 2\alpha \kappa_{\nu}\sigma_2^{2} \\
		\le& (1-\tfrac{3}{2}\alpha\kappa_{\nu}^{-1}+6\alpha^{2}\sigma^{2}+ \tfrac{1}{2}\alpha^3  \kappa_{\nu}^{-3})\mathbb{E} \|e_{t}\|^{2} + 3\alpha^{2}(2\sigma ^{2}\|p_{*}\|^{2} + \sigma_1^{2}+\sigma_2^{2}) + 2\alpha \kappa_{\nu}\sigma_2^{2}
		.\]
	Letting $\alpha \le \min \{ \kappa_{\nu} /2, 1 /(24\kappa_{\nu}\sigma^{2}) \}$ gives
	\[
		\mathbb{E}\|e_{t}\|^{2} \le& (1-\alpha \kappa_{\nu}^{-1} )\mathbb{E}\|e_{t-1}\|^{2} + 3\alpha^{2}(2\sigma ^{2}\|p_{*}\|^{2} + \sigma_1^{2}+\sigma_2^{2}) + 2\alpha \kappa_{\nu}\sigma_2^{2} \\
		\le& (1-\alpha \kappa_{\nu}^{-1} )^{t}\|e_{0}\|^{2} + 3\alpha \kappa_{\nu} (2\sigma ^{2}\|p_{*}\|^{2} + \sigma_1^{2}+\sigma_2^{2}) + 2 \kappa_{\nu}^{2}\sigma_2^{2}
		.\]
	Setting $\alpha = \kappa_{\nu} \ln t / t$ gives the desired result.
\end{proof}

\begin{remark}[Noisy query]
	Querying a social network with economic agents is often influenced by subjective biases, recall difficulty, and incomplete observation, introducing inevitable estimation noises in the data \citep{marsden1990NetworkData,jackson2011Socialeconomic}.
	\cref{prop:noise} showcases the robustness of \poll to such endogenous noise and exogenous shocks.
	Additionally, it justifies the ``approximate truthfulness'' in \cref{apx:eff-econ}.
\end{remark}

Within one iteration, dropping the time index $t$, we denote
$\xi_{1},\xi_{2}$ as the two random walks, $k_{1}=k^{(1)}_{t},k_2=k^{(2)}_{t}$ the lengths, $j^{(1)}=j^{(1)}_{k^{(1)}_t}, j^{(2)}=j^{(2)}_{k^{(2)}_t}$ the endpoints, and $d_1=\prod_{l=0}^{k^{(1)}_t-1} \deg_{j^{(1)}_{l}},d_2=\prod_{l=0}^{k^{(2)}_t-1} \deg_{j^{(2)}_{l}}$ the degree products, of the two random walks.

Although \cref{asmp:indep} implies the independence between the degree and preference of a single agent, it does not decouple the degree product and the preference product in the update direction.
For example, if two random walks end at the same node $j^{(1)}=j^{(2)}$, then $W_{j^{(1)}}$ and $W_{j^{(2)}}$ coincide.
In general, random walks may cause the indices $i$, $j^{(1)}$, or $j^{(2)}$ to overlap, and thus couple the degree product and the preference product.
One can argue that when $n$ is large, such scenarios are rare and thus the coupling is weak.
Nonetheless, we can have a uniform variance decomposition, by using a reference point $M$ to be determined.

Let $\mathcal{F}_{\mathrm{\X}}$ be the $\sigma$-field generated by the network sampling, including $i$, $k^{(1)}$, $k^{(2)}$, $\xi_1$, $\xi_2$, $d_1$, and $d_2$.
$\mathcal{F}_{\X}$ is independent of $\mathcal{F}_{t}$, and by \cref{asmp:indep}, the marginal distribution of $W_{i}$ is independent of both $\mathcal{F}_{t}$ and $\mathcal{F}_{\X}$, i.e., any $\mathcal{F}_{t}$- or $\mathcal{F}_{\X}$-measurable function, for any $i\in\X$.

\begin{lemma}[Variance decomposition]\label{lem:var-decomp}
	Let $V_{\X} \coloneqq \Var(d_1d_2)$, $E_{\X}\coloneqq\mathbb{E}[d_1d_2]$.
	For any $M\in R^{m\times m}$, let $E_{\A}\coloneqq\|M\|$.
	If we have
	\[
		\mathbb{E}[\|W_{z^{(1)}}^{-1} W_{x} W_{z^{(2)}}^{-1} - M\|^{2}\given \mathcal{F}_{\X}] \le V_{\A},
	\]
	then
	\[
		\sigma^{2} \le 6\nu^{2}\kappa_{\Psi}^{2}(1-\delta)^{-4} \left(V_{\A}V_{\X} + E_{\X}^{2}V_{\A} + E_{\A}^{2}V_{\X} \right)
		.\]
\end{lemma}

If $E_{\A}$ and $V_{\A}$ are the exact mean and variance of the preference distribution, then the above result exactly recovers the variance decomposition of the product of two independent random variables, i.e., $\Var(XY) = \Var(X)\Var(Y) + \mathbb{E}[X]^{2}\Var(Y) + \mathbb{E}[Y]^{2}\Var(X)$, suggesting its tightness.
That said, our update direction has a much more intricate and coupled structure.
The reference point $M$ to be determined, is used to decouple the degree product and the preference product, and later used to derive a bound of $V_{\A}$.

\begin{proof}
	We recall the following derivation from \cref{sec:eff-econ}:
	\[\label{eq:poll-2}
		\hat{\Phi}_t = \hat{\Phi}(Z_t) \coloneqq (1-\delta)^{-2}  \prod_{l=0}^{k^{(1)}_{t}-1} \deg_{j^{(1)}_{l}} \prod_{l=0}^{k^{(2)}_{t}-1} \deg_{j^{(2)}_{l}} \cdot \tilde{W}_{j^{(1)}_{k}}^{-1} \tilde{W}_{i_{t}} \tilde{W}_{j^{(2)}_{k}}^{-1}
		.\]
	By \cref{eq:poll-2}, we have
	\[
		\|\hat{\Phi}_{t}-\tilde{\Phi}\| \le& (1-\delta)^{-2}\ubar{\psi}^{-1} \|d_{1}d_{2}W_{j^{(1)}}^{-1} W_{i} W_{j^{(2)}}^{-1}  - \mathbb{E}[d_{1}d_{2}W_{j^{(1)}}^{-1} W_{x} W_{j^{(2)}}^{-1} ]\|
		.\]
	We have
	\[
		& d_{1}d_{2}W_{j^{(1)}}^{-1} W_{i} W_{j^{(2)}}^{-1}  - \mathbb{E}[d_{1}d_{2}W_{j^{(1)}}^{-1} W_{x} W_{j^{(2)}}^{-1} ] \\
		=&d_{1}d_{2}W_{j^{(1)}}^{-1} W_{i} W_{j^{(2)}}^{-1}  - d_{1}d_{2} M+ d_{1}d_{2}M - \mathbb{E}[d_{1}d_{2}]M + \mathbb{E}[d_{1}d_{2}]M - \mathbb{E}[d_{1}d_{2}W_{j^{(1)}}^{-1} W_{i} W_{j^{(2)}}^{-1}]  \\
		=&d_{1}d_{2}( W_{j^{(1)}}^{-1} W_{i} W_{j^{(2)}}^{-1} - M ) + (d_{1}d_{2}-\mathbb{E}[d_{1}d_{2}])M + \mathbb{E}[d_{1}d_{2}(M-W_{j^{(1)}}^{-1} W_{i} W_{j^{(2)}}^{-1})].
	\]
	The last term is a by-product of our arbitrary reference point; nonetheless, it's dominated by the first term by Jensen's inequality.
							Therefore, we have
	\[
		& \mathbb{E}\|d_{1}d_{2}W_{j^{(1)}}^{-1} W_{i} W_{j^{(2)}}^{-1}  - \mathbb{E}[d_{1}d_{2}W_{j^{(1)}}^{-1} W_{i} W_{j^{(2)}}^{-1} ]\|^{2} \\
		\le& 3\left( 2 \mathbb{E}\|d_1d_2 (W_{j^{(1)}}^{-1} W_{i} W_{j^{(2)}}^{-1} - M)\|^{2} + \mathbb{E}\|(d_1d_2-\mathbb{E}[d_1d_2])M\|^{2} \right) \\
		\le& 6\left( \mathbb{E}[d_1^{2}d_2^{2}] V_{\A} + E_{\X}^{2}\Var(d_1d_2) \right) \\
		\le& 6\left( V_{\X} V_{\A} + E_{\X}^{2}V_{\A} + E_{\A}^{2}V_{\X} \right)
		.\]
\end{proof}

We thus only need to bound the first and second moments of the degree product and the preference product.

\begin{lemma}\label{lem:varA}
	Let $M \coloneqq \mathbb{E}W_{i}^{-1}$, where the expectation is w.r.t. the preference generation.
	We have
	\[
		\mathbb{E}[\|W^{-1}_{j^{(1)} }W_{i}W_{j^{(2)} }^{-1} -M \|^{2}\given \mathcal{F}_{\X}] \le 9\kappa_{W}^{2}\sigma_{\A}^{2}
		.\]
\end{lemma}

\begin{proof}
	The reference point $M$ helps decouple the product:
	\[
		W_{j^{(1)} }^{-1}W_{i}W_{j^{(2)} }^{-1} - M
		= (W_{j^{(1)} }^{-1} - W_{i}^{-1} ) W_{i} W_{j^{(2)} }^{-1} + (W_{j^{(2)} }^{-1} - W_{i}^{-1}) + (W_{i}^{-1}  -M)
		.\]
	First note that we have a trivial bound
	\[\label{eq:varA-trivial}
		\|W_{j^{(1)} }^{-1}W_{i}W_{j^{(2)} }^{-1} - M\| \le 9 \bar{\lambda}^{2} / \underline{\lambda}^{4}
		.\]
	By \cref{asmp:indep}, $W_{i}^{-1} ,W_{j^{(1)} }^{-1},W_{j^{(2)} }^{-1}$ have the same marginal distribution, giving
	\[
		\mathbb{E}\|W_{i}^{-1}	-M\|_{\mathrm{F}}^{2} = \mathbb{E}\|W_{j^{(1)} }^{-1} - W_{i}^{-1} \|_{\mathrm{F}}^{2} = \mathbb{E}\|W_{j^{(2)} }^{-1}  - W_{i}^{-1} \|_{\mathrm{F}}^{2} = \sigma_{\A}^{2}
		.\]
	Let $\mathbb{E}_{\mathcal{F}_{\X}}$ be the expectation conditioned on $\mathcal{F}_{\X}$.
	We have
	\[
		&\mathbb{E}_{\mathcal{F}_{\X}}\|W_{j^{(1)} }^{-1}W_{i}W_{j^{(2)} }^{-1} - M\|^{2}\\
		\le& 3 \left( \kappa_{W}^{2} \mathbb{E}_{\mathcal{F}_{\X}}\|W_{j^{(1)}}^{-1} - W_{i}^{-1} \|_{\mathrm{F}}^{2} + \mathbb{E}_{\mathcal{F}_{\X}}\|W_{j^{(2)}}^{-1}  - W_{i}^{-1} \|_{\mathrm{F}}^{2} + \mathbb{E}_{\mathcal{F}_{\X}}\|W_{i}^{-1}-\mathbb{E}W_{i}^{-1}  \|_{\mathrm{F}}^{2} \right) \\
		=& 3 ( 2\kappa_{W}^{2}(1-\mathrm{Corr}_{\mathcal{F}_{\X}}(W_{j^{(1)}}^{-1} , W_{i}^{-1} )) + 2(1-\mathrm{Corr}_{\mathcal{F}_{\X}}(W_{j^{(2)}}^{-1} , W_{i}^{-1} )) + 1) \sigma_{\A}^{2}
		,\]
	where $\mathrm{Corr}_{\mathcal{F}_{\X}}$ returns the correlation conditioned on $\mathcal{F}_{\X}$.
	By symmetry, we also have
	\[
		&\mathbb{E}_{\mathcal{F}_{\X}}\|W_{j^{(1)} }^{-1}W_{i}W_{j^{(2)} }^{-1} - M\|^{2} \\
		\le& 3 ( 2\kappa_{W}^{2}(1-\mathrm{Corr}_{\mathcal{F}_{\X}}(W_{j^{(2)}}^{-1} , W_{i}^{-1} )) + 2(1-\mathrm{Corr}_{\mathcal{F}_{\X}}(W_{j^{(1)}}^{-1} , W_{i}^{-1} )) + 1) \sigma_{\A}^{2}
		.\]
	Note that $i,j^{(1)},j^{(2)}$ are $\mathcal{F}_{\X}$-measurable. Thus, we can define $\Corr_{\mathcal{F}_{\X}}\!\coloneqq \!\frac{1}{2}(\mathrm{Corr}_{\mathcal{F}_{\X}}(W_{j^{(1)}}^{-1} , W_{i}^{-1} ) + \mathrm{Corr}_{\mathcal{F}_{\X}}(W_{j^{(2)}}^{-1} , W_{i}^{-1} ))$ to get
	\[\label{eq:corr}
		\mathbb{E}_{\mathcal{F}_{\X}}\|W_{j^{(1)} }^{-1}W_{i}W_{j^{(2)} }^{-1} - M\|^{2}
		\le 3 ( 2(\kappa_{W}^{2} + 1)(1-\mathrm{Corr}_{\mathcal{F}_{\X}}) + 1) \sigma_{\A}^{2}
		.\]
	Or, defining $\bar{\Corr}_{\mathcal{F}_{\X}}\coloneqq \max \{ \Corr_{\mathcal{F}_{\X}}(W_{j^{(1)} }^{-1},W_{i}^{-1}  ),\Corr_{\mathcal{F}_{\X}}(W_{j^{(2)} }^{-1},W_{i}^{-1}  )  \}$, we get
	\[\label{eq:corr-1}
		\mathbb{E}_{\mathcal{F}_{\X}}\|W_{j^{(1)} }^{-1}W_{i}W_{j^{(2)} }^{-1} - M\|^{2}
		\le 3 ( 2\kappa_{W}^{2}(1-\bar{\Corr}_{\mathcal{F}_{\X}}) + 5) \sigma_{\A}^{2}
		.\]
	Note that $\kappa_{W}\ge 1$ and $\Corr_{\mathcal{F}_{\X}}, \bar{\Corr}_{\mathcal{F}_{\X}}\in [-1,1]$.
	Thus, a universal bound is
	\[
		\mathbb{E}\|W_{j^{(1)} }^{-1}W_{i}W_{j^{(2)} }^{-1} - M\|^{2} \le 9 \kappa_{W}^{2}\sigma_{\A}^{2}\le 9 \bar{\lambda}^{2}/ \ubar{\lambda}^{2} \sigma_{\A}^{2}
		.\]

\end{proof}

For a system with a wide range of technologies, $\kappa_{W}$ can be large and becomes a bottleneck. In such cases, a positive correlation $\Corr_{\mathcal{F}_{\X}}$ can be of huge help.
Thus, \cref{eq:corr} is a strictly sharper bound than the universal $V_{\A}\coloneqq 9\kappa_{W}^{2}\sigma_{\A}^{2}$, and we will revisit it later once we introduce the correlation structure.

\begin{lemma}\label{lem:varX}
	We have
	\[
		V_{\X} = \Var(d_1d_2) = \sigma_{\X}^{2}
		.\]
\end{lemma}

\begin{proof}
	In this proof we adopt a modular approach so we can reuse the calculations later.
	Conditioned on $i$, two random walks $\xi_1$ and $\xi_2$ are i.i.d. We thus drop the subscript to refer to either of them.

	By our sampling procedure, we have
	\[
		\mathbb{E}\left[ \prod_{l=0}^{k-1}\deg_{j_{l}}\Given i,k \right]
		= \sum_{|\xi|=k} \prod_{i=0}^{k-1}\deg_{j_{i}} \cdot \frac{1}{\deg_{i}} \frac{1}{\deg_{j_{1}}}\cdots \frac{1}{\deg_{j_{k-1}}}
		= \sum_{|\xi|=k}1,
	\]
	where the summation counts the number of distinct walks of length $k$ starting from $x$, which is equivalent to
	\[\label{eq:d1}
		\mathbb{E}\left[ \prod_{l=0}^{k-1}\deg_{j_{l}}\Given i,k \right]
		= \sum_{|\xi|=k}1 = \sum_{j_{1}\in\X, \dots, j_{k}\in \X} G_{ij_{1}}G_{j_{1}j_{2}}\cdots G_{j_{k-1}j_{k}} = e_{i}^T G^{k}\mathbf{1}.
	\]
	Integrating over $k$ gives
	\[
		\mathbb{E}\left[ \prod_{l=0}^{k-1}\deg(j_{l})\Given i \right] = (1-\delta) e_{i}^T(I-\delta G)^{-1} \mathbf{1}.
	\]
	Integrating over $i$ gives
	\[
		\mathbb{E}[d_{1}d_{2}] =& \mathbb{E}\left[ (1-\delta)^{2} (e_{i}^T(I-\delta G)^{-1}\mathbf{1})^{2} \right] \\
		=& (1-\delta)^{2} \left< (I-\delta G)^{-1}\mathbf{1},\mathbb{E}[e_{i}e_{i}^{T}](I-\delta G)^{-1}\mathbf{1} \right> \\
		=& \frac{(1-\delta)^{2}}{n}\|(1-\delta G)^{-1}\mathbf{1}\|^{2}
		.\]

	The second moment is calculated similarly:
	\[\label{eq:d2k}
		\mathbb{E}\left[\prod_{l=0}^{k-1}\deg_{j_{l}}^{2} \Given i,k\right] =& \sum_{|\xi|=k} \prod_{l=0}^{k-1}\deg_{j_{l}}^{2} \cdot \frac{1}{\deg_{i}} \frac{1}{\deg_{j_{1}}}\cdots \frac{1}{\deg_{j_{k-1}}} \\
		=& \sum_{|\xi|=k} \prod_{l=0}^{k-1}\deg_{j_{l}} \\
		=& \sum_{j_{1}\in\X,\cdots, j_{k}\in\X} \prod_{l=0}^{k-1}\deg_{j_{l}} G_{ij_{1}}\cdots G_{j_{k-1}j_{k}}  \\
		=& \sum_{j_{1}\in\X,\cdots, j_{k}\in\X} D_{ii}G_{ij_{1}}D_{j_{1}j_1}\cdots G_{j_{k-2}j_{k-1}}D_{j_{k-1}j_{k-1}}G_{j_{k-1}j_{k}}  \\
		=& e_{i}^T(DG)^{k} \mathbf{1}
		,\]
	where $D = \operatorname{diag}(\deg_{i})$ is the diagonal degree matrix.
	Integrating over $k$ gives
	\[\label{eq:d2i}
		\mathbb{E}\left[\prod_{l=0}^{k-1}\deg_{j_{l}}^{2} \Given i\right] = (1-\delta) e_{i}^T (I-\delta DG)^{-1} \mathbf{1}
		.\]
	Then,
	\[\label{eq:d2}
		\mathbb{E}[d_{1}^{2}d_{2}^{2}] = \mathbb{E} [(1-\delta)^{2} \left<  (I-\delta DG)^{-1} \mathbf{1}, e_{i}e_{i}^{T} (I-\delta DG)^{-1} \mathbf{1} \right>]
		= \frac{(1-\delta)^{2}}{n}\|(I-\delta DG)^{-1} \mathbf{1}\|^{2}
		.\]

	Note that $\|(I-\delta I)^{-1}\mathbf{1}\|^{2} = (1-\delta)^{-1}n$. Thus, we have
	\[
		\Var(d_1d_2) = \frac{(1-\delta)^{4}}{n^{2}} \left(\|(I-\delta DG)^{-1} \mathbf{1}\|^{2}\|(I-\delta I)^{-1}\mathbf{1}\|^{2}  - \|(I-\delta G)^{-1}\mathbf{1}\|^{4} \right)
		= \sigma_{\X}^{2}
		.\]
	A direct consequence is a trivial bound:
	\[\label{eq:varX-trivial}
		\Var(d_1d_2) \le \mathbb{E}[d_1^{2}d_2^{2}] \le (1-\delta)^{2}(1-\delta \sigma_{\max}(DG))^{-2} \le (1-\delta)^{2}(1-\delta \bar{d}^{2})^{-2}
		.\]
\end{proof}

Combining \cref{lem:err-dyn,lem:var-decomp,lem:varA,lem:varX} gives the desired query complexity in \cref{thm:query}.

\section{Variance reduction}\label{apx:var}

Throughout this section, we denote $\sigma_{\X}^{\mathrm{bl}},\sigma_{\A}^{\mathrm{bl}}$ as the standard deviations defined in \cref{def:het-pref,def:het-net} under the baseline model, and $\sigma_{\X}^{\mathrm{vr}},\sigma_{\A}^{\mathrm{vr}}$ as those after variance reduction.

\subsection{Control variate}\label{apx:cv}

Notice that the centrality in \cref{eq:cen-planner} is equivalent to
\[
	\tilde{\cen} = \sum_{k=0}^{\infty} \delta^{k}\tilde{G}^{k}\mathbf{1}
	= \mathbf{1} + \delta\sum_{k=0}^{\infty} \delta^{k}\tilde{G}^{k+1}\mathbf{1}
	= \mathbf{1} + \delta(1-\delta)^{-1} \mathbb{E}_{k \sim \operatorname{Geom}(\delta)} \tilde{G}^{k+1}\mathbf{1}
	.\]
If we sample the walks using the above expression, although we would have a one-step longer walk in expectation, the whole walk is discounted by an additional $\delta$, leading to a much smaller variance when $\delta$ is small (\cref{assump:net}).

Specifically, the algorithm is modified as follows: we sample the agent $i$ and $k_1 ,k_2 \sim \operatorname{Geom}(\delta)$ as before, but we set the walk lengths to be $k_1+1$ and $k_2+1$ instead of $k_1$ and $k_2$.
Let $d_1$ and $d_2$ be the degree products of the walks calculated the same way as before.
The agent composes the estimator as
\[
	\hat{\Phi} =& \left( I_{m} + \delta(1-\delta)^{-1}d_1  \tilde{W}_{j^{(1)}_{k_1+1}}^{-1} \tilde{W}_{i} \right) \tilde{W}_{i}^{-1} \left( I_{m} + \delta(1-\delta)^{-1}d_2 \tilde{W}_{i}\tilde{W}_{j^{(2)}_{k_{2}+1}}^{-1} \right) \\
	=& \left( \tilde{W}_{i}^{-1} + \delta(1-\delta)^{-1} \left( d_1  \tilde{W}_{j^{(1)}_{k_1+1}}^{-1} + d_2 \tilde{W}_{j^{(2)}_{k_{2}+1}}^{-1} \right) + \delta^{2}(1-\delta)^{-2}d_1d_2 \tilde{W}_{j^{(1)}_{k_1+1}}^{-1}\tilde{W}_{i}\tilde{W}_{j^{(2)}_{k_{2}+1}}^{-1} \right)
	\label{eq:cv}
	,\]
and then reports $\hat{\Phi}q_t$ as before.

\begin{proposition}\label{prop:cv}
	The above update has the same operation complexity, but enjoys the following variance reduction:
	\[
		\left(\frac{\sigma^{\mathrm{vr}}_{\X}}{\sigma_{\X}^{\mathrm{bl}}}\right)^{2} \lesssim \delta^{2} (\delta^{2} \vee (1-\delta\sigma_{\max}(DG)))
		.\]
\end{proposition}
\begin{proof}
	Note that $\tilde{W}_{j}^{-1} = \tilde{W}_{j}^{-1} \tilde{W}_{j} \tilde{W}_{j}^{-1}$ for any $j\in\X$. Thus, the variance of the preference products in \cref{eq:cv} is still bounded by $\sigma_{\A}$ as before. The variance of the degree products becomes
	\[
		\sigma_{\X}^{\mathrm{vr}} \lesssim \delta^{2}(1-\delta)^{-2}\Var(d_1) + \delta^{4} (1-\delta)^{-4}\Var(d_1d_2)
		.\]
	By \cref{lem:varX}, $\Var(d_1d_2)=(\sigma_{\X}^{\mathrm{bl}})^{2}$ and
	\[
		\Var(d_1) = \frac{1-\delta}{n}\mathbf{1}^{T}(I-\delta DG)^{-1} \mathbf{1} - \frac{1-\delta}{n}\|(I-\delta G)^{-1}\mathbf{1}\|^{2} \lesssim (1 - \delta \sigma_{\max}(DG))^{-1}
		.\]
	Note that $\sigma^{\mathrm{bl}}_{\X} \asymp (1-\delta \sigma_{\max}(DG))^{-2}$. Combining the above gives the desired result.
\end{proof}

Notably, the above technique, purely motivated by our optimal intervention characterizations, coincides with the standard \emph{control variate} technique \citep{zhang2026affpcl}.
A control variate is a random variable with known expectation that is positively correlated with the original estimator. Thus, subtracting the control variate and adding its expectation reduces the variance.
Alternatively, if we know the expectation of a \emph{part} of the estimator, then we only need to sample the residual part, and the difference between the original estimator and the residual is a control variate.
Formally,
\[
	& X = X^{\mathrm{cv}} + X^{\mathrm{residual}},\quad\text{with } \mathbb{E}X^{\mathrm{cv}} \text{ known}\\
	\implies & X^{\mathrm{new}} = \mathbb{E}[X^{\mathrm{cv}}] + X^{\mathrm{residual}} = X - X^{\mathrm{cv}} + \mathbb{E}[X^{\mathrm{cv}}] \text{ is unbiased and variance reduced}
	.\]
In our case, when $k=0$, we do not need to sample the walk, and we know the expectation of this scenario. This base case happens to be the most significant part of the estimator due the exponential tail, and thus serves as a very effective control variate that correlates well with the original estimator.
Then, the new estimator contributes the residual part.

\subsection{Preference correlation}

Real-world network systems often observe correlated preferences within local neighborhoods.
There are two scenarios in which a system generation model satisfying \cref{asmp:indep} produces positive correlation: preferences are developed on a given network, e.g., through assimilation or social learning \citep{acemoglu2011BayesianLearning}, or the network is formed based on preference similarity, e.g., through homophily or assortative mixing \citep{newman2018Networks}.

\begin{proposition}\label{prop:pref-corr}
	Suppose there exists a non-increasing function $r\colon [\bar{k}]\to [-1,1]$ with $\bar{k}\ge 1$ and $[\bar{k}]\coloneqq\left\{ 0,\dots,\bar{k} \right\}$ such that $r(0)=1$ and $r(1)>0$, such that for any two agents $i,j\in\X$,
	\[
		\Corr(W_{i}^{-1},W_{j}^{-1}) \ge r(\dist(i,j))
		,\]
	where $\Corr$ is w.r.t. the preference generation and Frobenius inner product, and $\dist$ returns the shortest path distance.
	Then,
	\[
		\left(\frac{\sigma_{\A}^{\mathrm{vr}}}{\sigma_{\A}^{\mathrm{bl}}}\right)^{2} \lesssim \frac{ \frac{5}{4\kappa_{W}^{2}} + \min_{k\in[\bar{k}]}f_{\delta}(k)}{\frac{5}{4\kappa_{W}^{2}} + 1}
		,\]
	where $f_{\delta}(k) \coloneqq (2(1+\delta^{4}))^{-1} (1-r(k) + 2\delta^{2k+2})$ and $\min_{k}f_{\delta}(k) \le f_{\delta}(1) = 1-\frac{1 + r(1)}{2(1+\delta^{4})} < 1$.
\end{proposition}

When $\kappa_{W}$ is large, positive preference correlation can significantly reduce the query complexity when $r(1)\uparrow 1$ and $\delta \downarrow 0$.

\begin{proof}
	To inspect the effect of preference correlation, we need to handle the coupling between the degree product and the preference product more carefully than in \cref{lem:varA}.
	Instead of conditioned on $\mathcal{F}_{\X}$, we condition on the walk lengths.
	By the tower property,
	\[
		\mathbb{E}[d_1^{2}d_2^{2}\|W^{-1}_{j^{(1)} }W_{i}W_{j^{(2)} }^{-1} -M \|^{2}\given k_1,k_2]
		= \mathbb{E}[d_1^{2}d_2^{2}  \mathbb{E}[\|W^{-1}_{j^{(1)} }W_{i}W_{j^{(2)} }^{-1} -M \|^{2} \given \mathcal{F}_{\X}]\given k_1,k_2]
		.\]
	By \cref{eq:corr-1},
	\[
		\mathbb{E}[d_1^{2}d_2^{2}\|W^{-1}_{j^{(1)} }W_{i}W_{j^{(2)} }^{-1} -M \|^{2}\given k_1,k_2]
		\le \mathbb{E}[d_1^{2}d_2^{2} \cdot 3(2\kappa_{W}^{2}(1-\bar{\Corr}_{\mathcal{F}_{\X} }) + 5)\sigma^{2}_{\A}  \given k_1,k_2]
		,\]
	where
	\[
		\bar{\Corr}_{\mathcal{F}_{\X}} = \max\{\mathrm{Corr}_{\mathcal{F}_{\X}}(W_{j^{(1)}}^{-1} , W_{i}^{-1} ), \mathrm{Corr}_{\mathcal{F}_{\X}}(W_{j^{(2)}}^{-1} , W_{i}^{-1} )\}
		.\]
	If $k_1\wedge k_2 \le \bar{k}$, since $r$ is non-increasing, we have
	\[
		\bar{\Corr}_{\mathcal{F}_{\X}} \ge r(k_1\wedge k_2)
		.\]
	If $k_1\wedge k_2 > \bar{k}$, we have the trivial bound $\bar{\Corr}_{\mathcal{F}_{\X}} \ge -1$.
	Thus,
	\[
		&\mathbb{E}[d_1^{2}d_2^{2}\|W^{-1}_{j^{(1)} }W_{i}W_{j^{(2)} }^{-1} -M \|^{2}\given k_1,k_2] \\
		\le& 3(2\kappa_{W}^{2}(1-r(k_1\wedge k_2)\mathbbm{1}\left\{ k_1\wedge k_2 \le \bar{k} \right\} + \mathbbm{1}\left\{ k_1\wedge k_2 > \bar{k} \right\}) + 5)\sigma^{2}_{\A} \cdot \mathbb{E}[d_1^{2}d_2^{2} \given k_1,k_2]
		.\]
	By \cref{eq:d2k} and the memoryless property of the geometric distribution, we have
	\[
		\mathbb{E}[d_1^{2}d_2^{2}\given k_1\wedge k_2 > \bar{k}]
		=& \frac{(1-\delta)^{2}}{n} \sum_{k_1=\bar{k}+1}^{\infty} \sum_{k_2=\bar{k}+1}^{\infty} \delta^{k_1+k_2-2(\bar{k}+1)} \left< (DG)^{k} \mathbf{1},  (DG)^{k'} \mathbf{1} \right>\\
		=& \frac{(1-\delta)^{2}}{n} \sum_{k_1=0}^{\infty} \sum_{k_2=0}^{\infty} \delta^{k_1+k_2} \left< (DG)^{k} \mathbf{1},  (DG)^{k'} \mathbf{1} \right>\\
		=& \frac{(1-\delta)^{2}}{n} \|(I-\delta DG)^{-1} \mathbf{1}\|^{2}\\
		=& \mathbb{E}[d_1^{2}d_2^{2}]
		.\]
	Finally, we have
	\[
		&\mathbb{E}[d_1^{2}d_2^{2}\|W^{-1}_{j^{(1)} }W_{i}W_{j^{(2)} }^{-1} -M \|^{2}]\\
		=& P(k_1\wedge k_2 \le \bar{k}) \cdot \mathbb{E}[d_1^{2}d_2^{2}\|W^{-1}_{j^{(1)} }W_{i}W_{j^{(2)} }^{-1} -M \|^{2}\given k_1\wedge k_2 \le \bar{k}] + \\
		& P(k_1\wedge k_2 > \bar{k}) \cdot \mathbb{E}[d_1^{2}d_2^{2}\|W^{-1}_{j^{(1)} }W_{i}W_{j^{(2)} }^{-1} -M \|^{2}\given k_1\wedge k_2 > \bar{k}] \\
		\le& 3(2\kappa_{W}^{2}(1-r(\bar{k})) + 5) \cdot P(k_1\wedge k_2 \le \bar{k}) \cdot  \mathbb{E}[d_1^{2}d_2^{2}\given k_1\wedge k_2 \le \bar{k}] + \\
		& 3(4\kappa_{W}^{2} + 5) \cdot P(k_1\wedge k_2 > \bar{k}) \cdot \mathbb{E}[d_1^{2}d_2^{2}\given k_1\wedge k_2 > \bar{k}] \\
		\le& 3(2\kappa_{W}^{2}(1-r(\bar{k})) + 5) \cdot \mathbb{E}[d_1^{2}d_2^{2}] + \\
		& 3(4\kappa_{W}^{2} + 5) \cdot P(k_1\wedge k_2 > \bar{k}) \cdot \mathbb{E}[d_1^{2}d_2^{2}\given k_1\wedge k_2 > \bar{k}] \\
		=& 3(2\kappa_{W}^{2}(1-r(\bar{k})) + 5) \cdot \mathbb{E}[d_1^{2}d_2^{2}] + \\
		& 3(4\kappa_{W}^{2} + 5) \cdot \delta^{2(\bar{k}+1)}
		\cdot \mathbb{E}[d_1^{2}d_2^{2}] \\
		=& 3(4\kappa_{W}^{2} + 5)\mathbb{E}[d_{1}^{2}d_2^{2}]\underbrace{\left( \frac{2\kappa_{W}^{2}(1-r(\bar{k})) + 5}{4\kappa_{W}^{2} + 5} + \delta^{2(\bar{k}+1)} \right)}_{\text{shrinkage of }\sigma_{\A}^{2}}
		.\]
	The above result holds for arbitrary $k\in[\bar{k}]$. Thus,
	\[
		\left(\frac{\sigma_{\A}^{\mathrm{vr}}}{\sigma_{\A}^{\mathrm{bl}}}\right)^{2} \le \min_{k\in[\bar{k}]} \frac{2\kappa_{W}^{2}(1-r(k)) + 5}{4\kappa_{W}^{2} + 5} + \delta^{2(k+1)}
		\le \frac{\frac{5(1+\delta^{4})}{2\kappa_{W}^{2}} + \min_{k\in [\bar{k}]}\left\{ 1-r(k)+ 2\delta^{2k +2} \right\} }{\frac{5}{2\kappa_{W}^{2}} +2}
		.\]
\end{proof}

\subsection{Degree correlation}

Real-world network systems often exhibit degree correlation \citep{newman2018Networks}.
Two polar cases are \emph{assortativity}, where high-degree agents tend to connect to other high-degree agents, and \emph{disassortativity}, where high-degree agents tend to connect to low-degree agents.
In our setting, degree correlation affects the variance through the degree products accumulated along the random walks.
While a general graph-level condition is less transparent than in \cref{prop:pref-corr}, a stylized two-degree model already captures the key effect: assortativity keeps the walk in the high-degree class and amplifies the estimator, whereas disassortativity forces alternation between high- and low-degree classes and stabilizes it.

\begin{proposition}\label{prop:deg-corr}
	Suppose the network is undirected and unweighted, and every agent belongs to one of two degree classes with degrees $d_{\mathrm{L}}<d_{\mathrm{H}}$.
	Assume that under the vanilla uniform-neighbor walk in \poll, the induced degree-class process is exactly the two-state Markov chain
	\[
		P_{\eta} \coloneqq
		\begin{pmatrix}
			\eta   & 1-\eta \\
			1-\eta & \eta
		\end{pmatrix},
		\qquad \eta\in[0,1],
	\]
	where $\eta=1$ corresponds to perfect assortativity, $\eta=1/2$ to degree-independent mixing at the class level, and $\eta=0$ to perfect disassortativity.
	Assume also that $\delta d_{\mathrm{H}}^{2}<1$.

	Let $(\sigma_{\X}^{\mathrm{bl}})^{2}$ and $(\sigma_{\X}^{\mathrm{vr}})^{2}$ denote the corresponding class-level upper bounds on the network variance term $\sigma_{\X}^{2}$ under $\eta=1$ and under a general $\eta\in[0,1]$, respectively.
	Then,
	\[
		\left(\frac{\sigma_{\X}^{\mathrm{vr}}}{\sigma_{\X}^{\mathrm{bl}}}\right)^{2}
		\lesssim
		\frac{1-\delta d_{\mathrm{H}}^{2}}{1-\delta f_{d}(\eta)},
	\]
	where
	\[
		f_{d}(\eta)
		\coloneqq
		\frac{
			\eta(d_{\mathrm{L}}^{2}+d_{\mathrm{H}}^{2})
			+
			\sqrt{
				\eta^{2}(d_{\mathrm{H}}^{2}-d_{\mathrm{L}}^{2})^{2}
				+
				4(1-\eta)^{2}d_{\mathrm{L}}^{2}d_{\mathrm{H}}^{2}
			}
		}{2}
		.
	\]
	Moreover, the right-hand side is non-decreasing in $\eta$, with
	\[
		\frac{1-\delta d_{\mathrm{H}}^{2}}{1-\delta f_{d}(1)}=1,
		\qquad
		\frac{1-\delta d_{\mathrm{H}}^{2}}{1-\delta f_{d}(0)}
		=
		\frac{1-\delta d_{\mathrm{H}}^{2}}{1-\delta d_{\mathrm{L}}d_{\mathrm{H}}}
		<1.
	\]
	In particular, degree disassortativity strictly reduces the variance proxy of the original estimator.
\end{proposition}

When the degree gap is large, positive degree correlation can severely inflate the query complexity, while strong disassortativity can substantially reduce it, especially as $\delta d_{\mathrm{H}}^{2}\uparrow 1$.

\begin{proof}
	Let $c_{l}\in\{\mathrm{L},\mathrm{H}\}$ denote the degree class visited at step $l$ of a uniformly sampled walk, and write $\bm{1}_{2}=(1,1)^{T}$.
	Under the assumed two-class model, $(c_{l})_{l\ge 0}$ is a Markov chain with transition matrix $P_{\eta}$.
	Define
	\[
		B_{\eta}
		\coloneqq
		\diag(d_{\mathrm{L}}^{2},d_{\mathrm{H}}^{2})P_{\eta}
		=
		\begin{pmatrix}
			\eta d_{\mathrm{L}}^{2}    & (1-\eta)d_{\mathrm{L}}^{2} \\
			(1-\eta)d_{\mathrm{H}}^{2} & \eta d_{\mathrm{H}}^{2}
		\end{pmatrix}.
	\]

	Then the class-level analogue of \cref{eq:d2k} gives, for a walk starting in class $c\in\{\mathrm{L},\mathrm{H}\}$,
	\[
		\mathbb{E}\!\left[\prod_{l=0}^{k-1}\deg_{j_{l}}^{2}\,\middle|\, c_{0}=c,k\right]
		=
		e_{c}^{T}B_{\eta}^{k}\bm{1}_{2},
	\]
	where $e_{\mathrm{L}},e_{\mathrm{H}}$ are the standard basis vectors in $\mathbb{R}^{2}$.
	Integrating over $k\sim \operatorname{Geom}(\delta)$ yields
	\[
		\mathbb{E}\!\left[\prod_{l=0}^{k-1}\deg_{j_{l}}^{2}\,\middle|\, c_{0}=c\right]
		=
		(1-\delta)e_{c}^{T}(I-\delta B_{\eta})^{-1}\bm{1}_{2}.
	\]

	Conditioned on the starting node, the two walks in \poll are independent. Hence
	\[
		\mathbb{E}[d_{1}^{2}d_{2}^{2}\mid c_{0}=c]
		=
		(1-\delta)^{2}
		\Bigl(
		e_{c}^{T}(I-\delta B_{\eta})^{-1}\bm{1}_{2}
		\Bigr)^{2}.
	\]
	Averaging over the starting class and using $\Var(d_{1}d_{2})\le \mathbb{E}[d_{1}^{2}d_{2}^{2}]$, we get
	\[
		(\sigma_{\X}^{\mathrm{vr}})^{2}
		\le
		\mathbb{E}[d_{1}^{2}d_{2}^{2}]
		\lesssim
		(1-\delta)^{2}\left\|(I-\delta B_{\eta})^{-1}\bm{1}_{2}\right\|_{2}^{2}.
	\]

	Let $T\coloneqq\diag(d_{\mathrm{L}},d_{\mathrm{H}})$. Since
	\[
		B_{\eta}=TS_{\eta}T^{-1},
		\qquad
		S_{\eta}\coloneqq
		\begin{pmatrix}
			\eta d_{\mathrm{L}}^{2}              & (1-\eta)d_{\mathrm{L}}d_{\mathrm{H}} \\
			(1-\eta)d_{\mathrm{L}}d_{\mathrm{H}} & \eta d_{\mathrm{H}}^{2}
		\end{pmatrix},
	\]
	the matrices $B_{\eta}$ and $S_{\eta}$ have the same eigenvalues.
	Since $S_{\eta}$ is symmetric and nonnegative,
	\[
		\left\|(I-\delta B_{\eta})^{-1}\bm{1}_{2}\right\|_{2}
		\lesssim
		(1-\delta\lambda_{\max}(S_{\eta}))^{-1},
	\]
	where the hidden constant is universal because the class space has fixed dimension two.
	Thus,
	\[
		(\sigma_{\X}^{\mathrm{vr}})^{2}
		\lesssim
		(1-\delta f_{d}(\eta))^{-2},
	\]
	where $f_{d}(\eta)=\lambda_{\max}(S_{\eta})$ equals
	\[
		f_{d}(\eta)
		=
		\frac{
			\eta(d_{\mathrm{L}}^{2}+d_{\mathrm{H}}^{2})
			+
			\sqrt{
				\eta^{2}(d_{\mathrm{H}}^{2}-d_{\mathrm{L}}^{2})^{2}
				+
				4(1-\eta)^{2}d_{\mathrm{L}}^{2}d_{\mathrm{H}}^{2}
			}
		}{2}.
	\]

	For the perfectly assortative baseline $\eta=1$, we have
	\[
		S_{1}=\diag(d_{\mathrm{L}}^{2},d_{\mathrm{H}}^{2}),
		\qquad
		f_{d}(1)=d_{\mathrm{H}}^{2},
	\]
	and hence
	\[
		(\sigma_{\X}^{\mathrm{bl}})^{2}
		\lesssim
		(1-\delta d_{\mathrm{H}}^{2})^{-2}.
	\]
	Taking square roots and then the ratio of the two bounds yields
	\[
		\left(\frac{\sigma_{\X}^{\mathrm{vr}}}{\sigma_{\X}^{\mathrm{bl}}}\right)^{2}
		\lesssim
		\frac{1-\delta d_{\mathrm{H}}^{2}}{1-\delta f_{d}(\eta)}.
	\]

	It remains to show monotonicity in $\eta$.
	For $\eta'\ge \eta$,
	\[
		S_{\eta'}-S_{\eta}
		=
		(\eta'-\eta)
		\begin{pmatrix}
			d_{\mathrm{L}}^{2}            & -d_{\mathrm{L}}d_{\mathrm{H}} \\
			-d_{\mathrm{L}}d_{\mathrm{H}} & d_{\mathrm{H}}^{2}
		\end{pmatrix}
		=
		(\eta'-\eta)
		\begin{pmatrix}
			d_{\mathrm{L}} \\
			-d_{\mathrm{H}}
		\end{pmatrix}
		\begin{pmatrix}
			d_{\mathrm{L}} \\
			-d_{\mathrm{H}}
		\end{pmatrix}^{T}
		\succeq 0.
	\]
	Thus $S_{\eta}$ is increasing in the Loewner order, and therefore $f_{d}(\eta)$ is non-decreasing in $\eta$.

	Finally,
	\[
		f_{d}(0)=d_{\mathrm{L}}d_{\mathrm{H}},
		\qquad
		f_{d}(1)=d_{\mathrm{H}}^{2},
	\]
	which implies
	\[
		\frac{1-\delta d_{\mathrm{H}}^{2}}{1-\delta f_{d}(1)}=1,
		\qquad
		\frac{1-\delta d_{\mathrm{H}}^{2}}{1-\delta f_{d}(0)}
		=
		\frac{1-\delta d_{\mathrm{H}}^{2}}{1-\delta d_{\mathrm{L}}d_{\mathrm{H}}}
		<1.
	\]
	This proves the claim.
\end{proof}

The same proof extends to any finite number of degree classes: the relevant quantity is the Perron root of the class-level matrix $\diag(d_{c}^{2})P$, where $P$ is the transition matrix of the degree-class process. The two-class model makes this dependence explicit.

\subsection{Power law and importance sampling}\label{apx:is}

Many real-world networks exhibit a heavy-tailed power law degree distribution $\deg_{i}\propto i^{-\gamma}$ \citep{newman2003StructureFunction}, which has large higher-order moments.
For these networks, $(1-\delta \sigma_{\max}(DG))^{-1}$ can be a severe bottleneck.
Our centrality characterization in \cref{eq:cen-planner} gives a natural remedy as it allows flexible sampling  in \cref{eq:walk}.
First, we note that the uniform sampling of neighbors can be directly generalized to accommodate weighted graphs using a weighted sampling such that agent $j_{l}$ is chosen with probability $\propto G_{j_{l-1}j_{l}}$ without any modification to the algorithm or its analysis.
We further notice that the degree product in our update direction is simply the probability normalization factor, or the \emph{importance weight}, of the random walk.
This motivates us to design a tailored sampling scheme and adjust the importance weight to ensure unbiasedness; this is the standard importance sampling technique.

\begin{proposition}\label{prop:power}
	Consider a random walk sampling scheme where a neighbor $j_{l}$ is sampled with probability $\propto \deg_{j_{l}}$. The importance weight of this scheme is
	\[
		\frac{\deg_{i}}{\deg_{j_{k}}} \prod_{i=0}^{k-1}\deg'_{j_{i}},\quad \text{where }\deg_{j}' = \frac{\sum_{j\in \X_{j}}\deg_{j}}{|\X_{j}|} \text{ is the \textup{average neighbor degree}}
		.\]
	We replace the degree product with this importance weight; other parts of the algorithm remain the same.

	Suppose the degrees $\{\deg_{i}\}_{i\in\X}$ follows a power law with exponent $\gamma$ and the network is generated by a mean field configuration model \citep{chung2002averagedistances} with $G_{ij} = \frac{d_{i}d_{j}}{\sum_{i}d_{i}}$.
	Then,
	\[
		\frac{\sigma^{\mathrm{vr}}_{\X}}{\sigma^{\mathrm{bl}}_{\X}} \lesssim \frac{ \frac{\gamma-4}{\delta(\gamma-2)} - 1 }{\frac{\gamma-4}{\delta(\gamma-2)} - \min_{\beta\in[\gamma-2]} f_{\gamma}(\beta) }
		,\]
	where $f_{\gamma}(\beta)\coloneqq\frac{1}{(1-\frac{\beta}{\gamma-2})(1+\frac{\beta}{\gamma-4})}$
	and $\min_{\beta}f_{\gamma}(\beta) = f_{\gamma}(1) = \frac{(\gamma-2)(\gamma-4)}{(\gamma-3)^{2}}<1$.
\end{proposition}

We see that $\sigma^{\mathrm{vr}}_{\X} / \sigma^{\mathrm{bl}}_\X \downarrow 0$ as $\gamma \downarrow (1-\delta)^{-1}(4-2\delta)$, which is when the original uniform sampling scheme fails, showcasing a significant variance reduction for heavy-tailed networks.

Heuristically, our weighted sampling offers variance reduction from three aspects:

\begin{enumerate}\tightlist
	\item  With lower entropy, weighted sampling returns more consistent trajectories. Specifically, weighted sampling consistently walks towards high-degree nodes.
	\item Essentially all information, technological, and biological networks appear to be \textit{disassortative} \citep{newman2003StructureFunction},
	      in which a high-degree node $j$ is more likely to satisfy $\deg_{j}\ge \deg_{j'}$ for its neighbors $j'$, thus it is more likely that $\deg_{j} = |\X_{j}|\ge \deg'_{j}$, resulting in a smaller value of the degree product when using weighted sampling.
	\item When ending at a high-degree node, $\deg_{i} / \deg_{j_{k}} < 1$, further reducing the magnitude of the estimator when using weighted sampling.
\end{enumerate}

Consistent with importance sampling theory, we cannot expect a universal variance reduction across all classes of networks; instead, we should choose the right scheme based on the given structure.
In \cref{apx:is}, we consider general importance sampling schemes and discuss principled ways of choosing the right importance weighting.

\subsubsection{Proof of \texorpdfstring{\cref{prop:power}}{Proposition 9}}

We consider a general importance sampling scheme which samples a neighbor of $j_{l-1}$ with probability $\propto G_{j_{l-1}j_{l}}h(j_{l})$, where $h(\cdot)$ is a positive function to be designed.
Following \cref{eq:walk}, we have
\[
	[G^{k}\mathbf{1}]_{i}
	= \sum_{j_1\in \X_{i}} \!G_{ij_1} [G^{k-1}\mathbf{1}]_{j_1}
	= \left(\sum_{j_1\in\X_{i}}\!G_{ij_{1}}h(j_{1})\right)\sum_{j_1\in \X_{i}} \!h^{-1}(j_1) \frac{G_{ij_1}h(j_1)}{\sum_{j_1\in\X_{i}}G_{ij_{1}}h(j_{1})} [G^{k-1}\mathbf{1}]_{j_1}
	.\]
Letting $h'(i) \coloneqq h^{-1}(i) \sum_{j\in\X_{i}}G_{ij}h(j)$
and denoting $\mathbb{E}^{\mathrm{imp}}$ as the expectation under the importance sampling, we get
\[
	[G^{k}\mathbf{1}]_{i} = h(i) h'(i) \mathbb{E}^{\mathrm{imp}} _{j_1}h^{-1}(j_1) [[G^{k-1}\mathbf{1} ]_{j_1}]
	.\]
By induction, we get
\[
	[G^{k}\mathbf{1} ]_{i} = \mathbb{E}^{\mathrm{imp}}_{(j_1,\dots,j_{k})} \left[ \prod_{l=0}^{k-1} h(j_{l})h'(j_l)h^{-1}(j_{l+1}) \right]
	= \mathbb{E}^{\mathrm{imp}}_{(j_1,\dots,j_{k})} \left[ h(i)h^{-1}(j_k) \prod_{l=0}^{k-1} h'(j_l) \right]
	.\]
When $h(j) \equiv 1$, $h'(j) = \deg_{j}$, and we recover the original uniform sampling procedure.

\begin{lemma}[Importance weight variance]\label{lem:imp}
	Denote $h_1,h_2$ as the importance weights, which replace the degree products, of the two walks under the importance sampling procedure. Let $H = \diag(h(i))$ and $H' = \diag(h'(i))$. We have
	\[
		\mathbb{E}^{\mathrm{imp}} [h_1^{2}h_2^{2}] = \frac{(1-\delta)^{2}}{n}\|(I-\delta H H' G H^{-1})^{-1}\mathbf{1}\|^{2}
		.\]
\end{lemma}

\begin{proof}
	We omit the superscript $\mathrm{imp}$ for brevity.
	Following \cref{eq:d2k}, we have
	\[
		&\mathbb{E}\left[ h^{2}(i)h^{-2}(j_k) \prod_{l=0}^{k-1}h'(j_{l})^{2} \Given i,k  \right] \\
		=& \sum_{|\xi|=k}h^{2}(i)h^{-2}(j_k)\prod_{l=0}^{k-1}h'(j_{l})^{2} \cdot \frac{G_{ij_1}h(j_1)}{\sum_{j_1\in\X_{i}}G_{ij_{1}}h(j_{1})} \cdots \frac{G_{j_{k-1}j_k}h(j_k)}{\sum_{j_k\in\X_{j_{k-1}}}G_{j_{k-1}j_k}h(j_k)} \\
		=& \sum_{|\xi|=k}h(i)h^{-1}(j_k)\prod_{l=0}^{k-1}h'(j_{l}) \cdot G_{ij_1} \cdots G_{j_{k-1}j_k} \\
		=& \sum_{|\xi|=k}H_{ii}H^{-1}_{j_{k}j_{k}} \prod_{l=0}^{k-1}H'_{j_{l}j_{l}}  G_{j_{l}j_{l+1}} \\
		=& e_{i}^T  H(H' G)^{k} H^{-1} \mathbf{1}
		.\]
	Then, following \cref{eq:d2i,eq:d2} gives
	\[
		\mathbb{E}[h_1^{2}h_2^{2}] = \frac{(1-\delta)^{2}}{n}\|(I - \delta HH'GH^{-1} )^{-1}  \mathbf{1} \|^{2}
		.\]
\end{proof}

Again, setting $h(i) \equiv 1$ gives the original variance in \cref{lem:varX}.
We should design $h$ to minimize $\sigma_{\max}(HH'GH^{-1})$.

\begin{lemma}[Expected configuration model]\label{lem:config}
	Consider a parameterized family of importance sampling schemes with $h_{\beta}(j) = \deg_{j}^{\beta}$ for $\beta>0$.
	Given a degree sequence $\{\deg_i\}_{i\in\X}$, let $G$ be generated by the expected Chung-Lu configuration model \citep{chung2002averagedistances}, i.e., $G_{ij} = \frac{\deg_i \deg_j}{\sum_{i}\deg_i}$.
	We have
	\[
		\frac{1}{n}\|(I - \delta H_{\beta}H'_{\beta}GH_{\beta}^{-1})^{-1}\mathbf{1}\|^{2} = (1 + 2z_{2}g_{\delta}(\beta) + z_{4}g^{2}_{\delta}(\beta) )
		,\]
	where $z_{l} = \mathbb{E}[\deg_{i}^{l}]$ is the $l$-th moment of the degree distribution, and $g_{\delta}(\beta) = \frac{\delta z_{1+\beta}z_{1-\beta}}{z_1^{2} - \delta z_{1+\beta}z_{3-\beta}}$.
\end{lemma}

\begin{proof}
	We denote $w_{l} \coloneqq (\deg_{1}^{l},\dots,\deg_{n}^{l})^{T}\in\R^{n}$ and $z_{l}\coloneqq \mathbb{E}[\deg_{i}^{l}]$ for $l> 0$. Then,
	\[
		G = \frac{1}{nz_{1}}w_{1}w_{1}^{T},\quad H_{\beta}=\diag(w_{\beta})
		,\]
	and
	\[
		H'_{\beta} = \diag(H_{\beta}^{-1} G w_\beta)
		= \frac{1}{nz_{1}} \diag(H_{\beta}^{-1}w_1 w_1^{T}w_\beta )
		= \frac{z_{1+\beta}}{z_{1}} \diag(w_{1-\beta} )
		.\]
	Thus,
	\[
		H_{\beta}H_{\beta}'GH_{\beta}^{-1}
		=& \frac{1}{nz_{1}} H_{\beta}' \diag(w_{\beta})w_1 (\diag(w_{\beta}^{-1})w_1)^{T}\\
		=& \frac{1}{nz_{1}} H_{\beta}' w_{1+\beta} w_{1-\beta}^{T} \\
		=& \frac{z_{1+\beta}}{nz_{1}^{2}}  w_{2} w_{1-\beta}^{T}.
	\]
	By the Sherman-Morrison formula, we have
	\[
		(I - \delta H_{\beta}H'_{\beta}GH_{\beta}^{-1} )^{-1}
		= I + \frac{\delta}{n}\frac{z_{1+\beta}}{z_{1}^{2} - \delta z_{1+\beta}z_{3-\beta}}  w_{2} w_{1-\beta}^{T}
		.\]
	Let $g_{\delta}(\beta)\coloneqq \frac{\delta z_{1+\beta}z_{1-\beta}}{z_1^{2} - \delta z_{1+\beta}z_{3-\beta}}$.
	Finally,
	\[
		\frac{1}{n}\|(I - \delta H_{\beta}H_{\beta}'GH_{\beta}^{-1})^{-1}\mathbf{1}\|^{2}
		= \frac{1}{n}\left\|\mathbf{1} + \frac{\delta z_{1+\beta}z_{1-\beta}}{z_{1}^{2} - \delta z_{1+\beta}z_{3-\beta} } w_{2} \right\|^{2}
		= \frac{1}{n}( 1 + 2z_{2}g_{\delta}(\beta) + z_{4}g^{2}_{\delta}(\beta) )
		.\]
\end{proof}

\begin{proof}[Proof of \cref{prop:power}]
	Since we allow weighted graphs, we consider real-valued degrees.
	The general moment of a continuous power law distribution is given by
	\[
		\mathbb{E}[\deg_{i}^{\beta}] = \frac{\gamma-1}{\gamma-1-\beta}
		.\]
	Plugging this into \cref{lem:config} gives
	\[
		\|(I - \delta H_{\beta}H_{\beta}'GH_{\beta}^{-1})^{-1}\mathbf{1}\|^{2} = \frac{n(\gamma-1)}{(\gamma-5)}\left( 1 - \delta \frac{(\gamma-2)^{2}}{(\gamma-2-\beta)(\gamma-4+\beta)} \right)^{-2} \frac{(\gamma-4+\beta)^{2}}{(\gamma-2+\beta)^{2}}
		.\]
	Note that the second term monotonically decreases with $f_{\gamma}(\beta)\coloneqq (\gamma-2-\beta)(\gamma-4+\beta)$, which achieves its maximum at $\beta = 1$.
	The third term monotonically increases with $\beta$.
	Thus, the optimal choice of $\beta$ is in $(0,1)$.

	Compared to the baseline:
	\[
		\|(I - \delta DG)^{-1} \mathbf{1}\|^{2} = \|(I - \delta H_{0}H_{0}'GH_{0}^{-1})^{-1}\mathbf{1}\|^{2} = \frac{n(\gamma-1)}{(\gamma-5)}\left( 1 - \delta \frac{(\gamma-2)^{2}}{(\gamma-2)(\gamma-4)} \right)^{\!\!\!-2}\! \frac{(\gamma-4)^{2}}{(\gamma-2)^{2}}
		,\]
	we have
	\[
		\frac{\sigma^{\mathrm{vr}}_{\X}}{\sigma^{\mathrm{bl}}_{\X}} = \left(\frac{ \frac{\gamma-4}{\delta(\gamma-2)} - 1 }{\frac{\gamma-4}{\delta(\gamma-2)} - \frac{1}{(1-\frac{\beta}{\gamma-2})(1+\frac{\beta}{\gamma-4})}} \cdot \frac{1+\frac{\beta}{\gamma-4}}{1+\frac{\beta}{\gamma-2}} \right)^{2}
		.\]
	We see that in the critical regime where $1 - \delta \frac{\gamma-2}{\gamma-4} \downarrow 0$, $\sigma^{\mathrm{bl}}_\X$ diverges while $\sigma^{\mathrm{\mathrm{vr}}}_{\X}$ is upper bounded by a constant, leading to $\sigma^{\mathrm{vr}}_{\X} / \sigma^{\mathrm{bl}}_\X \downarrow 0$.
\end{proof}

\section{Proofs of economic efficiency}\label{apx:eff-econ}

\subsection{Groves payment}\label{sec:groves}

For a large-scale system, given the small query complexity in \cref{thm:query}, a queried agent should not expect to be queried again.
Thus, consistent with memoryless agents, we consider \textit{myopic} agents who only care about maximizing their immediate utility in the current iteration when they are queried.

Another feature of our method is that its asynchronous nature greatly simplifies an agent's strategic considerations.
When queried, an agent is presented with the current price $p_{t}$ and can influence the updated price $p_{t+1}$. Previous agents' reports are settled in $p_{t}$, unchangeable by the current agent, and a myopic agent is also not strategic against future agents.
Therefore, the queried agent only plays a one-shot game against the planner.

A fundamental aspect of any asynchronous distributed mechanism for \cref{eq:planner} is that it's \emph{incentive compatible}, in the sense that at each iteration, the planner is essentially maximizing the ``social welfare'' of the queried agent, which is its utility:
\[
	\ts p_{t+1} = \argmax_{p} u_{i}(p) - \mathrm{reg}(p; p_{t},\alpha, \mu)
	,\]
where $u_{i}(p)\coloneqq u_{i}(s^*(p),p)$ and the regularizer ensures the update is not too aggressive and takes into account the planner's considerations.
In other words, the update direction (linear system residual) calculated at each iteration is the steepest ascent direction for the queried agent's regularized utility.
Intuitively, the fundamental idea of an asynchronous distributed mechanism is to use the queried agent's utility as a proxy for the social welfare, making the planner's and the agent's objectives aligned at each iteration.

Formally, let $\Phi_{i}$ be the utility kernel of agent $i$ derived in \cref{prop:ce} and $\tilde{\Phi}_{i}\coloneqq \Psi^{-1 /2}\Phi_{i}\Psi^{-1 /2}$.
Following \cref{eq:cen-2}, we have $\tilde{\Phi}_{i} = \tilde{\cen}_{i}^{T} \tilde{W}_{i}^{-1} \tilde{\cen}_{i}$, where $\tilde{W}_{i}^{-1} = \Psi^{-1 /2}W_{i}^{-1}\Psi^{-1 /2}$ is the planner-informed preference.
Notably, we have established that
\[
	\mathbb{E}[\hat{\Phi}_{t}\given i] = \tilde{\cen}_{i}^{T} \tilde{W}_{i}^{-1} \tilde{\cen}_{i} = \tilde{\Phi}_{i}
	,\]
i.e., $\hat{\Phi}_{t}$ is an unbiased estimate of $\tilde{\Phi}_{i}$,
implying that the planner optimizes their objective using the queried agent's utility as a proxy at each iteration, and thus \poll is incentive compatible for the agent.

However, elimination of any incentive to misreport requires the agent to accept the regularizer, i.e., internalize the planner's considerations in its utility.
Otherwise, for example in our setting where agents have a utility convex in the price, the queried agent can scale up the price update by an arbitrary factor, leading to an unbounded unregularized utility.

This internalization of the planner's considerations can be achieved through a \emph{payment}.

\begin{proposition}\label{prop:groves}
	Suppose the queried agent has full control over the price update $p_{t+1}-p_{t}$, i.e., it has a strategy space $\mathbb{R}^{m}$,
	and for a submitted updated price $p$, the planner collects a payment of
	\[
		v(p) = \frac{\mu}{2}\left(\langle p-p_0, \Psi (p-p_0)\rangle + \alpha_t^{-1}\left< p-p_{t},\Psi(p-p_t) \right>\right)
		.\]
	Then, maximizing the total payoff gives
	\[
		\argmax_{p} \{u_{i}(p)-v(p)\} = p_{t} - \alpha_{t}( (I-\nu \tilde{\Phi}_{i})p_t - p_{0} ) + O(\alpha_{t}^{2})
		,\]
	where $u_{i}(p) \coloneqq u_{i}(s^*(p),p)$ with $s^*$ specified in \cref{eq:ne} is the equilibrium utility of agent $i$, and $\tilde{\Phi}_{i}$ is its utility kernel such that $u_{i}(p) = \frac{1}{2}\langle p,\tilde{\Phi}_{i}p \rangle$.
\end{proposition}

Since the agent does not know $\tilde{\Phi}_{i}$ but only has a stochastic estimate $\hat{\Phi}_{i}$,
the above proposition implies that the agent's best response when facing a payment is to report $\hat{\Phi}_{i}$ according to \poll almost truthfully, i.e., with at most a higher-order deviation tolerated by our method.

Our payment is \textit{Groves}-like \citep{groves1973Incentivesteams}, as it does not depend on the reported preference of the queried agent but only on other agents' reports and the planner's objective.
This connection suggests the strategyproofness of \poll with payment.

We also share the same \textbf{information structure} as incentive design, where the planner first announces a strategy, the agent takes an action, and then the planner \textbf{observes} the action and takes its own action (collecting payment).
In contrast, in a classic Stackelberg game, the leader does not observe the follower's action before taking its own action.
Note that the price update rule is the planner's strategy but the  update itself can be seen as the agent's action, as it is determined by the agent's report and is not modified by the planner.

\begin{proof}[Proof of \cref{prop:groves}]
	Under payment, the total payoff of the queried agent is
	\[
		u_{i}^{\mathrm{G}}(p)\coloneqq u_{i}(p) - v(p) = \frac{1}{2} \langle p, \Phi_{i}p\rangle - \frac{\mu}{2}\langle p-p_0, \Psi (p-p_0)\rangle - \frac{\mu}{2\alpha _{t}} \left< p-p_t, \Psi(p-p_t) \right>
		,\]
	whose first-order condition is
	\[
		\frac{\mathrm{d} u^{\mathrm{G}}_{i}(p) }{\mathrm{d} p} = \Phi_{i}p - \mu \Psi(p-p_0) - \frac{\mu}{\alpha _{t}} \Psi(p-p_t) = 0
		.\]
	For a sufficiently small step size $\alpha_{t}$, the payoff is concave and achieves its maximum at
	\[
		p^\mathrm{G}_{*} =& (\mu(1+\alpha _{t}^{-1})\Psi - \Phi_{i})^{-1} \mu\Psi(p_0 + \alpha _{t}^{-1} p_t) \\
		=& \left(I - \nu \frac{\alpha _t}{1+\alpha _t}\Psi^{-1} \Phi_{i}\right)^{-1} \frac{\alpha _tp_0 + p_t}{1+\alpha _t}  \\
		=& \left(I + \nu \frac{\alpha _t}{1+\alpha _t}\Psi^{-1} \Phi_{i}\right) \frac{p_t + \alpha _tp_0}{1+\alpha _t} + O(\alpha _t^{2})  \\
		=& p_t - \frac{\alpha _t}{1+\alpha _t}\left( I - \frac{\nu}{1+\alpha _t}\Psi^{-1} \Phi_{i} \right)p_{t} + \frac{\alpha _t}{1+\alpha _t}\left( I + \nu \frac{\alpha _t}{1+\alpha _t}\Psi^{-1} \Phi_{i} \right)p_0 + O(\alpha _t^{2})\\
		=& p_t - \alpha _t\left( I - \nu\Psi^{-1} \Phi_{i} \right)p_{t} + \alpha _t p_0 + O(\alpha _t^{2})\\
		=& p_t - \alpha _t\left( \left( I - \nu\Psi^{-1} \Phi_{i} \right)p_{t} - p_0\right) + O(\alpha _t^{2})
		,\]
	where the third and fifth equalities use $(I + \alpha _{t}C)^{-1} = I - \alpha _tC + O(\alpha _t^{2})$ for any constant $C$.
	Renormalizing all prices by $p \gets \Psi^{1 /2}p$ gives the desired result.
\end{proof}

\subsection{Knightian uncertainty}\label{apx:knight}

\begin{lemma}[Regret]\label{lem:regret}
	Let $u_{i}(\hat{W}_{i})$ be the equilibrium utility of agent $i$ at the price updated by \poll with control variate and the reported preference $\hat{W}_{i}$. The expected loss (or negative gain) w.r.t. the agent's posterior belief $\mathcal{P}$ by deviating from truthful reporting is
	\[
		R_{\mathcal{P}}(\hat{W}_{i}) \coloneqq \mathbb{E}_{\mathcal{P}}[u_{i}(W_{i})-u_{i}(\hat{W}_{i})]
		= -\alpha _{t} \tr\bigl( C^{T} \bigl(I-W_{i}^{-1}\hat{W}_{i} \bigr)\bigl(I  - c_{\mathcal{P}}W_{i}^{-1} \hat{W}_{i} \bigr) \hat{W}_{i}^{-1}\bigr) + O(\alpha _{t}^{2})
		,\]
	where $C=\nu \Phi_{i}\Psi^{-1 /2}p_t p_t^{T} \Psi^{-1 /2}$ with $\Phi_{i}$ being the utility kernel of agent $i$, and $c_{\mathcal{P}} = \delta^{2} / (1-\delta)^{2} \mathbb{E}_{\mathcal{P}}[d_1d_2]$ are constants w.r.t. $\hat{W}_{i}$.
\end{lemma}

\begin{proof}
	Recall the update rule of \poll with control variate \cref{eq:cv}:
	\[
		p(\hat{W}_{i}) =& p_{t} - \alpha_t( (I-\nu \hat{\Phi}_{i})p_t - p_{0} )\\
		=& \alpha_t\nu \Psi^{-1 /2}\left(\hat{W}_{i}^{-1} +  \delta^{2}/(1-\delta)^{2}d_1d_2 W_{j^{(1)}}^{-1} \hat{W}_{i} W_{j^{(2)}}^{-1} \right)\Psi^{-1 /2}p_t \\
		& + \underbrace{(1-\alpha _t)p_{t} + \alpha _t p_0 + \alpha_t\nu \delta/(1-\delta) \Psi^{-1 /2}  \left( d_1 W_{j^{(1)}}^{-1} + d_2 W_{j^{(2)} }^{-1} \right) \Psi^{-1 /2}p_t}_{c}
		.\]
	$c$ is a constant vector with respect to $\hat{W}_{i}$.
	$W_{j^{(1)}}^{-1},W_{j^{(2)}}^{-1}$ are the preferences reported by the endpoint agents of the two random walks, which are unknown to agent $i$ and are not necessarily truthful.
	$d_1,d_2$ are the degree products of the two random walks, which are also unknown to agent $i$ but signaled by the planner.

	The actual utility of agent $i$ depends on its true preference $W_{i}$ and the updated price determined by $\hat{W}_{i}$.
	Let $\mathcal{P}$ be the agent's posterior belief over the unknown variables after receiving the planner's signals.
	We have
	\[
		\mathbb{E}_{\mathcal{P}} u_{i}(\hat{W}_{i})
		=& \mathbb{E}_{\mathcal{P}} \frac{1}{2} \bigl\| p(\hat{W}_{i}) \bigr\|_{\Phi_{i}}^{2} \\
		=& \mathbb{E}_{\mathcal{P}} \frac{1}{2} \bigl\| c + \alpha _{t}\nu \Psi^{-1 /2}(\hat{W}_{i}^{-1} + \delta^{2} /(1-\delta)^{2} d_1d_2 W_{j^{(1)}}^{-1} \hat{W}_{i} W_{j^{(2)}}^{-1} )\Psi^{-1 /2}p_{t} \bigr\|_{\Phi_{i}}^{2} \\
		=& \alpha _{t}\tr\bigl( (\underbrace{\nu \Phi_{i}\Psi^{-1 /2}p_t p_{t}^{T} \Psi^{-1 /2}}_{C})^{T}   \bigl(\underbrace{\hat{W}_{i}^{-1} + \delta^{2} /(1-\delta)^{2} \mathbb{E}_{\mathcal{P}}\bigl[d_1d_2 W_{j^{(1)}}^{-1} \hat{W}_{i} W_{j^{(2)}}^{-1} \bigr]}_{\mathbb{E}_{\mathcal{P}}\hat{\Phi}_{i}}\bigr) \bigr) \\
		&+ \frac{1}{2}\|c\|_{\Phi_{i}}^{2} + O(\alpha _{t})^{2} \label{eq:eu}
		.\]
		Suppose in the posterior belief $W_{j^{(1)}}^{-1}$ and $W_{j^{(2)}}^{-1}$ are independent and agent $i$ uses its own preference as the posterior mean.
	Thus,
	\[
		\mathbb{E}_{\mathcal{P}}u_{i}(\hat{W}_{i})
		= \alpha _{t} \tr(C^{T}(\hat{W}_{i}^{-1} + \underbrace{\delta^{2} / (1-\delta)^{2} \mathbb{E}_{\mathcal{P}}[d_1d_2]}_{c_{\mathcal{P}}} W_{i}^{-1} \hat{W}_{i} W_{i}^{-1}
		)) + \frac{1}{2}\|c\|_{\Phi_{i}}^{2} + O(\alpha _{t}^{2})
		.\]
	Then, the expected utility loss by deviating from truthful reporting is
	\[
		\mathbb{E}_{\mathcal{P}}[u_{i}(W_{i}) - u_{i}(\hat{W}_{i})]
		=& \alpha _{t} \tr\bigl( C^{T} \bigl( {W}_{i}^{-1}-\hat{W}_{i}^{-1} + c_{\mathcal{P}} W_{i}^{-1} \bigl( {W}_{i} - \hat{W}_{i} \bigr) W_{i}^{-1} \bigr) \bigr) + O(\alpha _{t}^{2}) \\
		=& -\alpha _{t} \tr\bigl( C^{T} \bigl(I-W_{i}^{-1}\hat{W}_{i} \bigr)\bigl(I  - c_{\mathcal{P}}W_{i}^{-1} \hat{W}_{i} \bigr) \hat{W}_{i}^{-1}\bigr) + O(\alpha _{t}^{2})
		.\]
\end{proof}

\begin{proof}[Proof of \cref{prop:knight}]
	For a radial misreport $\hat{W}_{i} = \rho W_{i}$, \cref{lem:regret} gives
	\[
		R_{\mathcal{P}}(\hat{W}_{i})
		= \alpha _t (1-\rho) (c_{\mathcal{P}} - \rho^{-1})\tr(C^{T}W_{i}^{-1}) + O(\alpha _t^{2})
		.\]
	Note that $C^{T}W_{i}^{-1}$ is of rank-one but non-zero, so $\tr(C^{T}W_{i}^{-1})
		\ne 0$. Without loss of generality, suppose $\tr(C^{T}W_{i}^{-1}) > 0$.
	Then, when the agent \emph{overstates} its preference, i.e., $\rho > 1$, the planner picks $c_{\mathcal{P}} < \rho^{-1}$;
	when the agent \emph{understates} its preference, i.e., $\rho < 1$, the planner picks $c_{\mathcal{P}} > \rho^{-1}$, leading to a positive regret in both cases when $\alpha _t$ is sufficiently small.

	Specifically, let $\ubar{\rho} = \inf\{\rho: \rho W_{i} \succeq \ubar{\lambda}I\}$ and $\bar{\rho} = \sup\{\rho: \rho W_{i} \preceq \bar{\lambda}I\}$. Then $\rho\in[\ubar{\rho},\bar{\rho}]$ and we know that $\ubar{\rho} \ge \ubar{\lambda} / \bar{\lambda} = \kappa_{W}^{-1}$ and $\bar{\rho} \le \bar{\lambda} / \ubar{\lambda} = \kappa_{W}$.
	The planner's signals ensure that the agent has posterior beliefs $\mathcal{P}_{-}$ and $\mathcal{P}_{+}$ such that $c_{\mathcal{P}_{-}} = \delta^{2} /(1-\delta)^{2} d_{-}^{2} < \kappa_{W}^{-1} \le \bar{\rho}^{-1}$ and $c_{\mathcal{P}_{+}} = \delta^{2} /(1-\delta)^{2} d_{+}^{2} > \kappa_{W} \ge \ubar{\rho}^{-1}$.
	Then, the worst-case regret is always non-negative, and almost truthful reporting is the only strategy that achieves a zero worst-case regret, or formally:
	\[
		\argmin_{\rho} \max_{\mathcal{P}\in\{\mathcal{P}_{-},\mathcal{P}_{+}\}} R_{\mathcal{P}}(\rho W_{i}) = 1 + O(\alpha _t)
		.\]
	In other words, the planner induces the agent to believe the true update might be an extreme case where any deviation from truthful reporting leads to a positive regret, thus driving the agent to report truthfully.

	Note that the slope of $\mathcal{R}_{\mathcal{P}}$ at $\rho=1$ is proportional to $c_{\mathcal{P}} - 1$, but it has a negative Hessian independent of $c_{\mathcal{P}}$.
	Therefore, the uncertainty structure needs to meet two requirements:
	it needs to change the sign of the slope and also make it steep enough to prevent the negative curvature from bending the regret below zero at larger deviations.
	\cref{fig:sb} illustrates the intuition.
\end{proof}

\begin{remark}[General manipulation space]\label{rmk:general-space}
	When the agent's strategy space is relaxed to the general $m$-dimensional space of all possible preference matrices, \emph{coaltion} between activities occurs (see \cref{rmk:coalition}).
	Specifically, since the decoy signals only kills one degree of freedom, the agent can coordinate its activities to misreport in an orthogonal hyperplane that is not penalized by the decoy signals.
	An illustration is provided in \cref{fig:sb-3d}, generalizing the 2D case in \cref{fig:sb}.

	To break this coalition, the planner can query one activity at a time.
	Consider the simplification in \cref{sec:voting} where agents share the same activity kernel eigenbasis $U$. Note that
	\[
		\Phi_{i} = \sum_{a\in [m]} \phi_{i,a} u_{a}u_{a}^{T} = \mathbb{E}_{a\sim \mathrm{Unif}([m])} \phi_{i,a}u_{a}u_{a}^{T}
	.\]
	Thus, the planner queries one activity $a$ at a time, which provides an unbiased estimate of $\Phi_{i}$.
	Since now the manipulation space of $\phi_{i,a}$ reduces to the one-dimensional case, \cref{prop:knight} applies.
\end{remark}

\begin{figure}[th]
	\centering
	\begin{tikzpicture}[scale=0.45]
		\begin{axis}[
				view={148}{26},
				axis lines=center,
				axis line style={thick, ->},
				xlabel={$\ln\rho_2$},
				ylabel={$\ln\rho_1$},
				zlabel={},
				xlabel style={font=\large, anchor=north east},
				ylabel style={font=\large, anchor=west},
				xmin=-1.4, xmax=1.4,
				ymin=-1.4, ymax=1.4,
				zmin=-1.5, zmax=1.7,
																xtick=\empty, ytick=\empty, ztick=\empty,
				width=13cm, height=13cm,
				clip=false,
				tick style={draw=none},
				declare function={
						ux(\x) = 0.35*\x + 0.5 * \x^2;
					},
				declare function={
						u1(\x,\y) = ux(\x) + 0.5 * \y^2 + 1*\y;
					},
				declare function={
						u2(\x,\y) = ux(\x) + 0.5 * \y^2 - 1*\y;
					}
			]

						\addplot3[
				surf, shader=flat,
				fill=mygreen!50, draw=mygreen!70,
				fill opacity=0.60, draw opacity=0.30,
				domain=-1:1, domain y=-1:0, samples=12,
			]{u1(x,y)};
						\addplot3[thick, color=black, domain=-1:1, samples y=0, samples=12] (x,-1,{u1(x,-1)});
						\addplot3[thick, color=black, domain=-1:0, samples=12] (-1,y,{u1(-1,y)});
						\addplot3[thick, color=black, domain=-1:0, samples=12] (1,y,{u1(1,y)});

						\addplot3[
				surf, shader=flat,
				fill=mygreen!50, draw=mygreen!70,
				fill opacity=0.60, draw opacity=0.30,
				domain=-1:1, domain y=0:1, samples=12,
			]{u2(x,y)};
						\addplot3[
				surf, shader=flat,
												fill=myviolet!50, draw=myviolet!70,
				fill opacity=0.60, draw opacity=0.30,
				domain=-1:1, domain y=-1:0, samples=12,
			]{u2(x,y)};
						\addplot3[thick, color=black, domain=-1:1, samples y=0, samples=12] (x,1,{u2(x,1)});
						\addplot3[thick, color=black, domain=-1:1, samples y=0, samples=12] (x,-1,{u2(x,-1)});
						\addplot3[thick, color=black, domain=-1:1, samples=12] (-1,y,{u2(-1,y)});
						\addplot3[thick, color=black, domain=-1:1, samples=12] (1,y,{u2(1,y)});

						\addplot3[
				surf, shader=flat,
												fill=myviolet!50, draw=myviolet!70,
				fill opacity=0.60, draw opacity=0.30,
				domain=-1:1, domain y=0:1, samples=12,
			]{u1(x,y)};
						\addplot3[thick, color=black, domain=-1:1, samples y=0, samples=12] (x,1,{u1(x,1)});
						\addplot3[thick, color=black, domain=0:1, samples=12] (-1,y,{u1(-1,y)});
						\addplot3[thick, color=black, domain=0:1, samples=12] (1,y,{u1(1,y)});

															\addplot3[dashed, color=gray!55, thick, domain=-1:1, samples=2]
			({-1},{y},{0});
			\addplot3[dashed, color=gray!55, thick, domain=-1:1, samples=2]
			({1},{y},{0});
			\addplot3[dashed, color=gray!55, thick, domain=-1:1, samples=2]
			({x},{-1},{0});
			\addplot3[dashed, color=gray!55, thick, domain=-1:1, samples=2]
			({x},{1},{0});

															\addplot3[
				domain=-1:1, samples y=0, samples=12,
				color=mygreen!80!black, line width=2pt,
			]({x},{0},{ux(x)});

						\addplot3[only marks, mark=o, mark size=4pt, color=mygreen!80!black, mark options={draw=mygreen!80!black, fill=white, line width=1pt}]
			coordinates {(0,0,0)};
			\node[font=\large, right] at (axis cs:0.08, 0.00, -0.22) {$0$};

						\addplot3[dashed, color=gray!55, line width=1pt,
				domain=0:1, samples=2]({1},{0},{min(x,ux(1))});
						
						\addplot3[only marks, mark=*, mark size=4pt, color=mygreen!80!black]
			({1}, {0}, {ux(1)});

						\node[font=\large, anchor=south west]
			at (axis cs:0.05, 0.05, 1.70) {$\hat{u}-u$};

						\node[font=\large]
			at (axis cs: -0.10, 2.10, 2.4) {$d_+$};
			\node[font=\large]
			at (axis cs:-0.95,-0.6, 1.6) {$d_-$};

		\end{axis}
	\end{tikzpicture}
	\caption{Although an ambiguity averse agent would avoid misreporting along one dimension, it can coordinate its activities to misreport along the other direction (highlighted in green).}\label{fig:sb-3d}
\end{figure}
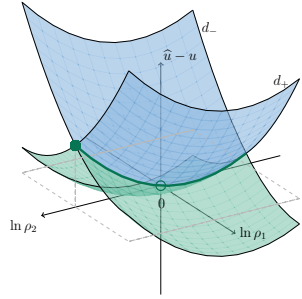

\begin{remark}[Regret]
	We benchmark regret against truthful reporting, as truth minimizes the agent's cognitive load \citep{li2017ObviouslyStrategyProof}. Using the optimal misreport as the benchmark, which leads to the standard minimax regret in online learning, would require both the agent to compute the best misreport for each signal and the planner to craft signals with approximately equal optimal utility gains.
\end{remark}

\subsubsection{Beyond full ambiguity aversion}
We extend \cref{prop:knight} to consider a continuum of ambiguity attitudes, from full ambiguity aversion to full optimism.

\begin{definition}[$\theta$-maxmin preference; \citet{marinacci2002Probabilisticsophistication,ghirardato2004Differentiatingambiguity}]\label{def:theta}
	Given an uncertainty set $\mathcal{C}$ and an ambiguity attitude $\theta\in[0,1]$, a \emph{$\theta$-maxmin} agent  solves
	\[
		\max_{\hat{W}_{i}} \Bigl( \theta \min_{\mathcal{P}\in\mathcal{C}} \mathbb{E}_{\mathcal{P}}[u_{i}(\hat{W}_{i})-u_{i}(W_{i})] + (1-\theta) \max_{\mathcal{P}\in\mathcal{C}} \mathbb{E}_{\mathcal{P}}[u_{i}(\hat{W}_{i})-u_{i}(W_{i})] \Bigr)
		.\]
\end{definition}

The case $\theta=1$ recovers the maxmin preference in \cref{prop:knight}, and $\theta=0$ corresponds to a fully optimistic agent that selects the best case.

\begin{proposition}\label{prop:theta}
	Consider the setup of \cref{prop:knight} with the maxmin preference replaced by the $\theta$-maxmin preference, and let $\mathcal{C}=\{\mathcal{P}_{-},\mathcal{P}_{+}\}$ be induced by signals $d_{-}<d_{+}$.
	If
	\[
		\theta > \theta_{*}(d_{-},d_{+}) \coloneqq \max\left\{ \frac{d_{+}^{2} - \kappa_{W}^{-1}(1-\delta)^{2}/\delta^{2}}{d_{+}^{2}-d_{-}^{2}},\; \frac{- d_{-}^{2} + \kappa_{W}(1-\delta)^{2}/\delta^{2} }{d_{+}^{2}-d_{-}^{2}} \right\},
	\]
	the agent's optimal strategy deviates from the truth $W_{i}$ by at most $O(\alpha_t)$.
\end{proposition}

\begin{proof}
	Following the proof of \cref{prop:knight}, up to $O(\alpha_t^{2})$, the expected utility gain of the report $\rho W_{i}$ increases in $c_{\mathcal{P}_\pm}$ if $\rho>1$ and decreasing if $\rho<1$.
	Thus, when $\rho>1$, the $\theta$-maxmin objective is proportional to $(\rho-1)(\theta c_{\mathcal{P}_-} + (1-\theta)c_{\mathcal{P}_+} - \rho^{-1})$, which is strictly negative for all $\rho\in(1,\bar{\rho}]$ if $\theta c_{\mathcal{P}_-}+(1-\theta)c_{\mathcal{P}_+} < \kappa_{W}^{-1}$;
	for $\rho<1$, it is proportional to $(\rho-1)(\theta c_{\mathcal{P}_+}+(1-\theta)c_{\mathcal{P}_-}-\rho^{-1})$, which is strictly negative for all $\rho\in[\ubar{\rho},1)$ if $\theta c_{+}+(1-\theta)c_{-} > \kappa_{W}$.
	Rearranging the two conditions gives $\theta>\theta_{*}(d_{-},d_{+})$, under which truthful reporting uniquely attains the maximal (zero) objective up to $O(\alpha_t)$.
\end{proof}

At $\theta=1$, $\theta_{*}<1$ recovers the conditions on $d_{\pm}$ in \cref{prop:knight}.
Larger separations between $d_{-}$ and $d_{+}$ do not necessarily lower $\theta_{*}$: a stronger $d_{+}$ deters understating but also amplifies the best case of overstating, and vice versa.

\subsection{Byzantine audit}

The above strategyproof mechanisms work for rational agents but are not necessarily robust to Byzantine agents, who may misreport an arbitrary preference for malicious purposes beyond their self-interest.

\poll already possesses a natural device to safeguard against Byzantine agents, by forming an \emph{audit committee} consisting of the agents at the endpoints of the random walks at each iteration.
Suppose agents observe or communicate each other's equilibrium efforts at the initial price, and try to form an estimate of each other's preferences.
Then, an endpoint agent can check this estimate against the reported preference of the queried agent, and submit an audit report.
The planner rejects the vote if many audit committee members report a large mismatch.

Note that an agent's equilibrium effort is influenced by its own preference and the preferences of all other agents it is connected to.
Therefore, without knowing other agents' preferences, an agent's estimate of another agent's preference is highly uncertain.
This motivates the use of a batch of random walks at each iteration, which not only reduces the variance of \poll but also forms a more credible audit committee.

\begin{proposition}\label{prop:byzantine}
	Suppose agents have the same set of activity bundles, the initial price is non-zero for all principal activity bundles, i.e., $U^{T}p_{0}$ has no zero entry.
	Suppose agents' equilibrium efforts at the initial price $p_{0}$ are public information.
	Let $i$ be the queried agent and $j$ be an agent at the end of a random walk in \poll. Let $d_{j}$ be the degree product of the corresponding random walk.
	Suppose no collusion. After $j$ sees $i$'s reported preference $\hat{W}_{i}$, it submits an audit report $\chi_{j} = \hat{W}_{i}^{-1}p_0 - s^{*}_{i}(p_0) + \delta(1-\delta)^{-1} d_{j}W_{j}^{-1}p_0$, which is an unbiased estimate of mismatch $\hat{W}_{i}^{-1}p_0 - W_{i}^{-1}p_0$.
	Given a batch of $B$ random walks at each iteration, we reject the vote if the audit committee gives an aggregated report $\|\frac{1}{B}\sum_{j=1}^{B}\chi_{j}\| > \epsilon_{B}$, where $\epsilon_{B} = \frac{2\delta\bar{d}\|p_0\|}{(1-\delta)^{2}\ubar{\lambda}} \sqrt{\frac{\ln (2m/\iota)}{B}}$.
	Then, with probability at least $1-\iota$, any accepted report $\hat{W}_{i}$ has a bounded mismatch $\|\hat{W}_{i}^{-1} - W_{i}^{-1}\| \le 2\ubar{\tilde{p}_{0}}^{-1}\epsilon_{B}$ and $\|\hat{W}_{i}-W_{i}\| \le 2\ubar{\tilde{p}_{0}}^{-1} \bar{\lambda}^{2}\epsilon_{B}$, where $\ubar{\tilde{p}_{0}}\coloneqq \min_{a\in\A}[U^{T}p_{0}]_{a}$.
\end{proposition}

To completely hold Byzantine agents in check, i.e., guarantee a mismatch of only higher order, \poll needs a large batch size, which is inevitable as we adopt a sampling approach.
A simple deterministic approach is using neighbors as the auditing committee, rejecting a vote if $\hat{W}_{i}\left( s^{*}_{i}(p_0)-\delta\sum_{j\in \X_{i}}s^{*}_{j}(p_0) \right) \ne p_0$, which guarantees that any accepted report is exactly truthful.
However, this approach
\begin{enumerate*}[label=(\alph*)]
	\item is not robust to any noise,
	\item makes it easier for agents to collude as the audit committee has fixed members consisting of the queried agent's nearest neighbors, and
	\item requires implementing a separate module rather than being directly enabled by \poll.
\end{enumerate*}

If we accept a non-vanishing Byzantine-induced error, our result is consistent with the current literature on Byzantine-robust federated learning \citep{yin2018ByzantineRobustDistributed}, which shows the error is of $\sqrt{\frac{b\sigma^{2}}{nB}}$, where $b /n$ is the fraction of Byzantine agents, $B$ is the number of local steps, and $\sigma^{2}$ is the variance of the stochastic gradients.
This is recovered by our result recognizing that multiple random walks are equivalent to multiple local steps, which also reduces the variance of the stochastic update.

\begin{proof}[Proof of \cref{prop:byzantine}]

	By \cref{eq:ne},
	\[
		s_{i}^{*}(p) =& [((I_{n}-\delta G)^{-1} \otimes I_m) \operatorname{blkdiag}(W_{i}^{-1})]_{i}p \\
		=& \sum_{k=0}^{\infty}\delta ^{k} \sum_{j\in\X} (G^{k})_{ij}W_{j}^{-1}p \\
		=& W_{i}^{-1} p + \delta(1-\delta)^{-1} \mathbb{E}_{k \sim \operatorname{Geom}(\delta)} \left[  \sum_{j\in\X} (G^{k+1})_{ij}W_{j}^{-1}p \right] \\
		=& W_{i}^{-1} p + \delta(1-\delta)^{-1} \underset{\substack{\text{random walk } (i, \dots,j_{k})\\ \text{ of length } k \sim \operatorname{Geom}(\delta)}}{\mathbb{E}} \left[ \prod_{l=0}^{k-1}\deg_{j_{l}}  W_{j_{k}}^{-1}p \right]
		.\]
	The expectation above is w.r.t. the exact same sampling process in \poll, giving a natural audit committee consisting of the agents at the endpoints of the random walks in \poll.
	Separating private $W_{i}$ and public information gives
	\[
		W_{i}^{-1}p = s^{*}_{i}(p) - \delta(1-\delta)^{-1} \mathbb{E}[d_{j} W_{j}^{-1}p]
		,\]
	where we denote $j$ as a generic endpoint agent of a random walk and $d_{j}$ as the degree product of the walk.
	Checking against the submitted report $\hat{W}_{i}$, an audit report from an endpoint agent $j$ is thus
	\[
		\chi_{j} = \hat{W}_{i}^{-1}p_{0} - s^{*}_{i}(p_{0}) + \delta(1-\delta)^{-1} d_j W_{j}^{-1}p_{0}
		,\]
	giving the average report
	\[
		\chi_{B}\coloneqq\frac{1}{B} \sum_{j=1}^{B} \chi_{j} = \hat{W}_{i}^{-1}p_{0} - s^{*}_{i}(p_{0}) + \delta(1-\delta)^{-1} \frac{1}{B}\sum_{j=1}^{B} d_j W_{j}^{-1}p_{0}
		.\]
	That is, we replace the expectation with the batch average.
	Note that $\|d_{j}W_{j}^{-1}\| \le k_{j}\bar{d} /\ubar{\lambda}$ where $k_{j}\sim \operatorname{Geom}(\delta)$ has a sub-exponential norm of $O((1-\delta)^{-1})$.
	Thus, by Bernstein's inequality, for any $\epsilon_{B} \in (0, (1-\delta)^{-1}\bar{d} /\ubar{\lambda})$, if agent $i$ reports truthfully, we have
	\[
		& P\left( \delta(1-\delta)^{-1} \left\| \frac{1}{B}\sum_{j=1}^{B} d_j W_{j}^{-1}p_{0} - \mathbb{E}[d_j W_{j}^{-1}p_{0}] \right\| > \epsilon_{B} \right) \\
		\le& 2m\exp\bigl( -4^{-1} B\delta^{-2}(1-\delta)^{4} \ubar{\lambda}^{2}\|p_0\|^{-2} \bar{d}^{-2}  \epsilon_{B}^{2} \bigr)
		.\]

	Let
	\[
		\iota = 2m\exp\bigl( -4^{-1} B\delta^{-2} (1-\delta)^{4}\ubar{\lambda}^{2} \bar{d}^{-2}\|p_0\|^{-2}  \epsilon_{B}^{2} \bigr)
		\implies \epsilon_{B} = \left( \frac{2\delta\bar{d}\|p_0\|}{(1-\delta)^{2}\ubar{\lambda}} \sqrt{\frac{\ln (2m/\iota)}{B}} \right)
		.\]
	Then with probability at least $1-\iota$, we know the noise in the average audit report is at most $\epsilon_{B}$.
	Thus, the planner rejects $\hat{W}_{i}$ if $\|\chi_{B}\| > \epsilon_{B}$.
	Consequently, if a report $\hat{W}_{i}$ is accepted, with probability at least $1-\iota$, we have
	\[
		\|\hat{W}_{i}^{-1}p_0 - W_{i}^{-1}p_0\|
		=& \|\hat{W}_{i}^{-1}p_0 - s^{*}_{i}(p_0) + \delta(1-\delta)^{-1} \mathbb{E}[d_j W_{j}^{-1}p_{0}] \|\\
		\le& \left\|\hat{W}_{i}^{-1}p_0 - s^{*}_{i}(p_0) + \delta(1-\delta)^{-1} \frac{1}{B}\sum_{j=1}^{B}d_j W_{j}^{-1}p_{0} \right\| \\
		&+ \delta(1-\delta)^{-1}\left\|\frac{1}{B}\sum_{j=1}^{B}d_j W_{j}^{-1}p_{0} - \mathbb{E}[d_j W_{j}^{-1}p_{0}] \right\| \\
		\le& \|\chi_{B}\| + \epsilon_{B} \\
		\le& 2\epsilon_{B}
		.\]
	Under the shared activity kernel eigenbasis assumption, we have
	\[
		\|\hat{W}_{i}^{-1}-W_{i}^{-1}\| \le \frac{1}{\min_{a\in\A}|\tilde{p}_{0,a}|} \|\hat{W}_{i}^{-1}p_0 - W_{i}^{-1}p_0\| \le 2\ubar{\tilde{p}_{0}}^{-1} \epsilon_{B}
		,\]
	where $\tilde{p_0} = U^{T} p_0$ and $\ubar{\tilde{p}_{0}} \coloneqq \min_{a\in\A}|\tilde{p}_{0,a}|$ is assumed to be non-zero.
	And we have
	\[
		\|\hat{W}_{i}-W_{i}\|
		\le 2\ubar{\tilde{p}_{0}}^{-1} \bar{\lambda}^{2} \epsilon_{B}
		.\]
\end{proof}

\end{document}